\documentclass[12pt]{article}
\usepackage{coool}

\begin{document}

	\title{Determination of the effective cointegration rank in high-dimensional time-series  predictive regressions} 
	
	\author[1]{Puyi Fang}
	\author[2]{Zhaoxing Gao}
	\author[3]{Ruey S. Tsay*}
	\affil[1]{\small School of Economics, Zhejiang University}
	\affil[2]{\small Center for Data Science, Zhejiang University}
	\affil[3]{\small Booth School of Business, University of Chicago}
	
	\date{Current Version: \today}	
	
	\begin{spacing}{1.5}
		\begin{titlepage}
			\maketitle

			\begin{abstract} 
				
				This paper proposes a new approach to identifying the effective cointegration rank in high-dimensional unit-root (HDUR) time series from a prediction perspective using reduced-rank regression. For a HDUR process $\bx_t\in \bbR^N$ and a stationary series $\by_t\in \bbR^p$ of interest, our goal is to predict future values of $\by_t$ using $\bx_t$ and lagged values of $\by_t$. The proposed framework consists of a two-step estimation procedure. First, the Principal Component Analysis is used to identify all cointegrating vectors of $\bx_t$. Second, the co-integrated stationary series are used as regressors, together with some lagged variables of $\by_t$, to predict $\by_t$. The estimated reduced rank is then defined as the effective coitegration rank of $\bx_t$. Under the scenario that the autoregressive coefficient matrices are sparse (or of low-rank), we apply the Least Absolute Shrinkage and Selection Operator (or the reduced-rank techniques) to estimate the autoregressive coefficients when the dimension involved is high. Theoretical properties of the estimators are established under the assumptions that the dimensions $p$ and $N$ and the sample size $T \to \infty$. Both simulated and real examples are used to illustrate the proposed framework, and  the empirical application suggests that the proposed procedure fares well in predicting stock returns.
			\end{abstract}
			
			\vspace{2em}
			
			
			\noindent\textit{Keywords}: {Cointegration, Factor model,  Reduced rank, High dimension, LASSO }
		\end{titlepage}
		\setcounter{page}{2}

		\section{Introduction}
		\label{section:introduction}
		
		The availability of large-scale or vast time-series data in recent years brings new challenges and opportunities to time series modeling. Analysis of high-dimensional (HD) time series  has emerged as one of the important and active research areas in statistics, economics, finance, and engineering, among other scientific fields. For example, returns of a large number of assets form a HD time series and play an important role in asset pricing, portfolio allocation, and risk management. Environmental studies often employ HD time series consisting of a large number of pollution indexes collected from many monitoring stations over time. In many applications, data often exhibit characteristics of unit-root nonstationarity. For instance, the series of quarterly gross domestic 
		products, total exports, and total imports of an economy tend to contain unit roots. In theory, the vector autoregressive moving-average (VARMA) models can be used to analyze such data, but they often encounter the difficulties of cointegration testing, overparametrization, and lack of identifiability. 
		See, for example, \textcite{johansen2002small,tiao1989model,lutkepohl2006new,tsay2014multivariate}, and the references therein. To overcome these difficulties, dimension reduction or structural regularization  becomes a necessity, and various methods have been developed in the literature including 
		the regularized estimation method for HD VAR models in \textcite{lin2017regularized} and the factor modeling by \textcite{stock2005implications,bai2002determining,forni2005generalized,pena2006nonstationary,lam2011estimation,lam2012factor,gao2019structural,gao2021modeling_JASA,gao2021modeling_IJF,gao2022divide}, among others. However, most of the studies mentioned above focus on stationary processes and are not applicable to unit-root nonstationary series. The only exceptions are  \textcite{bai2004estimating,pena2006nonstationary,gao2021modeling_IJF}. On the other hand, the unit-root nonstationarity is commonly seen in many empirical applications and the complexity of the dynamical dependence in such data requires further investigation.

		It is well known that  
		cointegration is often used to account for common trends and to avoid non-invertibility induced by over-differencing unit-root time series. See \textcite{engle1987co,johansen1988statistical,johansen1991estimation,tsay2014multivariate}, and the references therein. In practice, the cointegration rank of a given vector time series is unknown, and many approaches have been proposed to estimate the rank; see, for example, \textcite{engle1987co,johansen1988statistical,johansen1991estimation,saikkonen2000testing,aznar2002selecting}. However, these methods are rarely applied to HD time series due to their poor finite-sample performance, as discussed in \textcite{johansen2002small}. Yet there are many real applications that involve HD time series. For example, \textcite{banerjee2014forecasting} emphasized the importance of testing for no cross-sectional cointegration in panel cointegration analysis, and the cross-sectional dimension of modern macroeconomic panel can easily be as large as several hundreds. Recently, there are some studies on identifying the cointegration rank of  unit-root time series from a factor modeling perspective. See \textcite{pena2006nonstationary} for the case of fixed dimensions and \textcite{bai2004estimating,zhang2019identifying,gao2021modeling_IJF} for HD time series. However, the situation changes in the case  
		of growing dimension because the estimated cointegration rank usually grows as the dimension increases and the cointegration relationships are often hard to interpret when there are many cointegrating vectors.
		
		This paper marks a further development in estimating the cointegration rank of 
		HDUR time series from a predictive perspective. To avoid employing a large number of cointegrating vectors given by a high-dimensional method, we estimate the {\it effective cointegration rank} in a predictive framework. Specifically, suppose our goal is to predict the future values of a HD stationary time series $\by\in \bbR^{p}$ using $\bx$ as predictors. It is well known that only the cointegrated series have potential predictive power for the stationary process $\by$. If the number of cointegrating vectors is large, 
		the stacked variables obtained by cointegrating vectors form a HD stationary time series, and can be used as potential predictors. 
		But not all cointegrated series have predictive power for $\by$, and 
		we define the  {\it effective cointegration rank} as the effective dimension of the stacked variables that have predictive power for $\by$. The resulting effective rank can 
		be much smaller than the cointegration rank of $\bx$.

		The proposed method consists of a two-step estimation procedure. First, we postulate that the HDUR time series follows a factor model as that specified in \textcite{bai2004estimating}, where the common factors capture the nonstationary common trends of all the components, and the idiosyncratic term is a stationary process. We apply the Principal Component Analysis (PCA) to estimate the common stochastic trends and their associated loading matrix, and the orthogonal complement of the loading matrix consists of the cointegrating vectors. 
		Second, we put together all stationary series obtained by the cointegrating vectors of 
		the first step to form a set of predictors, and perform a reduced-rank regression between the $\by_t$ series of interest and the predictors. To further explain the variability of the data, we also include some lagged variables of $\by_t$ in the regression and assume their coefficient matrices are of low-dimensional structures. We propose two procedures to estimate all the coefficient matrices depending on whether 
		the autoregressive (AR) matrices are sparse or of low-rank. When the AR coefficient matrices are sparse, we apply the nuclear norm penalty to the regression coefficient matrix 
		of the stationary predictors obtained from the first step, and the LASSO penalty to the coefficients of the lagged variables.  When both the AR matrices and the coefficient matrix of the predictors are of low-rank,  
		we propose an integrative reduced-rank approach to estimate all unknown parameters. Two iterative, alternating procedures are proposed to estimate all unknown coefficients under the two aforementioned scenarios. Theoretical properties of the estimators are established under the assumption that the dimensions $p$ and $N$ and the sample size 
		$T \to \infty$.  Both simulated  and real examples are used to illustrate the proposed 
		procedure. The empirical application suggests that the 13 macroeconomic variables from \textcite{welch2008comprehensive}  provide satisfactory performance as predictors in forecasting the returns of 79 stocks in the S\&P 500 index.

		The idea of using predictive regression to estimate the cointegrating vector can be found in, for example, \textcite{koo2020high}. However, the method of \textcite{koo2020high} only identifies one cointegrating vector in predicting another univariate time series, whereas the proposed method not only recovers the total cointegration rank, but also identifies the effective cointegration rank in predicting a large panel of time series. In addition, the proposed estimation method is different from theirs as we use PCA, reduced-rank, and LASSO techniques to achieve our goals while they focus mainly on the use of LASSO regularization.  Note also that our framework is established for data with time series dependence structure and we use a combination of reduced-rank and sparsity techniques in the estimation procedure, which is different from  most of the methods discussed in \textcite{reinsel2022multivariate} that focus on the reduced-rank techniques for {\it i.i.d.} observations. The only exception is the work of  \textcite{lin2017regularized} in studying regularized estimation for multi-block stationary VAR models using the tools and techniques developed in \textcite{negahban2011estimation} and \textcite{agarwal2012noisy}. Furthermore, none of the work mentioned above deals with HDUR time series data.

		This paper makes multiple contributions. First, the cointegration problem has been a central issue in modeling HDUR time series, but the lack of clear interpretations of 
		a large number of cointegrating vectors renders the existing HD methods less appealing. We define the effective cointegration rank  to select the most significant cointegration relationships from  a predictive point of view using reduced-rank method. This method often  produces a small number of significant cointegrating vectors which are easier to interpret in general. Second, our predictive regression model consists of both nonstationary and stationary variables as predictors and has a wide range of applications including the prediction of stock returns using macroeconomic series in finance and the prediction of PM$_{2.5}$ values using other air pollution  and meteorological indexes in environmental studies. Third, the proposed approach combines the advantages of using two regularization methods, reduced-rank and LASSO, to reduce the dimension of a large system, and 
		the asymptotic results derived suggest that properties of both methods continue to hold when they are used simultaneously in a regression model with serially dependent data. This is a theoretical contribution.

		This paper is organized as follows. We introduce the proposed model, estimation methodology, and the modeling procedure in Section~\ref{section:model}. Section~\ref{sec:theory} is devoted to theoretical properties of the proposed model and its associated estimates, and  Section~\ref{sec:simulations} presents some simulation 
		results to demonstrate the performance of the proposed method in finite samples. 
		In Section~\ref{sec:real_data_analysis}, we apply the proposed method to 
		the prediction of stock returns using some commonly used macroeconomic predictors. Section~\ref{sec:conclusion} provides some discussions and concluding remarks. All technical proofs of the theorems are relegated to an online supplement.
		
		\begin{notation}
			To begin, we summarize here the  notation used throughout the paper. The bold upper case, bold lower case, and lower case letters are used to denote matrices, vectors, and scalars, respectively. For a matrix $\mat{A}\in{\mathbb{R}^{m\times n}}$, we use $\FNorm{\mat{A}}$, $\norm{\mat{A}}_*$, and $\norm{\mat{A}}_2$ to denote its Frobenius, nuclear, and operator norms, that is, $\sqrt{\tr(\mat{A}'\mat{A})}$, the sum of singular values of $\mat{A}$, and the largest singular value of $\mat{A}$, respectively. $\bI_p$ denotes the $p\times p$ identity matrix. The superscript ${'}$ denotes the transpose of a vector or a matrix. For a matrix $\mat{A}=(\ba_1, \ba_2, \ldots, \ba_n)$, we use $\vectorize(\mat{A})$ to denote its vectorization, which is equal to $(\ba_1', \ba_2', \ldots, \ba_n')'$, and we further use $\norm{\vectorize(\mat{A})}_1=\sum_{i,j} \abs{a_{ij}}$ to denote the $l_1$-norm of $\mat{A} = [a_{ij}]$. Finally, for two matrices $\mat{A}$ and $\bB$ with commensurate dimensions, their inner product is defined as $\inprod{\bA,\bB}=\tr(\bA' \bB)$.  We also use the notation $a\asymp b$ to denote $a=O(b)$ and $b=O(a)$.
			Finally, we use $L(\cdot)$ to denote the lag operator, which can shift a scalar, vector or matrix time series back by one time period. For instance, for the matrix $\mat{Y} = (\by_1,\by_2,\ldots,\by_T)$,  $L(\mat{Y}) = (\by_0,\by_1,\ldots,\by_{T-1})$.
		\end{notation}
		
		\section{The Model and Methodology}
		\label{section:model}
		
		\subsection{Model Setting}
		Let $\by_t = (y_{1t},y_{2t},\cdots,y_{pt})'$ be an observable $p$-dimensional  stationary time series, and $\bx_t = (x_{1t},x_{2t},\cdots,x_{Nt})'$ an observable $N$-dimensional $I(1)$ process. We consider the following predictive regression model:
		\begin{equation}\label{VAR_d}
		\by_t = \mat{W} \bx_{t-1} + \mat{\Phi}_1\by_{t-1} + \mat{\Phi}_2\by_{t-2}+ \cdots + \mat{\Phi}_d\by_{t-d} + \vect{e}_t,\,\,t=1,...,T,
		\end{equation}
		where $\mat{W}$ is a $p\times N$ coefficient matrix associated with the $I(1)$ process $\bx_t$, and $\mat{\Phi}_i$ is the $p\times p$ coefficient matrix of $\by_{t-i}$, for $1\leq i\leq d$, and $\vect{e}_t\sim$ WN$(0,\bSigma_e)$ is a white noise error term with mean zero and a nonsingular covariance $\bSigma_e$. Our goal is to estimate $\bW$ and $\bPhi_i$ based on a given sample, and to forecast future values of $\by_t$. For simplicity, all variables are 
		set to zero if the time index is not positive. Also, Model (\ref{VAR_d}) can be extended to multi-step ahead predictions for $\by_{t+h}$ with $h > 0$.
		
		In Model (\ref{VAR_d}), $\bx_t$ is nonstationary but all other variables are stationary 
		so that it only makes sense if some variables in $\bx_t$ are cointegrated, otherwise, $\bW$ would essentially be a zero matrix because the correlation between a stationary process and a unit-root nonstationary one is zero in general. If we blindly apply the Least Squares (LS) method to estimate the model, the number of parameters to be estimated is large, and the resulting estimator $\wh \bW$ would be hard to interpret as we do not know whether all or only a few rows in $\wh\bW$ are the estimated cointegrating vectors. 
		In theory, if all the cointegrating vectors of $\bx_t$ are known, then 
		the resulting linear combinations of the $I(1)$ variables are stationary and can be useful predictors in Model (\ref{VAR_d}). However, not all cointegration relationships are helpful in predicting $\by_{t+h}$ in general, especially when the dimension $N$ of $\bx_t$ is large. 

		In view of the above discussion, we modify Model (\ref{VAR_d}) as follows. First, similarly to the setting in \textcite{bai2004estimating}, we assume that $\bx_t$ admits a latent factor structure:
		\begin{equation}
		\label{factor_structure}
		\bx_{t} = \bB \bff_{t} + \bm{\varepsilon}_t,\  t=1,2,\ldots,T,
		\end{equation}
		where $\bff_t=(f_{1t},f_{2t},\ldots,f_{rt})'$ is an $r$-dimensional factor process that constitutes the common stochastic trends of $\bx_t$, that is,
		\begin{equation}\label{f_process}
		\bff_{t}=\bff_{t-1} + \vect{u}_t,
		\end{equation}
		where $\vect{u}_t$ is an $r$-dimensional zero-mean stationary process that drives $\bff_t$. The idiosyncratic term $\bve_t$ in (\ref{factor_structure}) is assumed to be a stationary process independent of the common factors $\bff_t$. Therefore, the cointegration rank of $\bx_t$ is $N-r$. 
		For ease in model identification, 
		we assume that $\bB$ is an orthonormal matrix such that $\bB'\bB=\bI_r$; see also \textcite{bai2002determining,fan2013large} for details.

		Let $\bB_c\in\bbR^{N\times(N-r)}$ be an orthogonal complement matrix of $\bB$ such that $\bB_c'\bB_c=\bI_{N-r}$ and $\bB_c'\bB={\bf 0}$. It follows from Model (\ref{factor_structure}) that the columns of $\bB_c$ can be treated as a set of cointegrating vectors of $\bx_t$ because $\bB_c'\bx_t=\bB_c'\bve_t$ is stationary. Letting  
		\begin{equation}\label{zt}
		\bz_t = \bB_c' \bx_t = \bB_c' \bm{\varepsilon}_t,
		\end{equation}
		we define $\bW=\bA\bB_c'$ and rewrite Model (\ref{VAR_d}) as follows:
		\begin{equation}
		\label{DGP}
		\by_t = \mat{A} \bz_{t-1} + \mat{\Phi}_1\by_{t-1} + \mat{\Phi}_2\by_{t-2}+ \cdots + \mat{\Phi}_d\by_{t-d} + \vect{e}_t,\,\,t=1,...,T,
		\end{equation}
		where $\bz_t$ is now a stationary process defined in (\ref{zt}). Similarly to the identifiability issue in factor models, $\bA$ and $\bB_c$ are not uniquely defined. Nonetheless, the product, $\bW=\bA\bB_c'$, is uniquely defined. Therefore, we split Model (\ref{VAR_d}) into (\ref{factor_structure})--(\ref{DGP}), and our goal is to estimate the factor loading matrix $\bB$ or equivalently the cointegrating vector matrix $\bB_c$, the coefficient matrices $\bA$ and $\bPhi_i$, for $1\leq i\leq d$.
		
		Although $\bB$, $\bB_c$ and $\bA$ are not uniquely defined due to the identification issue, the linear spaces spanned by the columns of $\bB$ and $\bB_c$, denoted as $\mathcal{M}(\bB)$ and $\mathcal{M}(\bB_c)$ respectively, are uniquely defined. For any specific choice of $\bB_c$, $\bA$ can also be uniquely determined. Therefore, when we mention the estimation or consistency of the loading matrix $\bB$ or $\bB_c$ in the sequel, we always refer to their column spaces to avoid any confusion. The estimation of $\bA$ is also based on a given and fixed $\bB_c$ so that the procedure is valid.

		\subsection{Estimation Methodology} \label{estimation_methodology}
		We consider two approaches to estimating the effective cointegration rank, or equivalently, the reduced-rank of the coefficient matrix $\bA$, and the AR  coefficients $\bPhi_i$'s for high-dimensional cases under different assumptions. The first approach is based on imposing a reduced-rank structure on the matrix $\bA$ and some sparsity assumptions on the AR coefficient 
		matrices. The second approach requires that all predictors, including the lagged variables, have  
		their own low-rank coefficient matrices. 
		
		\subsubsection{A Reduced-Rank and Sparse Regression Approach}\label{RRS}
		In this section, we introduce a Reduced-Rank and Sparse Regression approach (RRSRA) to estimating the coefficient matrices $\bA$ and $\bPhi_i$ for observed data  $\{\bx_1,...,\bx_T\}$ and $\{\by_1,...,\by_T\}$. Note that the dimensions of $\bA\in \bbR^{p\times (N-r)}$ and $\bPhi_i\in\bbR^{p\times p}$ can be very large under the assumption that the number of common stochastic trends $r$ is finite as  $p,N\rightarrow\infty$.
		Even if $\{\bz_1,...,\bz_T\}$ were given, the traditional LS method would lead to overfitting because there are many parameters to estimate. Therefore, some structure regularization must be imposed on the  coefficient matrices. For simplicity, we assume the matrix $\bA$ is singular and has a reduced-rank form with  $r_{\bA}=\text{rank}(\bA)\ll\min(p,N-r)$, and the AR  coefficient matrices $\bPhi_i$'s are sparse in the sense that only a small number of elements in each matrix are nonzero, for $1\leq i\leq d$.

		Assume that the number of common stochastic trends $r$ in Model (\ref{factor_structure}) and the order $d\geq 1$ in Model (\ref{DGP}) 
		are known. Their selections will be discussed below. Note that $\bz_t$ is unobservable in Model \eqref{DGP} and needs to be estimated from the data $\bx_t$. We briefly introduce the proposed two-step estimation procedure. First, similarly to that in \textcite{bai2004estimating}, we estimate the factor loading matrix $\bB$ by solving the following optimization problem:
		\begin{equation}\label{solve_factor}
		(\wh\bB,\wh\bF)=\arg\min_{\bB,\bF}\FNorm{\bX-\bB\bF}^2,\quad\text{subject to}\quad \bB'\bB =\bI_r,
		\end{equation}
		where $\mat{X}=[\bx_1,\bx_2,\ldots,\bx_T]$ and  $\mat{F}=[\bff_1,\bff_2,\ldots,\bff_T]$ are the stacked matrices across the time horizon. It is not hard to show that the optimization method in (\ref{solve_factor}) is equivalent to Principal Component estimation, and the columns of $\wh\bB$ are just the $r$ standardized eigenvectors of $\bX\bX'$ associated with the $r$ largest eigenvalues. Therefore, we choose $\wh\bB_c$ such that its columns are the $N-r$ standardized eigenvectors associated with the $N-r$ smallest  eigenvalues of $\bX\bX'$. Then, we define $\wh\bz_t=\wh\bB_c'\bx_t$, which serves as a proxy of $\bz_t$ and will be used as predictors in the second step of estimation. 
		
		Next, we introduce a method to estimate the coefficient matrices $\bA$ and $\bPhi_i$, for $1\leq i\leq d$. To begin, define $\bPhi=[\bPhi_1,...,\bPhi_d]\in \bbR^{p\times dp}$ and $\bP_{t-1}=(\by_{t-1}',...,\by_{t-d}')'$. For any given penalty parameters $\lambda_{\bA}>0$ and $\lambda_{\bPhi}>0$, we solve the following optimization problem:
		\begin{equation}
		\label{solve_A_Phi}
		(\wh{\mat{A}}, \wh{\mat{\Phi}}) = \arg\min_{\bA,\bPhi}\left\{\frac{1}{2T} \sum_{t=1}^T \|{\by_t - \mat{A}\wh\bz_{t-1} - \mat{\Phi}\bP_{t-1}}\|_2^2 + \lambda_{\mat{A}} \norm{\mat{A}}_* + \lambda_{\mat{\bPhi}} \norm{\vectorize(\mat{\Phi})}_1\right\},
		\end{equation}
		where the data are set to ${\bf 0}$ if the subscript $t\leq 0$. 
		For reduced-rank regression, we refer the readers to the new monograph by 
		\textcite{reinsel2022multivariate}. In particular,  its Chapters 9 to 12 discuss some recent developments in reduced-rank regressions under high-dimensional settings, including 
		the use of nuclear-norm penalty in (\ref{solve_A_Phi}). Similar ideas can also be found in \textcite{negahban2011estimation,chen2013reduced}, among others. However, most of the methods considered in the aforementioned literature only deal with i.i.d. data, while we consider serially dependent data in this paper both theoretically and empirically. 

		It is generally not easy to obtain the true global solutions to the optimization problem in \eqref{solve_A_Phi} because the objective function in the bracket of (\ref{solve_A_Phi}) involves different types of penalties. 
		Therefore, we formulate an iterative procedure to obtain an approximate set of numerical solutions to (\ref{solve_A_Phi}) in Algorithm~\ref{iterative_estimation}. 
		Specifically, for a fixed $\bA$, we can estimate $\bPhi$ via a standard LASSO procedure, and there are several methods and software packages available to obtain sparse solutions. See, for example, \textcite{hastie2015statistical}. When $\bPhi$ is fixed, the estimation of $\bA$ is an instance of  a {\it semidefinite program}. See \textcite{vandenberghe1996semidefinite,ji2009accelerated}. Since the objective function is convex, it is also biconvex  in both sets of parameters. If the estimates in all iterations lie within a small ball around the true parameters, the convergence of the estimates to a stationary point is guaranteed. Because the function is convex, the estimates also achieve a global minimum. 
		See, for example, \textcite{tseng2001convergence} and \textcite{burai2013necessary}. The theoretical results in Section \ref{sec:theory} below are developed for the optimal solutions $\wh\bA$ and $\wh\bPhi$. The simulation results in Section \ref{sec:simulations} suggest that the initial values in Algorithm~\ref{iterative_estimation} have little impact on the asymptotic behavior of the estimates.

		\begin{algorithm}[ht]
			\caption{An Iterative Procedure for Estimating  $\mat{A}$ and $\mat{\Phi}$}
			\label{iterative_estimation}
			\begin{algorithmic}[1]
				
				\Input the data matrices $\mat{Y}=[\by_1,...\by_T]$ and $\wh{\mat{Z}}=[\wh\bz_0,...,\wh\bz_{T-1}]$
				\Output  $\wh\bA\leftarrow\wh\bA^{(k)}$, $\wh\bPhi\leftarrow\wh\bPhi^{(k)}$
				
				\State {Initialize with $k=0$ and $\mat{\Phi}^{(0)} = \mat{0}_{p\times p}$}
				\While {$\wh{\mat{A}}^{(k)}$ or $\wh{\mat{\Phi}}^{(k)}$ is not convergent}
				\State {$\wh{\mat{A}}^{(k+1)} \gets \arg\min_{\bPhi} \frac{1}{2T} \sum_{t=1}^T \|{\by_t - \mat{A}\wh\bz_{t-1} - \wh\bPhi^{(k)}\bP_{t-1}}\|_2^2 + \lambda_{\mat{A}} \norm{\mat{A}}_*$}
				\State {$\wh{\mat{\Phi}}^{(k+1)} \gets \arg\min_{\bPhi} \frac{1}{2T} \sum_{t=1}^T \|{\by_t - \wh\bA^{(k)}\wh\bz_{t-1} - \mat{\Phi}\bP_{t-1}}\|_2^2 + \lambda_{\mat{\Phi}} \norm{\vectorize(\mat{\Phi})}_1$}
				\State $k \gets k+1$
				\EndWhile
			\end{algorithmic}
		\end{algorithm}
		
		Next we turn to the interpretation of the low-rank structure of the matrix $\bA$ in Model (\ref{DGP}). From \textcite{negahban2011estimation}, we see that the estimation of the rank of $\bA$ is equivalent to an optimal selection of the penalty parameter $\lambda_{\bA}$. A similar argument applies to the sparsity of $\bPhi$ and the choices of $\lambda_{\bPhi}$. See also, \textcite{hastie2015statistical}. Suppose the true rank of $\bA$ is equal to $k_0\ll\min\{p,N-r\}$, we may decompose $\bA$ as $\bA=\bC\bR'$ with $\bC\in \bbR^{p\times k_0}$ and $\bR\in \bbR^{(N-r)\times k_0}$. Therefore, $\bR'\bz_t$ is a $k_0$-dimensional stationary random vector and has some predictive power for the future values of $\by_t$. Note that $\bR'\bz_t=\bR'\bB_c'\bx_t$, implying that $\bR'\bB_c'$ is the reduced-rank matrix consisting of $k_0$ significant cointegrating vectors that play important roles in predicting the future values of $\by_t$. In other words, the cointegration rank $N-r$ can be reduced to a smaller number $k_0$ which is useful 
		in prediction and is also easier to interpret. 
		We call $k_0$ the {\em effective cointegration rank} of such a prediction application.

		\subsubsection{An Integrative Reduced-Rank Approach}
		\label{sec:IRRA}
		In this section, we introduce an integrative reduced-rank approach (IRRA) to estimating all the coefficient matrices of Model (\ref{VAR_d}). The approach is similar to the setting in Chapter 10 of \textcite{reinsel2022multivariate} for i.i.d. observations, but we focus on Model (\ref{DGP}) with time-series dependence.
		Specifically, under Models (\ref{factor_structure})--(\ref{DGP}), the IRRA assumes that each set of predictors has its own low-rank coefficient matrix, that is, in addition to the assumption in Section 2.2.1 that $r_{\bA}\ll \min(p,N-r)$, we also assume that $0\leq r_i=\text{rank}(\bPhi_i)\ll p$, for $1\leq i\leq d$, when both $p$ and $N$ are large. This approach bridges the reduced-rank and the sparse models in the sense that the coefficient matrix $\bPhi_i$ is fully sparse with all entries being zero if $r_i=0$. On the other hand, the groupwise low-rank structure in IRRA is more flexible and different from a globally low-rank structure for $\bPhi$ defined in Section 2.2.1. The low-rankness of $\bPhi_i$'s does not necessarily imply that $\bPhi$ is of low rank, while a low-rank matrix $\bPhi$ implies that each $\bPhi_i$ is of low-rank, which cannot exceed that of $\bPhi$. Under this assumption, we consider the following convex optimization problem:
		\begin{equation}
		\label{eq:IRRR}
		(\wh\bA,\wh\bPhi_i)=\arg\min_{\bA,\bPhi_i}\left\{\frac{1}{2T} \sum_{t=1}^T \|{\by_t - \mat{A}\wh\bz_{t-1} - \mat{\Phi}\bP_{t-1}}\|_2^2 + \lambda_{\mat{A}} \norm{\mat{A}}_* +  \sum_{i=1}^d\lambda_i\|\bPhi_i\|_*\right\},
		\end{equation}
		where $\lambda_{\bA}$ and $\lambda_i$ are the penalty parameters associated with $\bA$ and $\bPhi_i$, respectively. Similarly to the setting of \textcite{reinsel2022multivariate}, we may rewrite $\lambda_i$ as $\lambda_i=\lambda_{\bPhi}w_i$ for a global penalty $\lambda_{\bPhi}$ and some prescribed constant $w_i$, for $1\leq i\leq d$. It is clear that $\lambda_i$ is a tuning parameter controlling the amount of regularization applied to $\bPhi_i$. If $w_i=1$, and hence $\lambda_1=...=\lambda_d$, 
		all penalty parameters of $\bPhi_i$ are the same. A simple choice is to take  
		\[w_i=\sigma_1(\bY)\{\sqrt{p}+\sqrt{rank(\bY)}\}/T,\,\,i=1,...,d,\]
		so that we only have a single parameter $\lambda_{\bPhi}$ to control the regularization of the coefficient $\bPhi_i$, for $1\leq i\leq d$.
		

		Because the objective function in (\ref{eq:IRRR}) is convex, there are several feasible algorithms 
		available to solve the optimization problem therein. For example, following the recipe in \textcite{boyd2011distributed}, \textcite{li2019integrative} proposed an Alternating Direction Method of Multipliers (ADMM) algorithm to fit a model similar to that in (\ref{DGP}) with reduced-rank structures. However, the ADMM algorithm is relatively more involved as it alternates between  a primal step and a dual step. In this paper, we propose an easy-to-implement iterative procedure to estimate all the coefficient matrices with reduced-rank structures, which is similar to the block coordinate descent method in \textcite{tseng2001convergence}. The detailed procedure is outlined in  Algorithm~\ref{iRRR} below. Since the objective function is convex, by the argument in \textcite{tseng2001convergence},  the convergence of the estimators via Algorithm~\ref{iRRR} to a stationary point is guaranteed. On the other hand, from \textcite{boyd2004convex}, we know that the conjugate  of a conjugate function of a convex one is itself, by Theorem 2 of \textcite{burai2013necessary}, the stationary point obtained by Algorithm~\ref{iRRR} is a global minimum.  Simulation results in Section 4.2 suggest that the estimators obtained by Algorithm~\ref{algo:IRRA} are comparable to those obtained by the ADMM method, while the former is much easier to implement than the latter in practice. Similarly to the argument used at the end of Section 2.2.1, the cointegration rank has been reduced to a much smaller and effective one due to the reduced-rank structure of $\bA$. We omit the details to save space.
		
		\begin{algorithm}[ht]
			\caption{Iterative procedure for Estimations of $\mat{A}$ and $\mat{\Phi}_i, i=1,\ldots,d$}
			\label{algo:IRRA}
			\label{iRRR}
			\begin{algorithmic}[1]
				
				\Input the data matrices $\mat{Y}=[\by_1,...\by_T]$ and $\wh{\mat{Z}}=[\wh\bz_0,...,\wh\bz_{T-1}]$
				\Output  $\wh\bA \gets\wh\bA^{(k)}$, $\wh\bPhi_i \gets \wh\bPhi^{(k)}_i, i=1,\ldots,d$
				
				\State {Initialize with $k=0$ and $\mat{\Phi}^{(0)}_i = \mat{0}_{p\times p}, i=1,\ldots,d$}
				\While {any of $\wh{\mat{A}}^{(k)}, \wh{\mat{\Phi}}^{(k)}_1,\ldots, \wh{\mat{\Phi}}^{(k)}_d$ is not convergent}
				\State {$\wh{\mat{A}}^{(k+1)} \gets \arg\min_{\bPhi} \frac{1}{2T} \sum_{t=1}^T \|{\by_t - \mat{A}\wh\bz_{t-1} - \wh\bPhi^{(k)}\bP_{t-1}}\|_2^2 + \lambda_{\mat{A}} \norm{\mat{A}}_*$}
				\For {$i=1$ to $d$}
				\State {$\wh{\mat{\Phi}}^{(k+1)}_i \gets \arg\min_{\bPhi_i} \frac{1}{2T} \sum_{t=1}^T \|\by_t - \wh\bA^{(k)}\wh\bz_{t-1} - \wh\bPhi^{(k)} \bP_{t-1} + (\wh\bPhi^{(k)}_{i} - \bPhi_{i}) \by_{t-i} \|_2^2$}
				\Statex {\qquad \qquad \qquad \quad $+ \lambda_{\mat{\Phi}} \norm{\bPhi_i}_*$}
				\EndFor
				\State $k \gets k+1$
				
				\EndWhile
			\end{algorithmic}
		\end{algorithm}

		\subsection{Determination of the Number of Factors}
		\label{subsec:determination_r}
		The estimation of $\wh\bB$ and its orthogonal complement $\wh\bB_c$ in the prior sections is based on a given $r$, which is unknown in practice. There are several methods available in the literature to determine the number of unit-root factors in Equation (\ref{factor_structure}). See, for example, the information criterion in \textcite{bai2004estimating}, the Canonical Correlation Analysis (CCA) method in \textcite{pena2006nonstationary}, the autocorrelation-based method in \textcite{zhang2019identifying} and its modified version in \textcite{gao2021modeling_IJF}, among others. 
		
		In this paper, we adopt the auto-correlation based method of \textcite{gao2021modeling_IJF}. Specifically, let $\wh\bXi=(\wh\bxi_1,...,\wh\bxi_N):=[\wh\bB,\wh\bB_c]$ be the matrix containing all the eigenvectors of $\bX\bX'$ and $\wh f_{j,t}=\wh\bxi_j'\bx_t$ be the $j$-th principal component, for $1\leq j\leq N$.  For some prescribed integer $\bar{k}>0$, define
		\begin{equation}\label{auto:cor}
		S_j(\bar{k})=\sum_{k=1}^{\bar{k}}|\wh\rho_j(k)|,
		\end{equation}
		where  $\wh\rho_j(k)$ is the lag-$k$ sample autocorrelation function (ACF) of the principal component  $\wh f_{j,t}$, for $1\leq j\leq N$. If  $\wh f_{j,t}$ is stationary, then under some mild conditions,  $\wh\rho_j(k)$ decays to zero  exponentially as $k$ increases, and $\lim_{\bar{k}\rightarrow\infty} S_{j}(\bar{k})<\infty$  as $T\rightarrow\infty$.
		If $\wh f_{j,t}$ is unit-root nonstationary, then  $\wh\rho_j(k)\rightarrow 1$, and  $\lim_{\bar{k}\rightarrow\infty} S_{j}(\bar{k})=\infty$  as $T\rightarrow\infty$. Therefore, we start with $j=1$. If the average of the absolute sample ACFs $S_{j}(\bar{k})/\bar{k}\geq \delta_0$ for some constant $0< \delta_0<1$, then $\wh f_{j,t}$ has a unit root and we increase $j$ by $1$ to repeat the detecting process. This detecting process is continued until $S_{j}(\bar{k})/\bar{k}< \delta_0$ or $j=N$. If $S_{j}(\bar{k})/\bar{k}\geq \delta_0$ for all $j$, then $\wh r=N$; otherwise, we denote $\wh r=j-1$.

		
		\subsection{Selection of the Tuning Parameters}
		\label{subsec:tuning}
		In this section, we briefly introduce a way to choose the tuning parameters $\lambda_{\bA}$ and $\lambda_{\bPhi}$, and the order $d$ in (\ref{solve_A_Phi}). We only consider the procedure introduced in Section 2.2.1 since the one in Section 2.2.2 is similar. We first fix the order $d$ and consider the subsamples  $\{\by_1,...,\by_{T_1+j}\}$ and $\{\bx_{1},...,\bx_{T_1+j-1}\}$, for $0\leq j\leq T-T_1-1$ and $T_1<T$.  We then adopt a rolling-window-based method to select $\lambda_{\bA}$ and $\lambda_{\bPhi}$ from a forecasting perspective. Specifically, we prescribe two candidate intervals $[a_1,a_2]$ and $[b_1,b_2]$ with $a_2>a_1>0$ and $b_2>b_1>0$, and choose $(\lambda_{\bA},\lambda_{\bPhi})$ from $[a_1,a_2]\times [b_1,b_2]$ via a grid-search approach.  For any pair $(\lambda_{\bA},\lambda_{\bPhi})\in [a_1,a_2]\times[b_1,b_2]$ and each $0\leq j\leq T-T_1-1$, we first estimate the loading matrix and obtain the stationary process $\{\wh\bz_1,...,\wh\bz_{T_1+j-1}\}$ based on the sample $\{\bx_1,...,\bx_{T_1+j-1}\}$, and apply the iterative procedure in Algorithm \ref{iterative_estimation} to obtain the estimators for all the coefficients based on the subsample $\{\by_1,...,\by_{T_1+j}\}$. We can then obtain the predicted value $\wh\by_{T_1+j+1}$ for $\by_{T_1+j+1}$. We repeat the above procedure for $0\leq j\leq T-T_1-1$ and obtain all the forecasts $\{\wh\by_{T_1+j+1},...,\wh\by_{T}\}$. Define the average of forecast errors as
		\begin{equation}\label{FE}
		\text{FE}_d(\lambda_{\bA},\lambda_{\bPhi})=\frac{1}{p(T-T_1)}\sum_{j=0}^{T-T_1-1}\|\wh\by_{T_1+j+1}-\by_{T_1+j+1}\|_2^2.
		\end{equation}
		Note that the forecast errors defined in \eqref{FE} also depend on the value of $d$, which itself is unknown in practice. We may prescribe an integer $\bar{d}>0$ and search the optimal one over $0\leq \wh d\leq \bar{d}$ such that the forecast error is minimized. Consequently,   the optimal tuning parameters are chosen as
		\begin{equation}\label{tuning}
		(\wh\lambda_{\bA},\wh\lambda_{\bPhi},\wh d)=\arg\min_{\underset{0\leq d\leq \bar{d}}{(\lambda_{\bA},\lambda_{\bPhi})\in [a_1,a_2]\times[b_1,b_2]}}  \text{FE}_d(\lambda_{\bA},\lambda_{\bPhi}).
		\end{equation}
		
		In practice, for simplicity, $\bar{d}$ is often chosen as a small integer provided that the series under study is not seasonal. 
		This choice can also be justified theoretically, because 
		the marginal model of a $p$-dimensional VAR($d$) process is ARMA($pd, p(d-1)$) the order of which can be sufficiently 
		high when $p$ is large; see, for instance, Chapter 2 of \textcite{tsay2014multivariate}. 
		In this paper, we choose $\bar{d}=3$ and the proposed model and procedure work sufficiently well in the real data analysis.


		\section{Theoretical Properties}
		\label{sec:theory}
		
		In this section, we investigate some theoretical properties of the coefficient estimates $\wh{\bB}$, $\wh{\mat{A}}$, and $\wh{\mat{\Phi}}$ under the condition that $p,N,T\rightarrow\infty$. We 
		start with some assumptions and postpone proofs of all  theorems to an online supplement.
		
		\begin{assumption}
			\label{assumption:mixing}
			The process $\{\vect{u}_t, \bm{\varepsilon}_t\}$ is $\alpha$-mixing with the mixing coefficient satisfying the condition $\alpha(k)\leq \exp(-ck^{\gamma})$ for some constants $c > 0$ and $\gamma > 0$, where 
			\[
			\alpha(k) = \sup_i \sup\limits_{\substack{A\in \mathcal{F}_{-\infty}^i\\B\in \mathcal{F}_{i+k}^{\infty}}} \abs{\Pro(A\cap B) - \Pro(A)\Pro(B)},
			\]
			and $\mathcal{F}_i^j$ is the $\sigma$-algebra generated by $\{(\vect{u}_t, \bm{\varepsilon}_t):i \leq t \leq j\}$.
		\end{assumption}
		
		\begin{assumption}
			\label{assumption:sub_exponential}
			$\vect{u}_t$, $\bm{\varepsilon}_t$ and $\be_t$ are sub-exponentially distributed in the sense that there are two constants $C_1,C_2>0$ such that $\Pro(\abs{\vect{v}' (\bm{\eta}_t - \E(\bm{\eta}_t))} > x) \leq C_1 \exp(-C_2x)$ holds for any $x>0$ and $\norm{\vect{v}}_2=1$, where $\bm{\eta}_t$ can be any process of $\vect{u}_t$, $\bm{\varepsilon}_t$ or $\be_t$.
		\end{assumption}

		With the identification condition $\bB'\bB=\bI_r$, the processes $\bff_t$ and $\bu_t$ have an additional strength of $\sqrt{N}$.	For the stationary process $\bu_t$  in (\ref{f_process}), define a normalized process
		\[
		\vect{S}_{T}^r (\vect{t}) = (S_{T}^1 (t_1),\ldots,S_{T}^r (t_r))' = \left(\frac{1}{\sqrt{NT}}\sum_{s=1}^{[T t_1]}u_{1s}, \ldots,\frac{1}{\sqrt{NT}}\sum_{s=1}^{[T t_r]}u_{rs}\right)',
		\]
		where $\vect{t}=(t_1,t_2,\ldots,t_r)'$ is a constant vector with $0\leq t_1\leq\cdots\leq t_r \leq 1$. 
		
		\begin{assumption}
			\label{assumption:convergence}
			For any vector $\vect{t}=(t_1,t_2,\ldots,t_r)'$ with $0\leq t_1\leq\cdots\leq t_r \leq 1$,  there exists a Gaussian process $\vect{W}(\vect{t})=(W_1(t_1),\ldots,W_r(t_r))'$ such that $\vect{S}_{T}^r (\vect{t})\overset{J_1}{\Longrightarrow} \vect{W}(\vect{t})$ on $D_r[0,1]$  as $T\to\infty$, where $\overset{J_1}{\Longrightarrow}$ denotes weak convergence under the Skorokhod $J_1$ topology (see \textcite[Chapter 3]{billingsley1999convergence}), and $\vect{W}(\vect{1})$ has a positive deﬁnite covariance matrix.
		\end{assumption}

		\begin{assumption}
			\label{assumption:dependence}
			For any $i\leq r$, $j\leq N$, it holds that
			\[
			\dfrac{1}{T}\sum_{t=1}^{T} f_{it} \varepsilon_{jt} = O_p(1),
			\]
			uniformly in $i$ and $j$.
		\end{assumption}
		
		\begin{assumption}
			\label{assumption:VAR_stationary}
			For the $p\times p$ matrix polynomial $\bPhi(L)=\bI_p-\sum_{i=1}^d\bPhi_iL^i$, 
			all solutions of the determinant equation $|\bPhi(L)|={\bf 0}$ are outside the unit circle. 
		\end{assumption}
		
		Assumption \ref{assumption:mixing} is standard for dependent random processes. For a theoretical justiﬁcation of the mixing conditions for VAR models, see \textcite{gao2019banded}.  Assumption \ref{assumption:sub_exponential} implies that all moment conditions for the idiosyncratic terms in \textcite{bai2004estimating} are satisfied. Assumption \ref{assumption:convergence} is used to characterize the limiting behavior of the unit-root factors. Similar assumptions are used in \textcite{bai2004estimating}, \textcite{zhang2019identifying}, and \textcite{gao2021modeling_IJF}, among others. Assumptions \ref{assumption:mixing}-\ref{assumption:convergence} 
		imply that all conditions for the common factors and the idiosyncratic terms in \textcite{bai2004estimating} hold. Assumption \ref{assumption:dependence} is used to control the sample covariance between the common factors and the idiosyncratic terms. The rate in Assumption \ref{assumption:dependence} is not strong and can be established under the setting of \textcite{stock1987asymptotic}, where we can assume the factors and idiosyncratic terms have similar structure as those in (2.4) therein. Assumption \ref{assumption:VAR_stationary} is the standard stationarity condition for a VAR process.
		
		Turn to the convergence of the estimated loading matrix and its orthogonal complement. Note that the loading matrix $\bB$ is not uniquely defined due to the identification issue, only the linear space spanned by its columns, denoted by $\mathcal{M}(\bB)$, or the matrix product $\bB\bB'$ is uniquely defined. We state the convergence of the estimated loading matrix and its orthogonal complements in the following theorem.

		\begin{theorem}\label{loading_estimator}
			Suppose Assumptions \ref{assumption:mixing}-\ref{assumption:dependence} hold. Assume $r$ is finite and known. Then, as $N,T\rightarrow\infty$,
			\begin{equation}\label{bbh}
			\norm{\wh{\bB}\wh{\bB}'- \bB\bB'}_2 = O_p(T^{-1})\quad\text{and}\quad
			\norm{\wh{\bB}_c\wh{\bB}_c' - \bB_c\bB_c'}_2 =  O_p(T^{-1}).
			\end{equation}
			Consequently, 
			\[N^{-1/2}\|\wh\bB\wh\bff_t-\bB\bff_t\|_2=O_p(N^{-1/2}+T^{-1/2}).\]
		\end{theorem}

		\begin{remark}
			\label{rmk1}
			From Theorem~\ref{loading_estimator}, the two distances in 
			(\ref{bbh}) are of the same rate which is reasonable because we used the matrix perturbation theory in the proofs and the two matrices play symmetric roles in Lemma A1 of the Supplement. The discrepancy measure used in Theorem~\ref{loading_estimator} is equivalent to the $\sin(\bTheta)$ distance in the literature concerning the distance between two orthogonal matrices. See  (3.2)--(3.4) of \textcite{gao2021two} for details. In addition, based on  \textcite{gao2021two}, the first distance $	\norm{\wh{\bB}\wh{\bB}'- \bB\bB'}_2$ in (\ref{bbh}) is also equivalent to the measure between two linear spaces defined in \textcite{pan2008modelling}:
			\[D(\mathcal{M}(\bB),\mathcal{M}(\wh\bB))=\sqrt{1-\tr(\bB\bB' \wh{\bB}\wh{\bB}')/r},\]
			when $r$ is finite, but the second distance in \eqref{bbh} is not because the dimension of $\bB_c$ is diverging.
		\end{remark}

		The following theorem establishes the convergence of the estimated number of common stochastic trends.	
		\begin{theorem}
			\label{thm:r_consistency}
			Suppose Assumptions \ref{assumption:mixing}--\ref{assumption:dependence} hold. If $N^{1/2}\log(T)T^{-1/2}\rightarrow 0$, then $P(\wh r=r)\rightarrow 1$ as $N,T\rightarrow \infty$, where $\wh r$ is obtained by the autocorrelation-based method in Section \ref{subsec:determination_r}.
		\end{theorem}
		
		Next, turn to the convergence of the estimated regression coefficients obtained in Section \ref{section:model}. 
		To control the errors between the estimated coefficients and the true ones, we  introduce a Restricted Strong Convexity (RSC) condition which is often used in  high-dimensional regularized estimation problems. See \textcite{agarwal2012noisy} and \textcite[Chapter 9]{wainwright2019high} for details. For any given $\lambda_{\bA},\lambda_{\bPhi}>0$, and a matrix $\bDelta\in \bbR^{p\times (N-r+dp)}=[\bDelta_1,\bDelta_2]$ with $\bDelta_1\in\bbR^{p\times (N-r)}$ and $\bDelta_2\in\bbR^{p\times dp}$, we  use a weighted combination  to define an associated norm as follows:
		\begin{equation}\label{psi:n}
		\Psi(\bDelta):=\lambda_{\bA}\|\bDelta_1\|_*+\lambda_{\bPhi}\|\vectorize(\bDelta_2)\|_1.
		\end{equation}
		The restricted strong convexity condition under our setting is defined below.
		
		\begin{definition}
			\label{def:RSC}
			Consider a generic operator $\mathscr{X}:\mathbb{R}^{p\times (N-r+dp)} \mapsto \mathbb{R}^{p\times T}$. We say that it 
			satisﬁes the RSC condition 
			with respect to norm $\Psi$, if
			\[
			\dfrac{1}{2T}\FNorm{\mathscr{X}(\bDelta)}^2 \geq \frac{\kappa_1}{2} \FNorm{\bDelta}^2 - \tau_T \Psi^2(\bDelta), \quad \text{for some}\ \Delta\in\bbR^{p\times(N-r+dp)},
			\]
			where $\kappa_1 > 0$ and $\tau_T > 0 $ are the curvature and tolerance constants, respectively. 
		\end{definition}
		When $\tau_T=0$, the RSC condition in Definition \ref{def:RSC} is called a locally strong convexity condition. See \textcite[Chapter 9]{wainwright2019high}. 
		Denote $\bDelta=[\bDelta_{\bA},\bDelta_{\bPhi}]$ with $\bDelta_{\bA}=\wh\bA-\bA$ and $\bDelta_{\bPhi}=\wh\bPhi-\bPhi$. We now establish the convergence rates of the estimated coefficient matrices below.	
		
		\begin{theorem}
			\label{RRR_LASSO}
			Suppose Assumptions \ref{assumption:mixing}--\ref{assumption:VAR_stationary} hold. For the augmented data matrices $\bZ=[\bz_0,...,\bz_{T-1}]$ and $\bP=[\bP_0,...,\bP_{T-1}]$, where all variables with zero or negative time indexes are set to $0$, if the operator
			\[
			\mathscr{X}([\mat{\Delta}_{\mat{A}}, \mat{\Delta}_{\mat{\Phi}}]) :=\bDelta_{\bA}\bZ+\bDelta_{\bPhi} \bP
			\] satisfies the RSC condition with the norm in the form of \eqref{psi:n}, curvature $\kappa_1$ and tolerance $\tau_T$ such that
			\[
			\kappa_1 \geq C\tau_T r_{\mat{A}} \lambda_{\bA}^2,\ \text{ and }\  \kappa_1 \geq C\tau_T s_{\mat{\Phi}} \lambda_{\bPhi}^2,
			\]
			where $r_{\mat{A}}$ and $s_{\mat{\Phi}}$ are the rank of $\mat{A}$ and the cardinality of the support of $\mat{\Phi}$, respectively, then with the regularization parameters $\lambda_{\mat{A}}$ and $\lambda_{\mat{\Phi}}$ satisfying
			\[
			\lambda_{\mat{A}} \geq \dfrac{{3}}{T} \norm{\mat{E}\mat{Z}'}_2 \ \text{ and }\  \lambda_{\mat{\Phi}} \geq \dfrac{2}{T} \norm{\vectorize\left(\mat{E}  \mat{P}' \right) }_\infty,
			\]
			where $\bE=[\be_1,...,
			\be_T]$ is the error matrix of (\ref{DGP}),
			we have
			\begin{equation}
			\FNorm{\wh{\mat{A}}-\mat{A}}^2+\FNorm{\wh{\mat{\Phi}}-\mat{\Phi}}^2 \leq C  \dfrac{\lambda_{\mat{A}}^2 r_{\mat{A}} + \lambda_{\mat{\Phi}}^2 s_{\mat{\Phi}}}{\kappa_1^2}.
			\end{equation}
		\end{theorem}
		
		\begin{remark}
			\label{remark2}
			(i) Under Assumptions \ref{assumption:mixing}--\ref{assumption:VAR_stationary}, by the Bernstein-type inequality for weakly dependent data in \textcite{merlevede2011bernstein} and the argument in the proofs of Lemma 3 in \textcite{negahban2011estimation}, it is not hard to show that $\|\bE\bZ'\|_2=O_p(\sqrt{(p+N)T})$. Then, the condition for $\lambda_{\bA}$ becomes $\lambda_{\bA}\geq C\sqrt{(p+N)/T}$. Similarly, by the Bernstein-type inequality in \textcite{merlevede2011bernstein}, we can also show that $\norm{\vectorize\left(\bE \bP' \right) }_\infty=O_p(\sqrt{T\log(p)})$, and therefore, the condition for $\lambda_{\bPhi}$ reduces to $\lambda_{\bPhi}\geq C\sqrt{\log(p)/T}$, which is the same as that in the LASSO literature. See \textcite{wainwright2019high}.\\
			(ii) For a properly chosen $C_*>0$ such that $\lambda_{\bA}= C_*\sqrt{(p+N)/T}$ and $\lambda_{\bPhi}= C_*\sqrt{\log(p)/T}$ satisfy the conditions in Theorem \ref{RRR_LASSO}, under the setting that $p/T\rightarrow 0$ and $N/T\rightarrow 0$, we may choose an $\tau_T>0$ such that $\kappa_1 > C\max(\tau_T r_{\mat{A}} \lambda_{\bA}^2,\tau_T s_{\mat{\Phi}} \lambda_{\bPhi}^2)>0$ is a positive constant, and then it follows from Theorem \ref{RRR_LASSO} that
			\[	\FNorm{\wh{\mat{A}}-\mat{A}}^2+\FNorm{\wh{\mat{\Phi}}-\mat{\Phi}}^2 \leq C\left(r_{\bA}\frac{p+N}{T}+s_{\bPhi}\frac{\log(p)}{T}\right)\rightarrow 0,\]
			as $p,N,T\rightarrow \infty$ for finite $r_{\bA}$ and $s_{\bPhi}$, implying that the estimated coefficient matrices are consistent.\\
			(iii) Under the settings in Remark \ref{remark2}(ii), we immediately obtain the consistencies for both matrices:
			\begin{equation}\label{conss}
			\FNorm{\wh{\mat{A}}-\mat{A}}^2\rightarrow 0\,\,\text{and}\,\,\FNorm{\wh{\mat{\Phi}}-\mat{\Phi}}^2\rightarrow 0,\,\,	\text{as}\,\, p,N,T\rightarrow\infty.
			\end{equation}
			If there is a positive constant $C>0$ such that the minimum nonzero singular value of $\bA$ and the minimum absolute elements in $\bPhi$, denoted by $\sigma_{r_{\bA}}$ and $|\bPhi|_{\min}$ respectively, satisfy $\sigma_{r_{\bA}}>C>0$ and $|\bPhi|_{\min}>C>0$ as $p,N,T\rightarrow\infty$, (\ref{conss}) implies that $P(\wh r_{\bA}=r_{\bA})\rightarrow 1$ and $P(\widehat{\mathcal{S}}_{\bPhi}=\mathcal{S}_{\bPhi})\rightarrow 1$, where $\wh r_{\bA}=$ rank$(\wh\bA)$, $r_{\bA}=$ rank$(\bA)$, and $\widehat{\mathcal{S}}_{\bPhi}$ and $\mathcal{S}_{\bPhi}$ contain all the indexes of the nonzero elements in $\wh\bPhi$ and $\bPhi$, respectively. We omit the details to save space.
		\end{remark}
		
		To establish properties of the estimated coefficients using the IRRA of Section 2.2.2, we first introduce a restricted set that is constructed by a projection of any matrix onto a subspace generated by another one of the same shape. Specifically, for any $m\times n$ matrix $\mat{\Theta}$,  we  perform a singular value decomposition (SVD) $\mat{\Theta} = \mat{U}\mat{D}\mat{V}'$ with a partition as follows,
		\begin{equation}
		\label{SVD}
		\mat{\Theta} = \begin{bmatrix}
		\mat{U}_k \ \mat{U}_{k,c} 
		\end{bmatrix}\begin{bmatrix}
		\mat{D}_k & \\
		& \mat{D}_{k,c}
		\end{bmatrix}\begin{bmatrix}
		\mat{V}_k' \\
		\mat{V}_{k,c}'
		\end{bmatrix},
		\end{equation}
		where $\mat{U}_k\in\mathbb{R}^{m\times k}$ and $\mat{V}_k\in\mathbb{R}^{n\times k}$ are the sub-matrices consisting of the left and right singular vectors associated with the $k$ largest  singular values of $\mat{\Theta}$, respectively, and $\mat{U}_{k,c} \in\mathbb{R}^{m\times (m-k)}$ and $\mat{V}_{k,c} \in\mathbb{R}^{n\times (n-k)}$ are the remaining ones. Similarly to \textcite{negahban2011estimation}, we define  two subspaces as follows,
		\begin{equation}
		\label{subspaces}
		\begin{split}
		\mathcal{S}_\mat{\Theta}(k) = \{\bA\in\mathbb{R}^{m\times n}:\range(\bA)\subseteq \range(\mat{U}_k), \range(\bA')\subseteq \range(\mat{V}_k)\},\,\,\text{and}\\
		\mathcal{S}_\mat{\Theta}^\perp (k) = \{\bA\in\mathbb{R}^{m\times n}:\range(\bA)\perp \range(\mat{U}_k), \range(\bA')\perp \range(\mat{V}_k)\}.
		\end{split}
		\end{equation}
		For any matrix $\mat{M}\in\mathbb{R}^{m\times n}$, we decompose it as $\mat{M}=\mat{M}_1+\mat{M}_2$, where
		\begin{equation}
		\label{decom12}
		\mat{M}_2= \mat{U}_{k,c} \mat{U}_{k,c}' \mat{M} \mat{V}_{k,c} \mat{V}_{k,c}',\ \text{and}\   \mat{M}_1= \mat{M} - \mat{M}_2.
		\end{equation}
		Because $\mat{M}_2\in \mathcal{S}_\mat{\Theta}^\perp (k)$, we use $\Pi_{\mathcal{S}_\mat{\Theta}^\perp (k)}(\mat{M}) = \mat{M}_2$ 
		to denote the projection of matrix $\mat{M}$ onto the subspace $\mathcal{S}_\mat{\Theta}^\perp (k)$. 
		
		Turn to the estimated coefficients using the IRRA of Section 2.2.2. By an abuse of notation, we define $\bDelta=[\bDelta_{\bA},\bDelta_{\bPhi}]=[\bDelta_{\bA},\bDelta_{\bPhi_1},...,\bDelta_{\bPhi_d}]$ with $\bDelta_{\bA}=\wh\bA-\bA\in \mathbb{R}^{p \times (N-r)}$ and $\bDelta_{\bPhi_i}=\wh\bPhi_i-\bPhi_i\in\mathbb{R}^{p \times p}$, and hence $\bDelta_{\bPhi}\in\mathbb{R}^{p\times dp}$. 
		We decompose $\bDelta_{\bA}$ as $\bDelta_{\bA} = \bDelta_{\bA,1} + \bDelta_{\bA,2}$ and $\bDelta_{\bPhi_i}$ as $\bDelta_{\bPhi_i} = \bDelta_{\bPhi_i,1} + \bDelta_{\bPhi_i,2}$, for $1\leq i\leq d$. It follows that  $\bDelta_{\bA,2} = \Pi_{\mathcal{S}_{\bA}^\perp(r_{\bA})}(\bDelta_{\bA})$, and $\bDelta_{\bPhi_i,2} = \Pi_{\mathcal{S}_{\bPhi_i}^\perp(r_i)}(\bDelta_{\bPhi_i})$, for $1\leq i\leq d$. We define a restricted set $\mathcal{C}$ as
		\begin{equation}
		\label{eq:set}
		\begin{split}
		\mathcal{C}(r_1,...,r_d) = \bigg\{ \bDelta \in \bbR^{p\times (N-r+dp)} \mid& \norm{\bDelta_{\bA,2}}_*  + \sum_{i=1}^{d} \norm{\bDelta_{\bPhi_i,2} }_* \leq 3 \norm{\bDelta_{\bA,1}}_*\\
		&+ 3 \sum_{i=1}^{d}  \norm{\bDelta_{\bPhi_i,1}}_* \bigg\}.   
		\end{split}
		\end{equation}
		We make an additional assumption below.
		\begin{assumption}
			\label{asm:RE}
			For the operator $\mathscr{X}$ defined in Theorem~3, we assume
			\[
			\frac{1}{2T} \FNorm{\mathscr{X}(\bDelta)}^2 = \frac{1}{2T}\FNorm{\bDelta_{\bA}\bZ+\bDelta_{\bPhi}\bP}^2\geq \kappa_2 \FNorm{\bDelta}^2,\,\,\text{for all}\,\, \bDelta\in\mathcal{C}(r_1,...,r_d),
			\]
			where $\kappa_2 > 0$ is a constant and  $\mathcal{C}(r_1,...,r_d)$ is defined in (\ref{eq:set}).
		\end{assumption}
		
		Note that Assumption 6 is a locally restricted strong convexity condition by setting $\tau_T=0$ in Definition~\ref{def:RSC}. Similar assumptions are also considered in Chapter 10 of \textcite{reinsel2022multivariate} for {\it i.i.d.} data. We next state the convergence of the estimated coefficients based on the IRRA of Section 2.2.2.
		\begin{theorem}\label{tm4}
			Assume Assumptions 1--5 hold. Suppose the predictor matrices $\bZ$ and $\bP$ satisfy the condition in Assumption 6 over the set $\mathcal{C}$ defined in \eqref{eq:set}. If $\lambda_{\bA}$ and $\lambda_i$ satisfy
			\[
			\lambda_{\mat{A}} \geq \dfrac{3}{T} \norm{\mat{E}\mat{Z}'}_2 \ \text{ and }\ \lambda_{i} \geq \dfrac{2}{T} \norm{\bE L^i(\bY)'}_2,\ \text{ for }\ i=1,2,\ldots,d,
			\]
			then, as $p,N,T\rightarrow\infty$,  we have 
			\[
			\FNorm{\wh\bA-\bA}^2 + \sum_{i=1}^{d}\FNorm{\wh\bPhi_i-\bPhi_i}^2 \leq C\left(r_{\bA}\lambda_{\bA}^2 + \sum_{i=1}^{d} r_{i} \lambda_{i}^2\right) / \kappa_2^2.
			\]
		\end{theorem}
		\begin{remark}
			(i) Assumption 6 can be replaced by a weaker RSC condition as that in Theorem 3, and the results in Theorem 4 continue to hold with minor modifications in the proofs given in the online supplement. \\
			(ii) The convergence rates of the estimated coefficients are the same as those in Chapter 10 of \textcite{reinsel2022multivariate}, even for time-series data with mild serial dependence.\\
			(iii) By the discussions in Remark 2(i)--(ii), we may also choose $\lambda_{\bA}= C_*\sqrt{(p+N)/T}$ and $\lambda_{i}= C_*\sqrt{p/T}$ for some constant $C_*>0$ satisfying the conditions in Theorem \ref{tm4}, such that the convergence results in Theorem \ref{tm4} can be rewritten as
			\[\FNorm{\wh\bA-\bA}^2 + \sum_{i=1}^{d}\FNorm{\wh\bPhi_i-\bPhi_i}^2\leq C\left\{\frac{(p+N)r_{\bA}}{T}+\sum_{i=1}^d\frac{pr_i}{T}\right\},\]
			which approaches zero asymptotically under the setting that $p/T\rightarrow 0$ and $N/T\rightarrow 0$, implying that the estimators are consistent. It is straightforward to see that the convergence rates above are slightly slower than those in Remark 2(ii) if the sparsity parameter therein satisfies $s_{\Phi}/p\rightarrow 0$, which is often the case in sparse regression. This is understandable since there are usually more autoregressive coefficients to estimate in a reduced-rank regression in (\ref{eq:IRRR}) than in the sparse counterpart in (\ref{solve_A_Phi}).
		\end{remark}

		\section{Simulation Study}
		\label{sec:simulations}
		
		In this section, we evaluate the finite-sample performance of the proposed methodologies under the scenarios when both $p$ and $N$ are increasing from small to large. Though the dimensions of $\bB$ and $\wh\bB$ are not necessarily the same, as estimation error in $r$ may occur, the discrepancy measure adopted in Theorem \ref{loading_estimator} remains valid. To simplify the presentation and without loss of generality,  we set $d=1$ in \eqref{DGP}, and similar results can also be obtained for other choices of finite $d$.
		
		\subsection{Example 1: The Reduced-Rank and Sparse Regression}
		\subsubsection{Data Generating Process}
		\label{sec:DGP}
		We follow the data generating process in \eqref{factor_structure} and \eqref{DGP} and consider a three-factor model, 
		where the factors are $I(1)$ processes generated by \eqref{f_process}.
		{We further multiply  the factors  by $\sqrt{N}$ because  we will use orthonormal loading matrices  below and the  strength of general loadings is imposed on the factors in line with the assumptions and identification conditions.} While the number of factors $r=3$ is fixed, we set $p=20, 40, 60$, and $N=20, 40, 60$, respectively, and in each configuration of $(p,N)$, we set the sample size $T=400, 800, 1200$ to illustrate the proposed method and to exam certain theoretical properties of the estimators. In order to obtain reproducible results, we initialize a random generator in the \texttt{NumPy} package in \texttt{Python} by setting the seed to $1024$, and this seed is used throughout the simulation.
		
		To begin, we need to obtain the coefficient matrices of the model. We start with generating the loading matrix $\bB$ and its corresponding orthogonal complement $\bB_c$. As $[\wh\bB,\wh\bB_c]$ is an $N\times N$ full-rank orthonormal matrix, we first randomly generate an $N\times N$ orthogonal matrix, and divide its columns in such a way that the submatrix with the first $r$ columns is chosen as $\bB$ and the remaining columns form naturally the $\bB_c$ matrix. For the low-rank matrix $\mat{A}$, we first randomly generate two orthonormal matrices $\mat{U}\in\bbR^{p\times p}$ and $\mat{V}\in \mathbb{R}^{(N-r)\times (N-r)}$, and a $p\times (N-r)$ rectangular diagonal matrix $\mat{D}$ with only five positive entries on the upper left of the diagonal while all the other entries are set to zero. The positive diagonal entries in $\bD$ are drawn independently from a uniform distribution on the interval of $[0.1,1)$ so that all the five elements are strictly greater than $0$.
		The matrix $\mat{A}$ with rank $r_{\bA}=5$ is then chosen as $\mat{A} = \mat{U}\mat{D}\bV'$. Next, for the sparse matrix $\bPhi$, we first create a sparse matrix $\bPhi_1$ with only $20$ randomly located non-zero entries each of which is drawn  uniformly  on the intervals $(-1,-0.1] \cup [0.1,1)$. In order to guarantee the stationarity of $\by_t$ in \eqref{DGP}, we use the normalized matrix $\bPhi = 0.9 \times \bPhi_1/\norm{\bPhi_1}_2$ as the autoregressive coefficient matrix, which implies that  Assumption \ref{assumption:VAR_stationary} holds.

		For each configuration of $(p,N,T)$, with the coefficient matrices $\bB, \bB_c, \bA$ and $\bPhi$ chosen by the aforementioned methods, we generate $\bx_t, \bz_t$ and $\by_t$ according to Models \eqref{factor_structure}, \eqref{zt} and \eqref{DGP}, respectively. To obtain stable results, we use $500$ replications for each $(p,N,T)$ configuration and set $\bve_t\sim N(\vect{0},\bI_N)$, $\vect{u}_t \sim N(\vect{0},\bI_r)$, and $\vect{e}_t \sim N(\vect{0},\bI_p)$ in each realization.

		\subsubsection{Performance Evaluation}
		
		We ﬁrst study the performance of (\ref{auto:cor}) in estimating the number of factors. Because the data generating process $\bx_t$ of the previous section is independent of the dimension $p$, we only illustrate the proposed method for the case of $p=20$, and similar results can also be obtained for other cases. 
		Table~\ref{table:estimate_r} reports the empirical probabilities of $P(\wh{r} = r)$ based on $500$ repetitions for each $(N,T)$ 
		configuration when $p=20$, where we use the method described in Section \ref{subsec:determination_r} 
		with $\bar{k}=10$ and $\delta_0=0.3$. From Table~\ref{table:estimate_r}, we see that the auto-correlation based method can successfully recover the number of common stochastic trends. This is understandable because all the factors used in the simulation are strong ones. Similar results can also be found in \textcite{bai2002determining} and \textcite{lam2012factor}.
		
		\begin{table}[!htb]
			\centering
			\caption{Empirical probabilities of $P(\wh{r} = r)$ for various $(N,T)$ configurations, where the value is $r=3$ and the dimension is $p=20$. The estimation method of Section \ref{subsec:determination_r} with $\bar{k}=10$ and $\delta_0=0.3$ is used, and the results are based on $500$ iterations.}
			\label{table:estimate_r}
			\begin{tabular}{cccc}
				\hline
				& \multicolumn{3}{c}{$N$}\\
				\cline{2-4}
				$T$ & 20 & 40 & 60\\
				\hline
				400 & 100.00\% & 100.00\% & 100.00\% \\
				800 & 100.00\% & 100.00\% & 100.00\% \\
				1200 & 100.00\% & 100.00\% & 100.00\%\\
				\hline
			\end{tabular}
		\end{table}

		Next, we consider the estimation accuracy of the loading matrix $\bB$, which is measured by $\norm{\bB\bB' - \wh{\bB}\wh{\bB}'}_2$ over 500 replications. For the same reason mentioned before, we only show the results for the case of $p=20$. Boxplots of the discrepancies are shown in Figure \ref{fig:error_B}, from which 
		we see that  for each $N$, the discrepancy between the estimated loading matrix and the true one decreases as the
		sample size $T$ increases. This result is in agreement with our theorems.
		Furthermore, we also evaluate the estimation errors of the extracted factors. For each $(N,T)$ configuration, we define the the root-mean-squared-error (RMSE) of the estimated factors as
		\begin{equation}
		\label{denoised_RMSE}
		\mathit{RMSE} = \bigg(\dfrac{1}{NT} \sum_{t=1}^{T} \norm{\bB \bff_t - \wh\bB \wh{\bff}_t}_2^2\bigg)^{1/2},
		\end{equation}
		which quantifies the accuracy in recovering the common stochastic trends. Figure \ref{fig:error_X_recover} shows the results via boxplots using 500 replications. From Figure \ref{fig:error_X_recover}, we see clearly that, the recovery errors of the common factors  decrease as the sample size $T$ increases, which is consistent with the theoretical results in Theorem \ref{loading_estimator}.
		
		\begin{figure}[htb]
			\centering 
			\includegraphics[width=0.9\textwidth,height=2.5in]{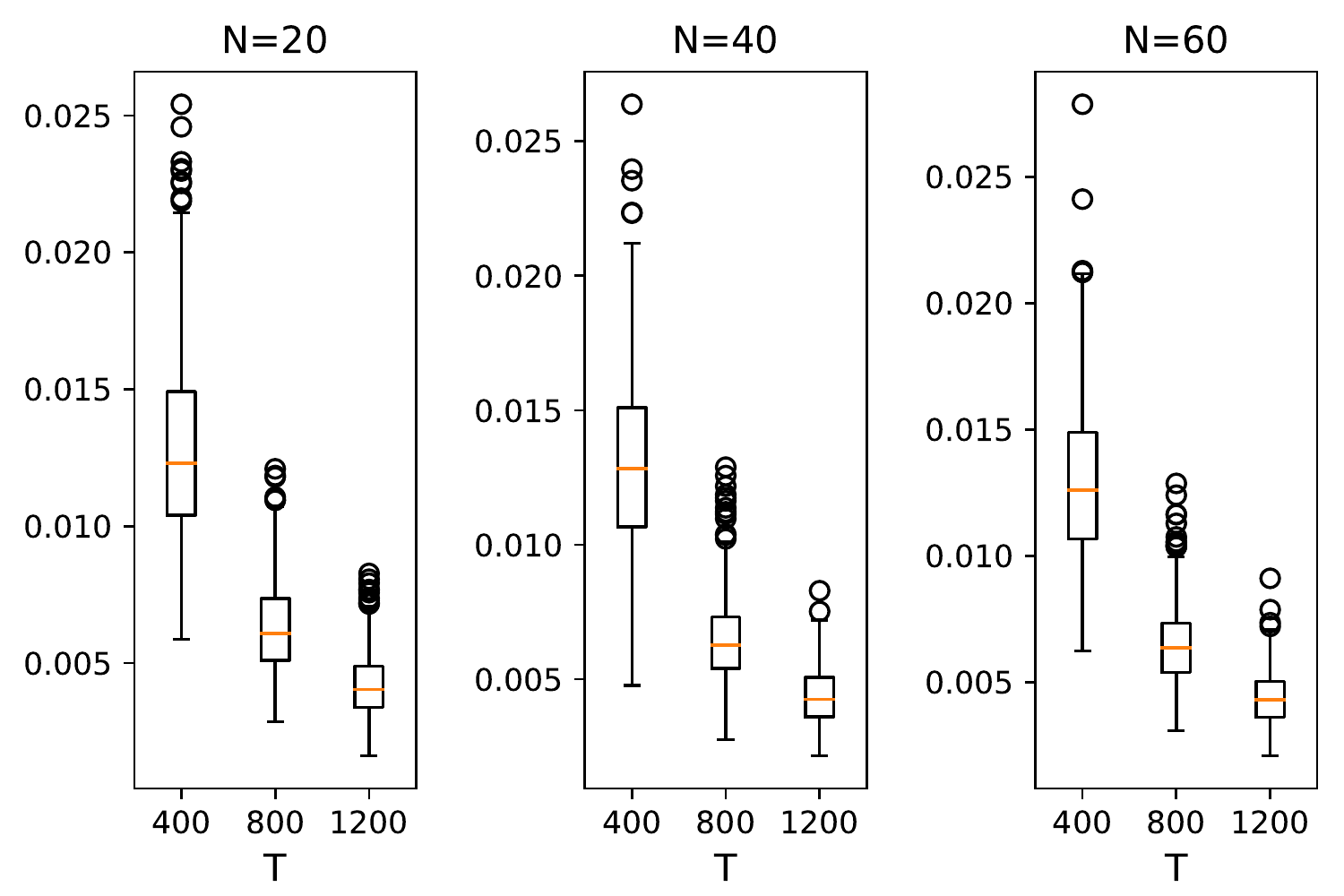}
			\caption{Boxplots of $\norm{\bB\bB' - \wh{\bB}\wh{\bB}'}_2$ with $r=3$ and $p=20$ in Example 1. For each $N$ (dimension of $\bx_t$), the sample sizes used are $400$, $800$ and $1200$, respectively. The results are based on $500$ replications.}
			\label{fig:error_B}
		\end{figure}
		
		\begin{figure}[htb]
			\centering 
			\includegraphics[width=0.9\textwidth,height=2.5in]{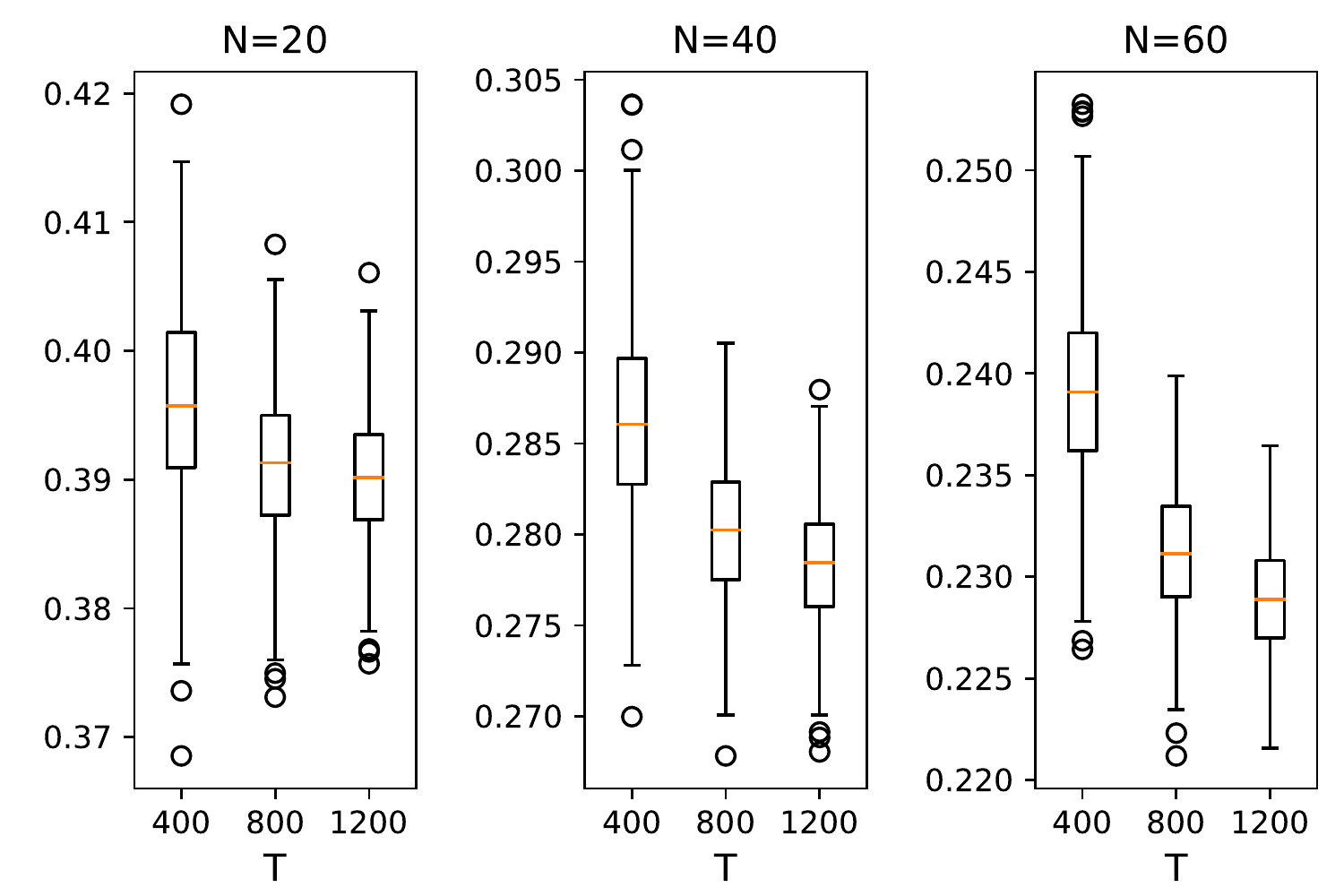}
			\caption{Boxplots for RMSE of the extracted factors defined in (\ref{denoised_RMSE}) with $r=3$ and $p=20$ in Example 1. For each $N$, the sample sizes used are $400$, $800$, and $1200$, respectively. The results are based on $500$ replications.}
			\label{fig:error_X_recover}
		\end{figure}

		We then study the estimation accuracy of the low-rank matrix $\bA$ and the sparse matrix $\mat{\Phi}$ using the procedure in Algorithm~\ref{iterative_estimation}. For simplicity, we set the tuning parameters $\lambda_{\bA} = \sqrt{(p+N)/T}$ and $\lambda_{\bPhi} = \sqrt{\log(p)/T}$, which are just taken from the rates  discussed in Remark \ref{remark2}(ii) by setting $C_*=1$ and this choice is good enough to produce satisfactory performance in the simulation. In practice, we may choose an optimal $C_*$ from an interval using grid search.  Due to the identification issue as that for $\bB$, we also use $\norm{\bA\bA' - \wh{\bA}\wh{\bA}'}_2$ to evaluate the discrepancy between $\wh\bA$ and $\bA$. Because there is no identification issues with $\bPhi$ and the estimated $\wh\bPhi$, we use $\norm{\mat{\Phi} - \wh{\mat{\Phi}}}_2$ to measure the estimation accuracy of the autoregressive coefficients. Boxplots of the estimation errors for $\wh\bA$ and $\wh\bPhi$ are presented in the Figures \ref{fig:error_A} and \ref{fig:error_Phi}, respectively. As expected from Theorem \ref{RRR_LASSO}, in each case of $(p,N)$, the estimation errors of $\bA$ and $\bPhi$ both decrease as the  sample size $T$ increases, which is also consistent with our theoretical properties.
		
		\begin{figure}[ht]
			\centering 
			\includegraphics[width=0.9\textwidth,height=2.5in]{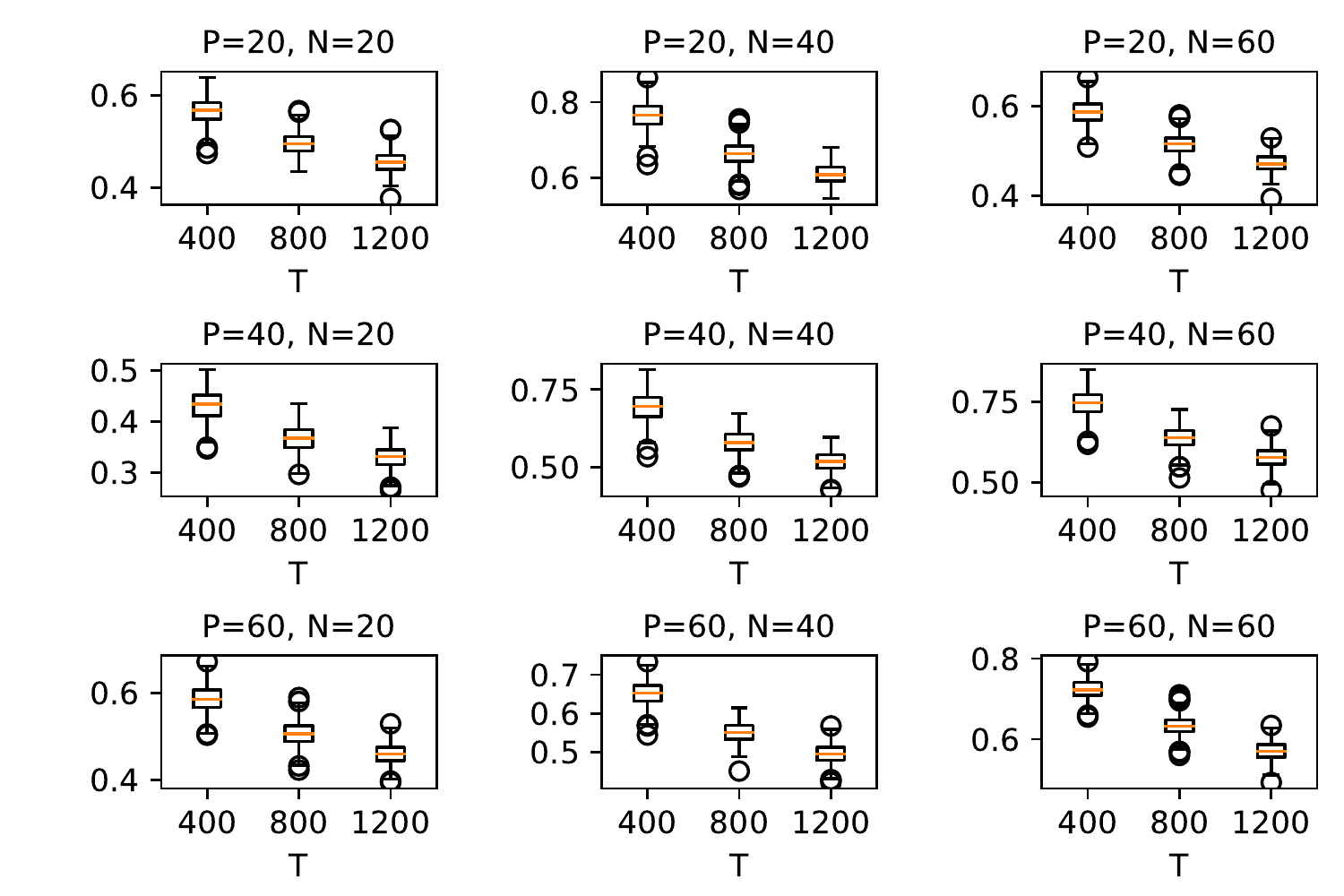}
			\caption{Boxplots for $\norm{\mat{A}\mat{A}' - \wh{\bA}\wh{\bA}'}_2$ of Example 1. For each $(p,N)$ configuration, the sample sizes used are $T=400,800,1200$, and the number of repetitions is $500$.}
			\label{fig:error_A}
		\end{figure}
		
		\begin{figure}[ht]
			\centering 
			\includegraphics[width=0.9\textwidth,height=2.5in]{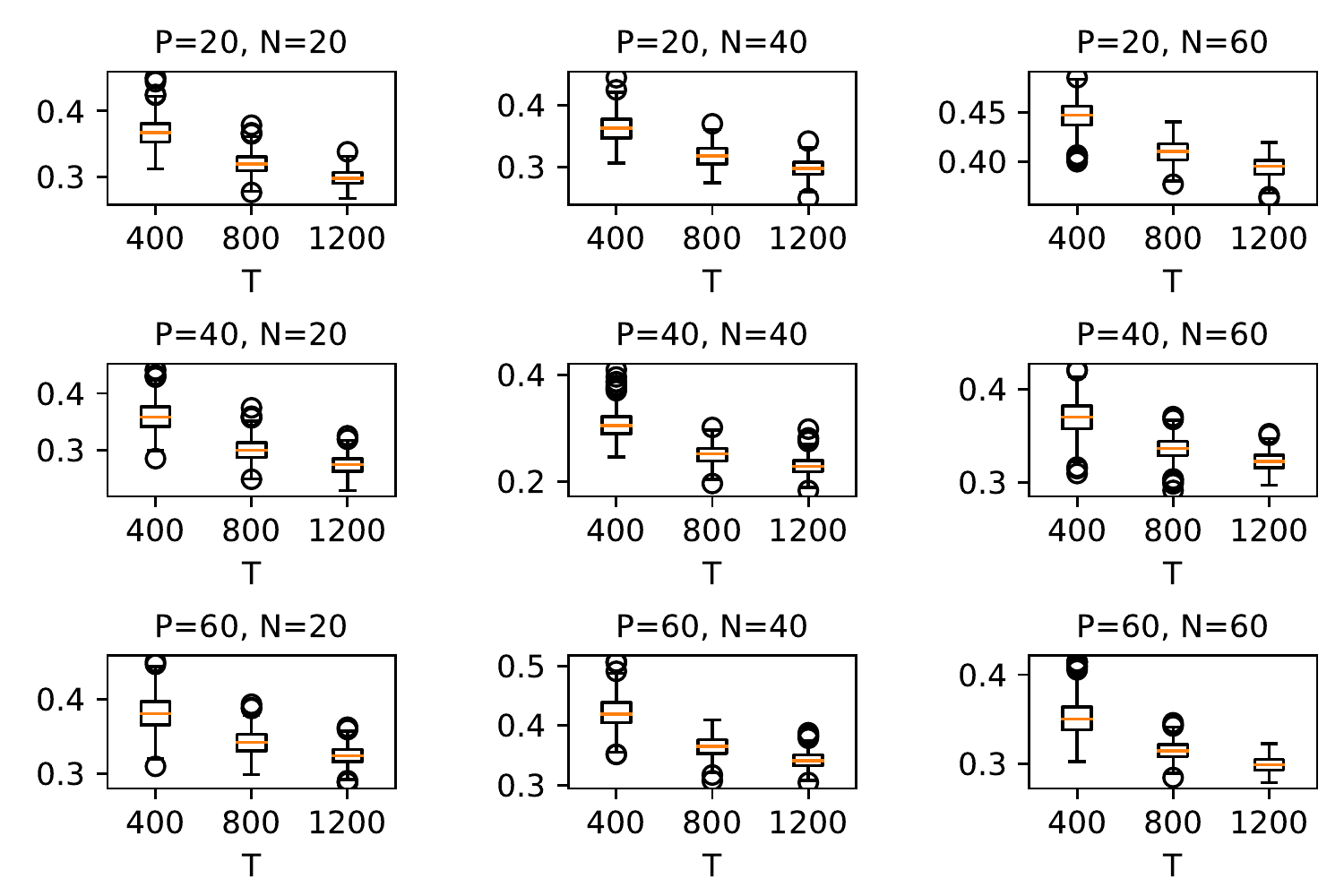}
			\caption{Boxplots for $\norm{\mat{\Phi} - \wh{\mat{\Phi}}}_2$ of Example 1. In each case of $(p,N)$, the sample sizes used are $T=400,800,1200$, and the results are based on $500$ repetitions.}
			\label{fig:error_Phi}
		\end{figure}

		
		Finally, we consider the estimation errors of the estimated explanatory variables and the true ones in Model (\ref{DGP}). Similarly to that in (\ref{denoised_RMSE}), we define the RMSE for the regression model (\ref{DGP}) as
		\begin{equation}\label{rmse2}
		\mathit{RMSE} = \bigg(\dfrac{1}{pT}\sum_{t=1}^{T} \norm{\mat{A} \bz_{t-1} + \mat{\Phi} \by_{t-1} - (\wh{\mat{A}} \wh{\bz}_{t-1} + \wh{\mat{\Phi}} \by_{t-1})}_2^2\bigg)^{1/2}, 
		\end{equation}
		which is similar to the in-sample errors of a regression model.
		Figure \ref{fig:fit_R2} displays boxplots of the  RMSEs in (\ref{rmse2}). 
		From the plot, we see that the patterns of the boxplots are similar to those obtained before.
		For each given $(p,N)$, the RMSEs decrease as  the sample size $T$ increases, illustrating the efficacy of the proposed method. Overall, the simulation results indicate that the proposed procedure works well in recovering the estimated coefficients.
		
		\begin{figure}
			\centering 
			\includegraphics[width=0.9\textwidth,height=2.5in]{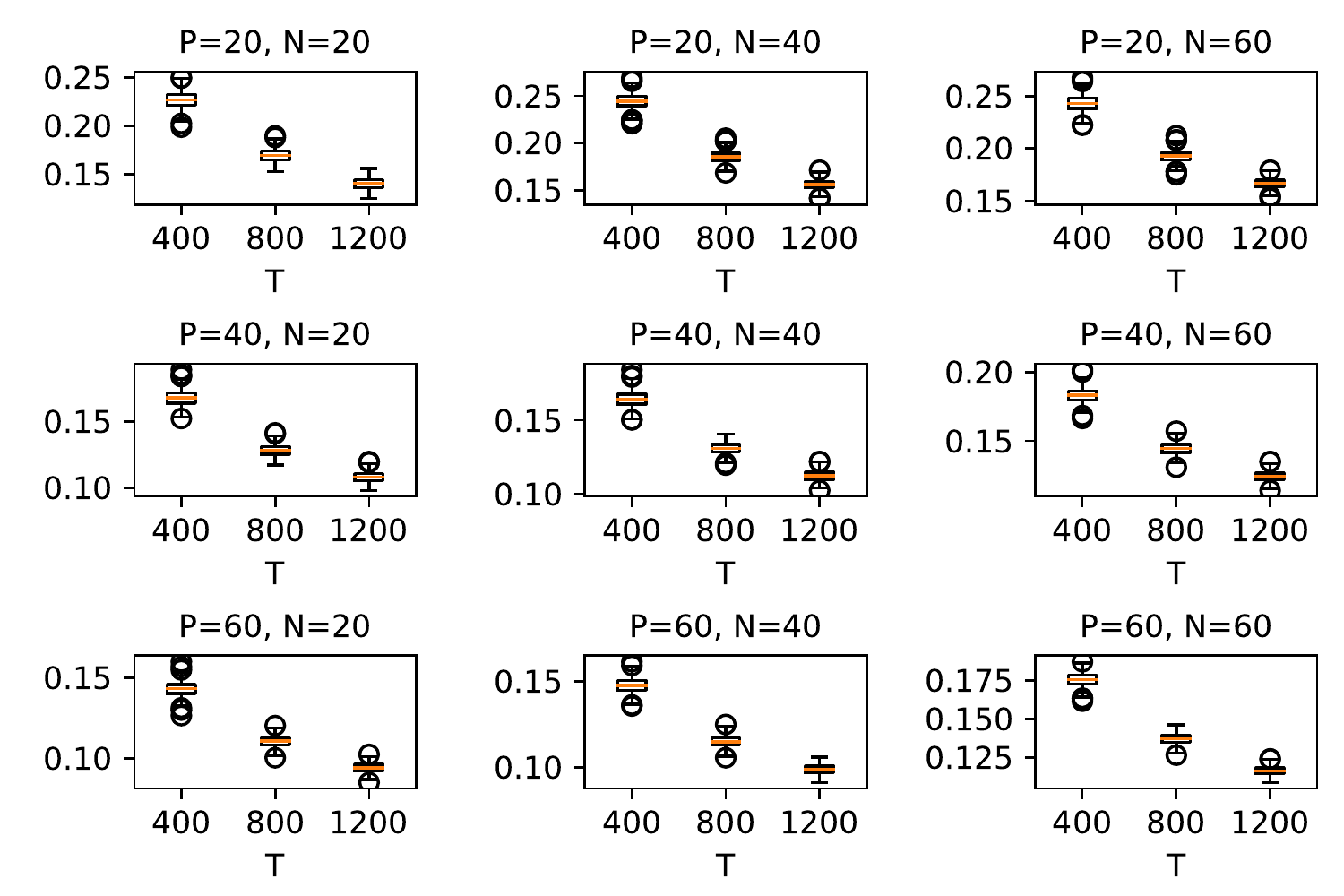}
			\caption{Boxplots for denoised RMSE of $\by_t$ defined in \eqref{rmse2} of Example 1. For each pair of $(p,N)$, the sample sizes used are $400$, $800$ and $1200$ and the number of repetitions is $500$.}
			\label{fig:fit_R2}
		\end{figure}

		\subsection{Example 2: The Integrative Reduced-Rank Approach}
		In this example, we investigate the performance of IRRA of  Section~\ref{sec:IRRA}. First, we generate the data $\bx_t$ using the same method as that of Section~\ref{sec:DGP}. 
		Second, unlike the sparse autoregressive matrices in Example 1, we generate two low-rank matrices $\bPhi_1 \in \bbR^{p\times p}$ and $\bPhi_2\in \bbR^{p\times p}$ under the context of the IRRA. Without loss of generality, we generate the those low-rank matrices 
		in the same way as that of $\bA$, and set $\mathit{rank}(\bPhi_1)=
		\mathit{rank}(\bPhi_2)=3$.
		Third, the process $\by_t$ is then generated according to \eqref{DGP} with the  coefficients given above, 
		where we choose $d=2$. 
		
		Similarly to the procedure in Example 1, we apply (\ref{auto:cor}) and Algorithm~\ref{algo:IRRA} to estimate the number of factors and the coefficients, respectively. 
		Since the performance of the auto-correlation based method is shown in Example 1, we omit the details here. Figures~\ref{fig:error_A_ADMM}, \ref{fig:error_Phi1_ADMM} and \ref{fig:error_Phi2_ADMM} show the discrepancies between the estimated coefficients and the true ones using Algorithm~\ref{algo:IRRA}. From these boxplots, we see that, for each configuration of $(p, N)$, all three coefficient estimates $\wh\bA, \wh\bPhi_1$ and $\wh\bPhi_2$ converge to the true ones as $T \to \infty$, which is consistent with our theory. For comparison, we also test the ADMM algorithm of  \textcite{li2019integrative}, and find that 
		the results of ADMM are quite close to those of the Algorithm~\ref{algo:IRRA} in the sense that the distance $\norm{\wh{\bTheta}^{(\text{Ite})}-\wh{\bTheta}^{(\text{ADMM})}}_2/\norm{\wh{\bTheta}^{(\text{Ite})}}_2$ is less than $10\%$ in most cases, where $\wh\bTheta=\wh\bA, \wh\bPhi_1$ or $\wh\bPhi_2$, and $\wh{\bTheta}^{(\text{Ite})}$ and $\wh{\bTheta}^{(\text{ADMM})}$ are the coefficient matrices estimated by Algorithm~\ref{algo:IRRA} and ADMM, respectively. Therefore, we omit the results obtained by the ADMM algorithm to save space. 
		
		\begin{figure}
			\centering 
			\includegraphics[width=0.9\textwidth,height=2.5in]{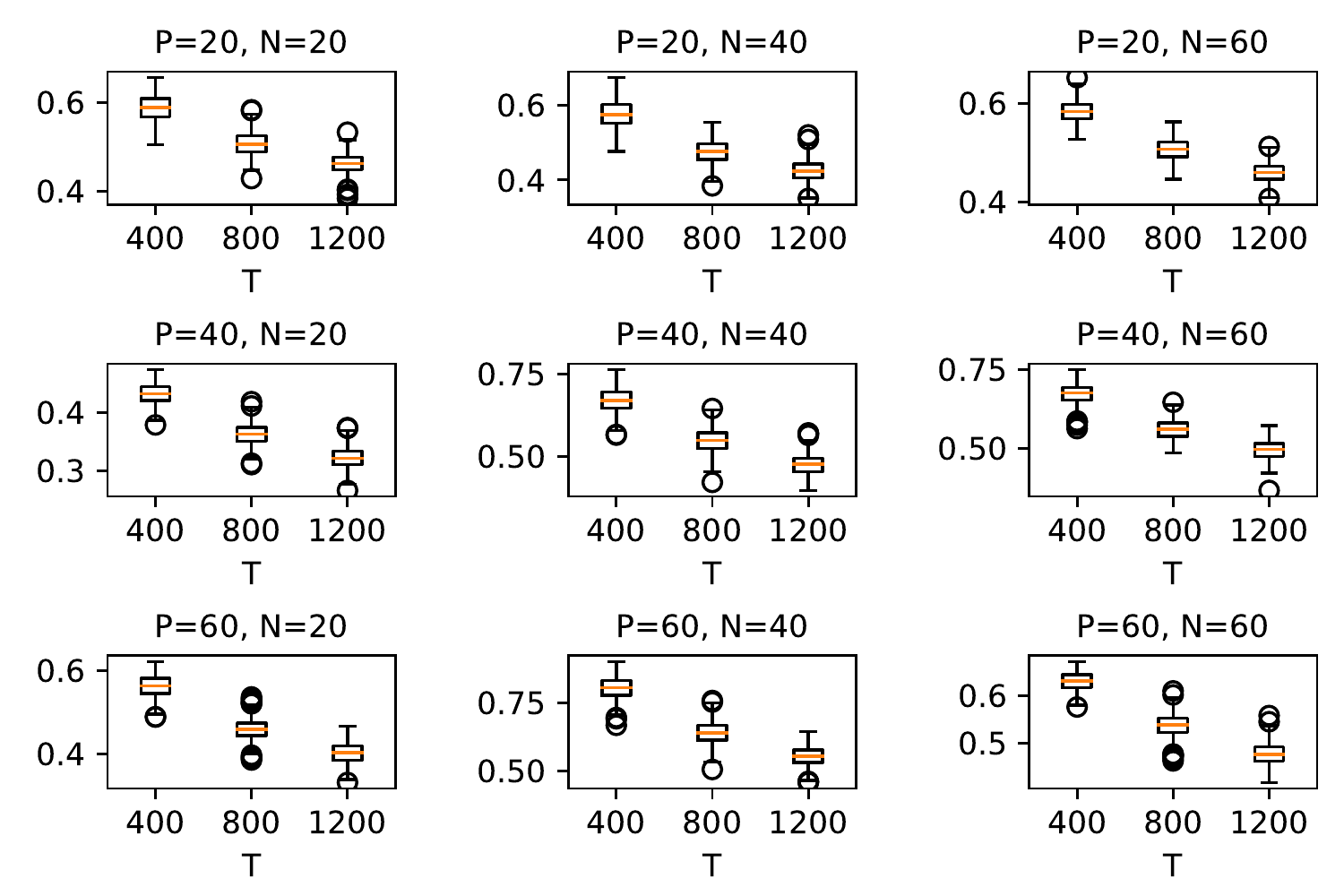}
			\caption{Boxplots for $\norm{\bA\bA' - \wh{\bA} \wh{\bA}'}_2$ of Example 2, where $\wh{\bA}$ is estimated by Algorithm~\ref{algo:IRRA}. For each pair of $(p,N)$, the sample sizes used are $T=400,800$ or 1200, and the number of repetitions is $500$.}
			\label{fig:error_A_ADMM}
		\end{figure}
		
		\begin{figure}
			\centering 
			\includegraphics[width=0.9\textwidth,height=2.5in]{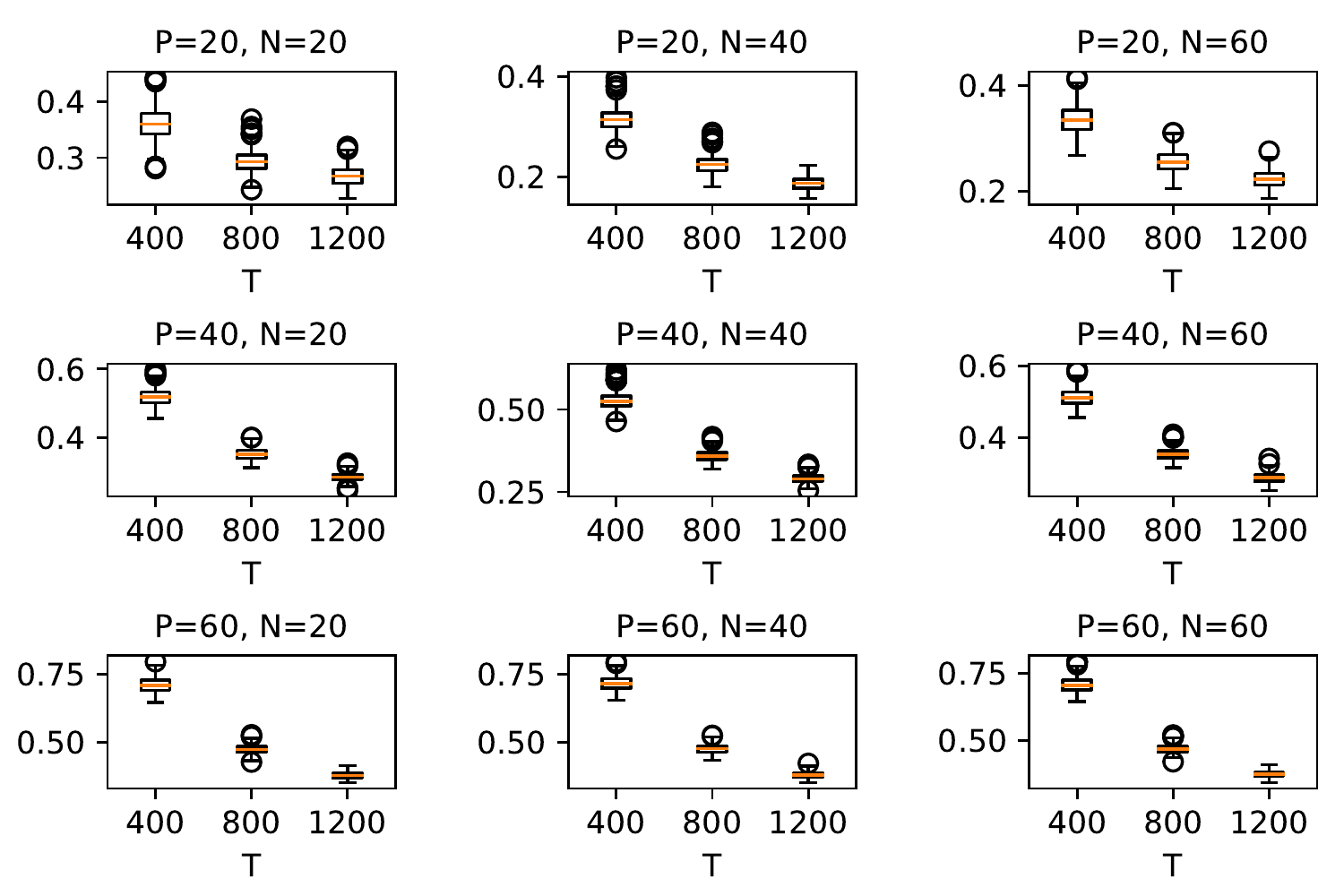}
			\caption{Boxplots for $\norm{\bPhi_1 - \wh\bPhi_1}_2$ in Example 2, where $\wh{\bPhi}$ is estimated by Algorithm~\ref{algo:IRRA}. For each pair of $(p,N)$, the sample sizes used are $T=400,800$ and $1200$, and the number of repetitions is $500$.}
			\label{fig:error_Phi1_ADMM}
		\end{figure}
		
		\begin{figure}
			\centering 
			\includegraphics[width=0.9\textwidth,height=2.5in]{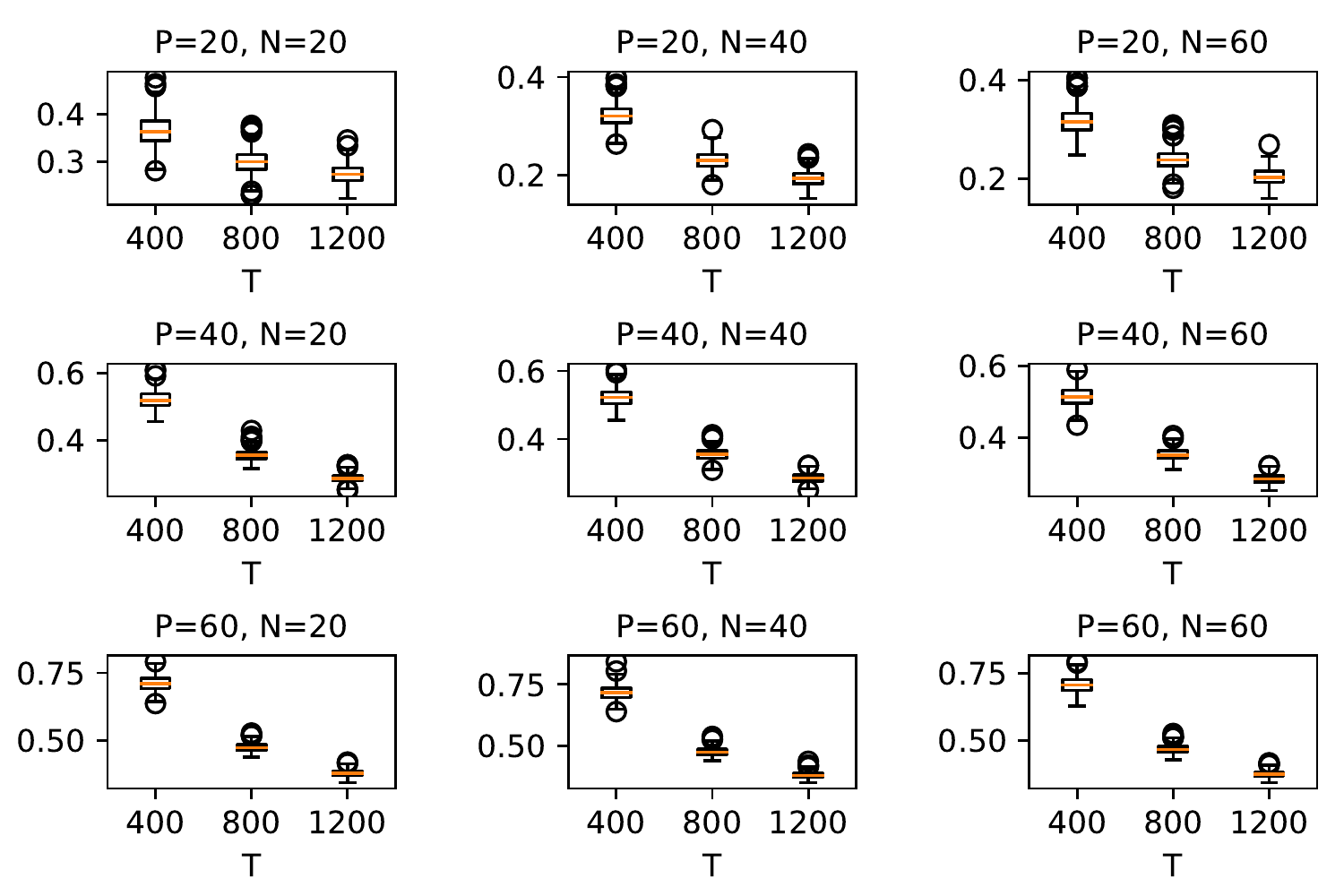}
			\caption{Boxplots for $\norm{\bPhi_2 - \wh\bPhi_2}_2$ in Example 2, where $\wh{\bPhi}$ is estimated by Algorithm~\ref{algo:IRRA}. For each pair of $(p,N)$, the sample sizes used are $T=400,800$ and $1200$, and the results are based on $500$ repetitions.}
			\label{fig:error_Phi2_ADMM}
		\end{figure}

		\section{Real Data Analysis}
		\label{sec:real_data_analysis}
		
		In this section, we apply the proposed method to predicting monthly stock returns. \textcite{welch2008comprehensive} examined the predictability of some macroeconomic variables to the equity premium, and concluded that the performance of the predictions,  both in-sample and out-of-sample, is poor and unstable. Using the same set of predictors, \textcite{koo2020high} exploited the cointegration relationship of the predictors, and showed that LASSO can improve the predictability of the macroeconomic variables in forecasting the equity premium of S\&P 500 index. We use an extended data set and conduct a forecasting experiment using the proposed method. Note that we predict the stock returns of a cross-section, instead of the equity premium of an individual stock or index.
		
		\subsection{Data and Empirical Strategy}
		
		Consider the monthly returns of selected stocks in the S\&P 500 index. Using the constituents of the index in January 2011 and the data structures in the CRSP (Center for Research in Security Prices) database, we select 79 stocks, which have no missing values during the time span from January 1960 to December 2019, as our sample. Therefore, we have 720 monthly observations of 79 $I(0)$ processes. We also collect the monthly macroeconomic variables  in \textcite{welch2008comprehensive} as the predictors. An updated version of the data can be downloaded from 
		Prof. Amit Goyal's personal website (\url{https://sites.google.com/view/agoyal145}), where we choose 13 macroeconomic  predictors as the $I(1)$ processes $\bx_t$ from December 1959 to November 2019. Therefore, we have $N=13$, $p=79$ and $T=720$ in this illustration.

		Table \ref{table:macro_stat} presents some descriptive statistics of the macro predictors, including their first-order sample autocorrelation coefficients $\rho(1)$ over the entire sample period. As shown in the table,  
		nine predictors have a first-order sample autocorrelation coefficient higher than 0.95, but four variables (inflation, long-term yield,  corporate bond returns, and stock variance) show little persistence. Therefore, we see that most of the variables are highly persistent and can be used as the $I(1)$ predictors in our model. 
		
		\begin{table}[htbp]
			\centering
			\caption{Descriptive statistics of the macroeconomic predictors, and their first-order sample autocorrelation coefficients over the entire sample period. The sample size is $T= 720$.}
			\small
			\begin{tabular}{lrrrrrrrr}
				\hline
				\multicolumn{1}{c}{variable}  & \multicolumn{1}{c}{mean} & \multicolumn{1}{c}{std} & \multicolumn{1}{c}{min} & \multicolumn{1}{c}{25\%}  & \multicolumn{1}{c}{50\%}  & \multicolumn{1}{c}{75\%}  & \multicolumn{1}{c}{max} & \multicolumn{1}{c}{$\rho(1)$}\\
				\hline
				D12   & 14.5619  & 13.4247  & 1.8667  & 3.6175  & 11.0988  & 19.5073  & 58.2406 & 0.9919\\
				E12   & 34.2835  & 34.3074  & 3.0300  & 8.1850  & 18.1334  & 51.0358  & 139.4700 & 0.9921\\
				b/m   & 0.4899  & 0.2565  & 0.1205  & 0.2872  & 0.4382  & 0.6395  & 1.2065 & 0.9937\\
				tbl   & 0.0454  & 0.0316  & 0.0001  & 0.0229  & 0.0462  & 0.0612  & 0.1630 & 0.9904\\
				AAA   & 0.0705  & 0.0263  & 0.0298  & 0.0488  & 0.0698  & 0.0856  & 0.1549 & 0.9943\\
				BAA   & 0.0806  & 0.0287  & 0.0387  & 0.0566  & 0.0789  & 0.0960  & 0.1718 & 0.9952\\
				lty   & 0.0631  & 0.0274  & 0.0163  & 0.0423  & 0.0598  & 0.0800  & 0.1482 & 0.9921\\
				ntis  & 0.0100  & 0.0199  & -0.0560  & -0.0022  & 0.0130  & 0.0245  & 0.0512 & 0.9809\\
				Rfree & 0.0037  & 0.0026  & 0.0000  & 0.0018  & 0.0037  & 0.0050  & 0.0135 & 0.9719\\
				infl  & 0.0030  & 0.0036  & -0.0192  & 0.0007  & 0.0029  & 0.0050  & 0.0181 & 0.5700\\
				ltr   & 0.0061  & 0.0291  & -0.1124  & -0.0104  & 0.0040  & 0.0228  & 0.1523 & 0.0310\\
				corpr & 0.0063  & 0.0257  & -0.0949  & -0.0071  & 0.0052  & 0.0191  & 0.1560 & 0.1077\\
				svar  & 0.0021  & 0.0043  & 0.0001  & 0.0006  & 0.0011  & 0.0021  & 0.0709 & 0.4717\\
				\hline
			\end{tabular}
			\label{table:macro_stat}
		\end{table}
		
		For the purpose of evaluating the forecasting performance of the proposed method, 
		we adopt the out-of-sample $R^2$ measure commonly used in the literature regarding the prediction of stock returns, see \textcite{gu2020empirical}. 
		At each time point, we define 
		the out-of-sample $R^2$ as
		\[
		R^2_{\mathrm{OOS}}(t) = 1 - \dfrac{\norm{\by_{t} - (\wh{\bA} \wh{\bz}_{t-1} + \wh{\bPhi} \bP_{t-1})}_2^2}{\norm{\by_{t}}_2^2},
		\]
		which differs from that in \textcite{gu2020empirical} by being a function of the time index $t$, which denotes the forecasting origin.
		
		Our empirical analysis works as follows. First, 
		we divide the time span of 60 years into two periods. The first period is the initial estimation period from January 1960 to December 2010, and the second one is the testing period from January 2011 to December 2019. Specifically, we conduct the empirical test according to the following procedure, which is similar to that commonly used in 
		the asset pricing literature. At the beginning, we use the data from January 1960 to December 2010 to estimate the coefficient matrices $\bA$ and $\mat{\Phi}$ of the model \eqref{DGP}, then predict the returns of the 79 stocks of January 2011 and calculate the out-of-sample $R^2$ of the prediction. We then add the returns  of January 2011 to the estimation period, and refit the model \eqref{DGP} with data from January 1960 to January 2011 to obtain updated coefficient matrices. The updated model is used to predict the returns of February 2011 with the model (\ref{DGP}) and to  calculate the out-of-sample $R^2$ again. We repeat this estimation-prediction process by 
		adding one-month returns to the estimation period in each iteration 
		until November 2019, which enables us to predict the returns for 
		December 2019. In addition, we choose the tuning parameters $\lambda_{\bA}$ and $\lambda_{\bPhi}$ based on the procedure described in Section~\ref{subsec:tuning} but letting $\lambda_{\bA}$ and $\lambda_{\bPhi}$ be proportional to $\sqrt{(p+N)/T}$ and   $\sqrt{\log(p)/T}$, respectively.  
		See Remark \ref{remark2}(ii).

		For comparison, we consider some alternative models commonly seen in the literature as benchmarks. The first benchmark is the naive VAR($d$) model, that is,
		\[
		\by_t = \bPsi \bP_{t-1}+\be_t, \ t=1,2,\ldots,T.
		\]
		For each series of $\by_t$ and the corresponding row of $\bPsi$, we can treat the above equation as a simple regression problem with $dp$ explanatory variables and $T$ observations. Therefore, we may use the LS method to estimate each row of $\bPsi$, then put them together to construct an estimator of $\bPsi$. We use the model in~\textcite{koo2020high} as another benchmark, where the tuning parameter $\lambda$ is fixed to $\frac{\log(p)}{10\sqrt{T}}$. Since the original model in~\textcite{koo2020high} is developed for predicting a scalar  time series, we apply their model with $13$ macroeconomic predictors described above to predict each stock separately, then stack the predictions together to calculate the out-of-sample $R^2$ for all $79$ stocks. Our final benchmark is the random walk model, in which we predict the returns of the next period using returns of the current period. For a more comprehensive comparison, we conduct the experiment for  $d=1,2$ and $3$, respectively, for the naive VAR, the RRSRA, and the IRRA models.

		\subsection{Prediction Performance of RRSRA}
		\label{subsec:pre_RRSRA}
		We evaluate the empirical performance of the RRSRA in this section. 
		To begin, Figure \ref{fig:r_hat} shows a time plot of 
		the estimated number of common trends by the method described in Section~\ref{subsec:determination_r} in the testing period. The figure shows that, except for the first three months of 2011 that may be affected by some economic crisis, the estimated number of common trends within $\bx_t$ is four, which is fairly stable over the entire test period. 
		Thus, we have nine cointegrating vectors to produce the stationary process $\wh{\bz}_t$ as a proxy of macroeconomic predictors. Before analyzing the forecasting performance of these estimated $\wh{\bz}_t$ variables, we take a look at the number of parameters to be estimated in the two coefficient matrices $\wh\bA$ and $\wh\bPhi$. Suppose that $\wh{r}=4$, then there are $9$ cointegrating vectors and hence, the matrix $\wh{\bA}$ has $79\times 9=711$ entries, and the matrix $\wh{\mat{\Phi}}$ has $79\times 79=6241$ entries to be estimated, both of which are  relatively large. Therefore, we expect that the dimensions of the two matrices can further be reduced to low-rank or sparse ones, which is commonly assumed in the literature to avoid over-fitting and to produce better forecasting performance.
		For this reason, we expect that the tuning parameters $\lambda_{\mat{A}}$ and $\lambda_{\mat{\Phi}}$ in our framework should be relatively large to guarantee that the dimensions can be reduced.
		
		\begin{figure}[htbp]
			\centering 
			\includegraphics[width=0.9\textwidth,height=2.5in]{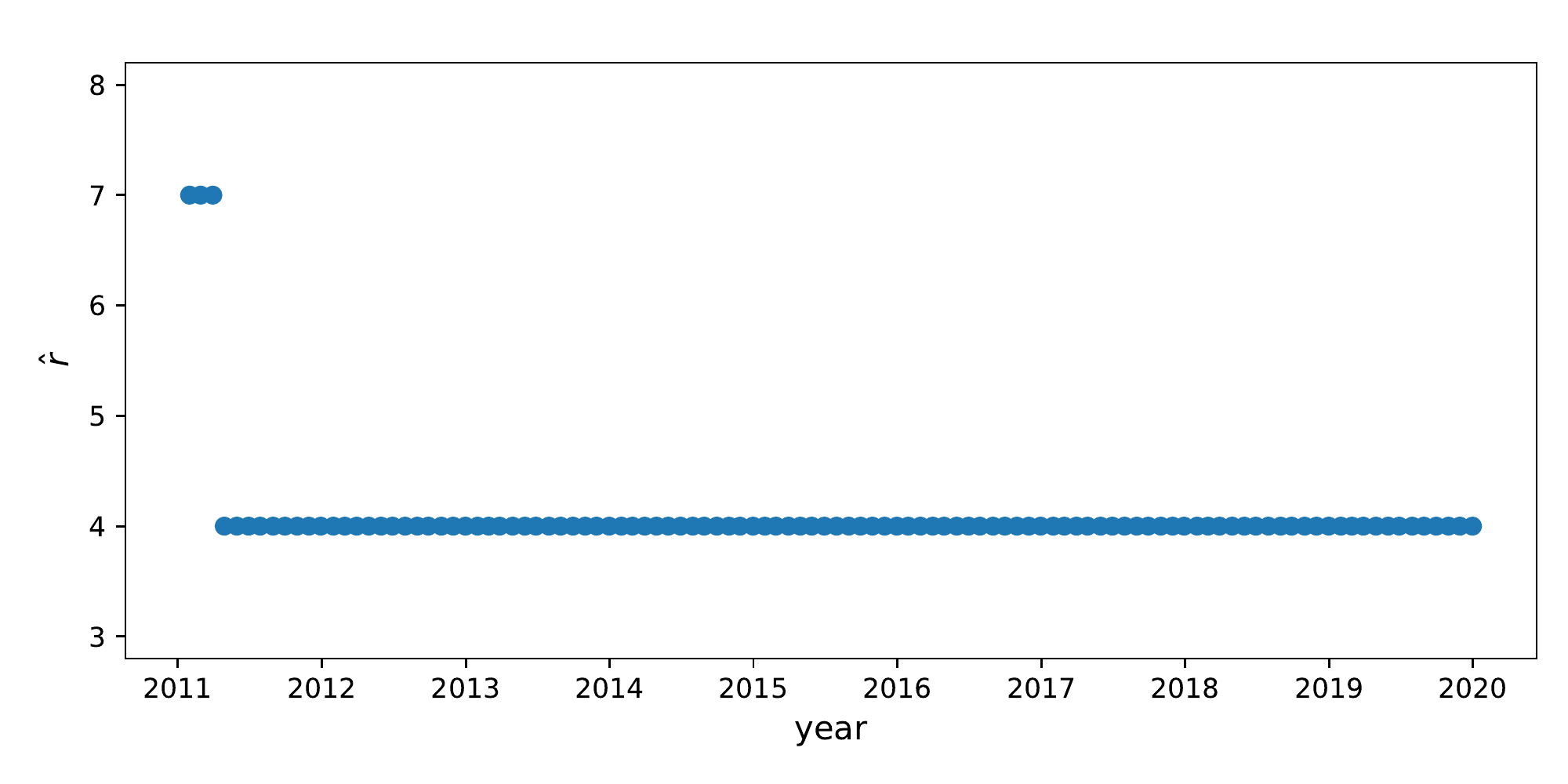}
			\caption{Time plot of $\wh{r}$ obtained using all the data before the corresponding time point on the horizontal axis.}
			\label{fig:r_hat}
		\end{figure}
		
		Figure \ref{fig:rank_sparse} shows the estimated rank of $\wh{\mat{A}}$ and the estimated number of non-zero entries of $\wh{\mat{\Phi}}$ for the proposed model with $d=1$ at each prediction time point. The average rank of $\wh{\mat{A}}$ is $1.97$ and the average number of non-zero entries of $\wh{\bPhi}$ is $5.88$. 
		Except for the first three months in 2011, the 
		estimated rank of $\wh\bA$ is $2$ in the estimation period. In addition,  $\wh{\bPhi}$ has at most $13$ non-zero entries at all time points, which is extremely small compared to $6241$ of the total number of entries. Overall, the proposed method provides an effective way to reduce the number of parameters and the dimension of the coefficient matrices.
		
		\begin{figure}[htbp]
			\centering 
			\includegraphics[width=0.9\textwidth,height=2.5in]{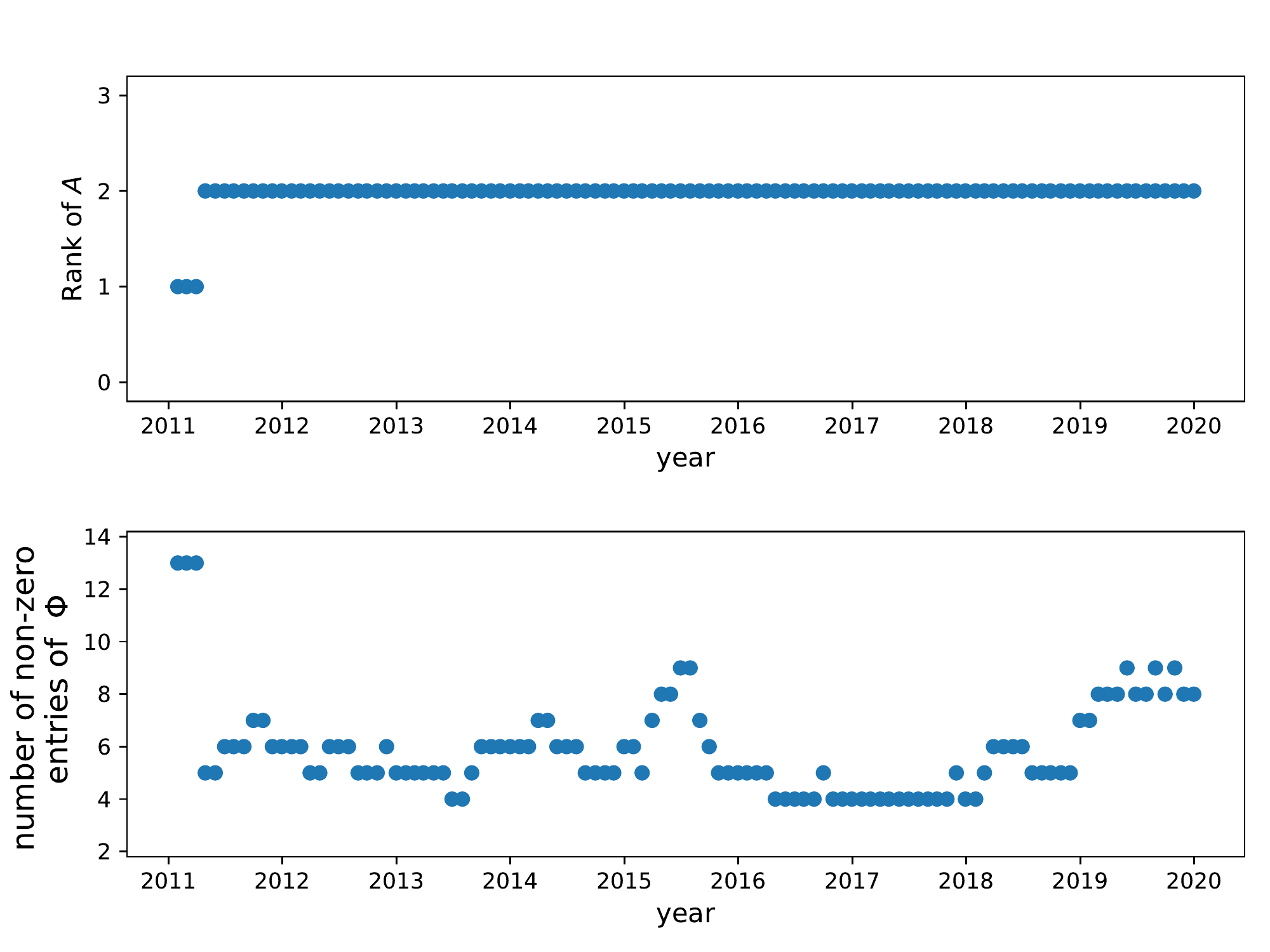}
			\caption{The estimated rank of $\wh\bA$ and the number of non-zero entries of $\wh\bPhi$. The two coefficient matrices are estimated with $d=1$ using all the data prior to the time point.}
			\label{fig:rank_sparse}
		\end{figure}

		Next we show the forecasting results in detail by following the  method described in Section~\ref{estimation_methodology} and the forecasting procedure mentioned above to evaluate the performance of different models. Table \ref{table:empirical_performance} reports the overall comparisons of our proposed method against the three benchmarks mentioned before, in terms of $R^2_{\text{OOS}}$. From Panels A and C of Table~\ref{table:empirical_performance}, we see that the proposed RRSRA substantially outperforms the naive VAR model (denoted by VAR($d$)), the method of \textcite{koo2020high} (denoted by Koo), and the random walk model (denoted by RW). The mean of out-of-sample $R^2$ of our method with $d=1$ is $0.91\%$ and the result is nearly the same for $d=2$ or $3$. We also note that our results are slightly better than those in \textcite{gu2020empirical}, where the highest monthly out-of-sample $R^2$ for all stocks is 0.40\% among all machine learning methods considered in their paper, and it is $0.70\%$ for the top $1,000$ stocks and $0.47\%$ for the bottom $1,000$ stocks by market values. In particular, the VAR model performs relatively poorly, as the out-of-sample $R^2$s with different lags all assume negative values with large magnitudes, which are $-22.54\%$, $-41.08\%$ and $-87.81\%$, respectively. One possible reason is that the number of parameters to be estimated in VAR models is significantly large and this often leads to severe over-fitting, which in turn produces high 
		variations in out-of-sample forecasting. When the lag order $d$ increases, the number of parameters also increases, so the performance would further deteriorate. For the model in \textcite{koo2020high}, the results in Table~\ref{table:empirical_performance} imply that the Koo method has limited predictive power when forecasting the returns of individual stocks. Finally, we see that the random walk model performs the worst. In summary, the proposed model has marked advantages in prediction 
		over the three benchmark models considered.
		
		\begin{table}[htbp]
			\centering
			\caption{Comparison of the RRSRA model, IRRA model and several benchmarks in terms of out-of-sample $R^2$. Panel A shows the results for RRSRA model with lag $d=1,2,3$. Panel B is the results for IRRA model, where we use both Algorithm~\ref{algo:IRRA} and ADMM method in estimation, and with $d=1,2,3$, respectively. Panel C reports several benchmark models, including the naive VAR model with lag $d=1,2,3$, the method in \textcite{koo2020high} denoted by Koo, and the random walk model denoted by RW. For each time point in the test sample, we calculate the value of $R^2_{\text{OOS}}$. The size of the test sample is $108-d+1$ for all models.}
			\begin{tabular}{lrrrrrrr}
				\hline
				& \multicolumn{7}{c}{Out-of-sample $R^2(t)$ (in percentage)}           \\ \cline{2-8} 
				& \multicolumn{1}{r}{mean} & std    & min     & 25\%    & 50\%    & 75\%   & max   \\ \hline
				\multicolumn{8}{l}{\textit{Panel A. Method~\eqref{solve_A_Phi} with different $d$}}\\
				RRSRA(1) & 0.91                     & 18.73  & -85.75  & -6.42   & 4.61    & 13.25  & 31.45 \\
				RRSRA(2) & 0.92                     & 18.70  & -85.76  & -6.38   & 4.66    & 13.20  & 31.53 \\
				RRSRA(3) & 0.90                     & 18.83  & -86.44  & -6.32   & 4.67    & 13.33  & 31.78 \\
				&&&&&&&\\
				\multicolumn{8}{l}{\makecell[l]{\textit{Panel B. Method~\eqref{eq:IRRR} with different $d$, fitted using both iterative method}\\ \textit{and ADMM method}}} \\
				Iterative(1) & 0.75 & 18.29 & -80.36 & -6.77 & 4.79 & 13.09 & 30.90 \\
				Iterative(2) & 0.66 & 18.39 & -76.71 & -7.01 & 3.96 & 13.18 & 29.99 \\
				Iterative(3) & 0.59 & 18.53 & -77.88 & -7.71 & 4.21 & 13.66 & 28.77 \\
				ADMM(1) & 0.80 & 18.18 & -79.28 & -6.57 & 4.83 & 13.11 & 30.61 \\
				ADMM(2) & 0.65 & 18.38 & -76.52 & -7.01 & 3.96 & 13.16 & 30.03 \\
				ADMM(3) & 0.60 & 18.53 & -77.87 & -7.67 & 4.11 & 13.66 & 29.13 \\
				&&&&&&&\\
				\multicolumn{8}{l}{\textit{Panel C. Benchmark models}}\\
				VAR(1)      & -22.54                   & 31.23  & -133.10 & -34.86  & -18.96  & 0.05   & 24.27 \\
				VAR(2)      & -41.08                   & 49.96  & -344.47 & -65.06  & -31.93  & -12.30 & 37.23 \\
				VAR(3)      & -87.81                   & 97.26  & -797.91 & -117.46 & -66.46  & -35.08 & 46.18 \\
				Koo         & -5.23                    & 21.67  & -91.40  & -17.98  & -0.53   & 11.32  & 26.52 \\
				RW          & -127.08                  & 117.45 & -665.03 & -163.63 & -111.06 & -49.64 & 37.24 \\ \hline
			\end{tabular}
			\label{table:empirical_performance}%
		\end{table}%

		To explore the in-sample goodness of fit, we apply all 
		entertained models except the random walk to the entire data set, and calculate the in-sample $R^2$ at each time point. The results over the time are shown in Table \ref{table:empi_IS_R2}. As  expected, the VAR model, which has many more degrees of freedom than the others, produces the highest in-sample $R^2$. Our model and the one of \textcite{koo2020high} provide a robust in-sample fit, and the result by the Koo method is only slightly worse than those of the proposed models.

		\begin{table}[htbp]
			\centering
			\caption{In-sample $R^2$ of method~(\ref{solve_A_Phi}), method~(\ref{eq:IRRR}) and some benchmarks. The three models in Panel A are the RRSRA model with lag $d=1,2,3$. Panel B shows the results for IRRA, where we use both iterative method and ADMM method in estimation with $d=1,2,3$, respectively. Panel C reports some benchmark models, that is, the naive VAR model with $d=1,2,3$, the method in \textcite{koo2020high} denoted by Koo, and the random walk model denoted by RW. We fit each model with the entire data set and obtain fitted values for returns of individual stocks, then calculate $R^2$ at each time point. The sample size is $719-d+1$ for all models.}
			\begin{tabular}{lrrrrrrr}
				\hline
				& \multicolumn{7}{c}{In-sample $R^2$ (in percentage)} \\
				\cline{2-8} & \multicolumn{1}{r}{mean}    & \multicolumn{1}{r}{std}    & \multicolumn{1}{r}{min}     & \multicolumn{1}{r}{25\%}    & \multicolumn{1}{r}{50\%}    & \multicolumn{1}{r}{75\%} & \multicolumn{1}{r}{max} \\
				\hline
				\multicolumn{8}{l}{\textit{Panel A. Method~\eqref{solve_A_Phi} with different $d$}} \\
				RRSRA(1) & 1.53                        & 16.02                      & -68.32                      & -9.36                       & 3.80                        & 13.40                      & 42.66                     \\
				RRSRA(2) & 1.52                        & 15.99                      & -68.44                      & -9.37                       & 3.77                        & 13.32                      & 42.47 \\
				RRSRA(3) & 1.55 & 16.04 & -68.63 & -9.29 & 3.75 & 13.42 & 42.63 \\
				&&&&&&&\\
				\multicolumn{8}{l}{\makecell[l]{\textit{Panel B. Method~\eqref{eq:IRRR} with different $d$, fitted using both iterative}\\ \textit{method and ADMM method}}} \\
				Iterative(1) & 1.61 & 16.00 & -68.16 & -8.95 & 3.85 & 13.25 & 43.03 \\
				Iterative(2) & 1.58 & 15.99 & -65.52 & -8.84 & 3.89 & 13.27 & 43.42 \\
				Iterative(3) & 1.71 & 16.03 & -65.41 & -9.10 & 3.81 & 13.40 & 42.95 \\
				ADMM(1) & 1.59 & 15.98 & -67.62 & -8.99 & 3.91 & 13.25 & 42.84 \\
				ADMM(2) & 1.58 & 15.97 & -65.14 & -8.87 & 3.92 & 13.26 & 43.34 \\
				ADMM(3) & 1.69 & 16.01 & -65.19 & -8.96 & 3.78 & 13.33 & 42.98 \\
				&&&&&&&\\
				\multicolumn{8}{l}{\textit{Panel C. Several benchmark models}}\\
				VAR(1)      & 8.10                        & 24.69                      & -137.87                     & -3.87                       & 11.85                       & 25.34                      & 69.86                     \\
				VAR(2)      & 15.96                       & 30.09                      & -121.38                     & 0.80                        & 20.39                       & 37.38                      & 84.32                     \\
				VAR(3)      & 24.69                       & 35.48                      & -303.66                     & 9.14                        & 30.57                       & 48.35                      & 89.87                     \\
				Koo         & -1.13                       & 16.30                      & -48.29                      & -14.18                      & 2.03                        & 12.90 & 34.17 \\
				\hline
			\end{tabular}
			\label{table:empi_IS_R2}
		\end{table}%

		To check whether our model outperforms the benchmarks uniformly over the in-sample and the out-of-sample periods, we plot the in-sample $R^2$s and out-of-sample ones of the RRSRA$(1)$ model in Figures~\ref{fig:ISR2} and \ref{fig:OOSR2}, respectively, where the time index is on the horizontal axis. For a better illustration, the points in  Figure~\ref{fig:ISR2} are the in-sample $R^2$s based on the data of each year from 1960 to 2019. In both figures, we also plot the VAR(1) as a benchmark. An additional plot of the random walk model is also included in Figure~\ref{fig:OOSR2} as another benchmark. Because the results produced by the Koo method  in \textcite{koo2020high} are very close to ours, they are omitted.  From Figure~\ref{fig:ISR2}, we see that our method fits the data relatively poorly 
		compared to the VAR model in most years according to the in-sample $R^2$. This is understandable since the VAR model fits the data via the LS method to minimize the squared distance between the fitted values and the true ones, while our method adopts regularization, which often introduces some in-sample biases in order to provide more stable predictions in out-of-samples. Furthermore, Figure~\ref{fig:OOSR2} shows that our method produces more robust predictions than the VAR and the random walk model, and outperforms them over most of the time points 
		based on the out-of-sample $R^2$. This illustrates the predictive advantages of using the proposed method.
		
		\begin{figure}[htbp]
			\centering 
			\includegraphics[width=0.9\textwidth,height=2.5in]{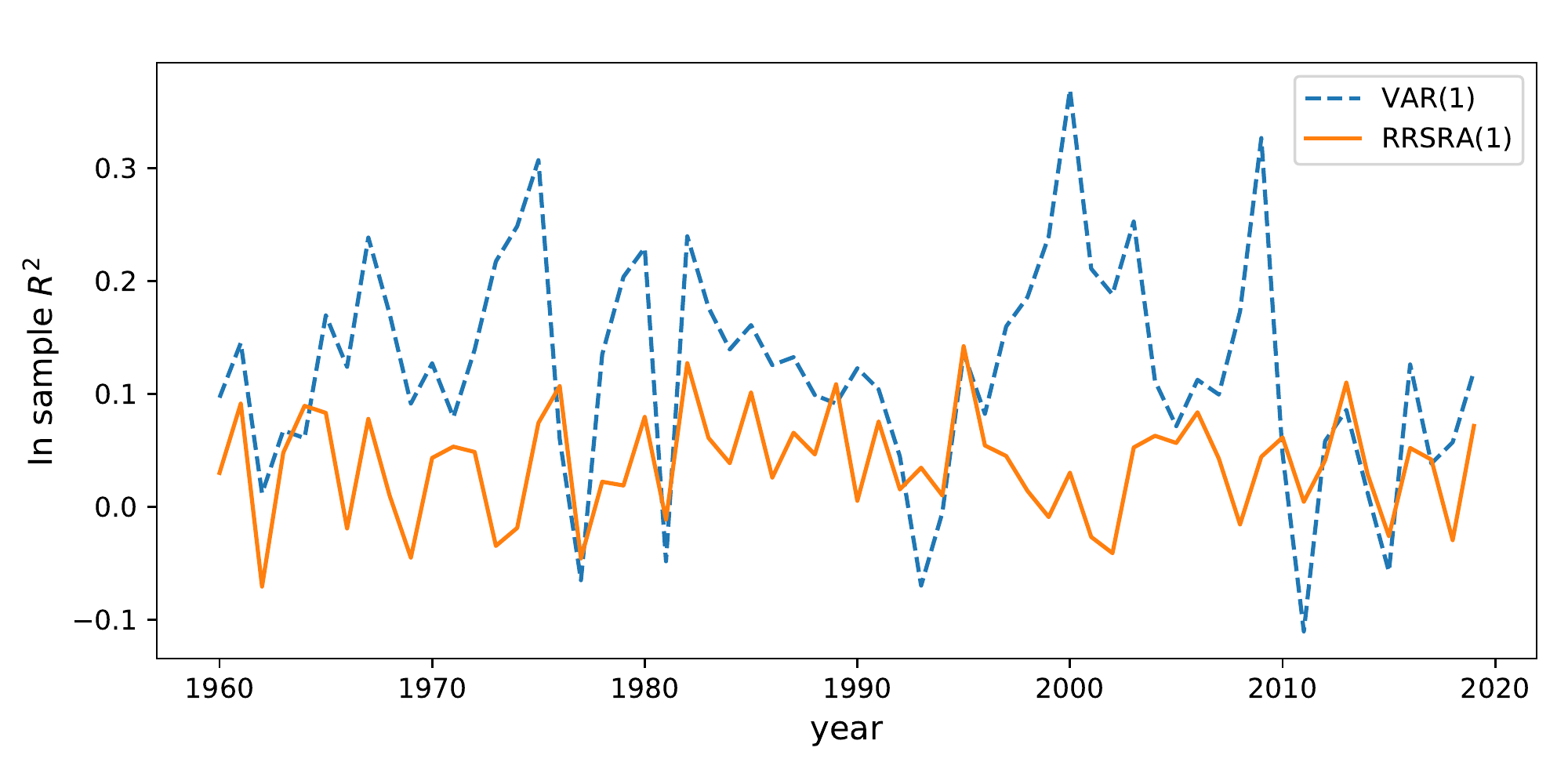}
			\caption{In-sample $R^2$ for the VAR(1) model and our method with $d=1$. We use the entire data set in estimation and calculate the in-sample $R^2$ using the data of each year from 1960 to 2019.}
			\label{fig:ISR2}
		\end{figure}
		
		\begin{figure}[htbp]
			\centering 
			\includegraphics[width=0.9\textwidth,height=2.5in]{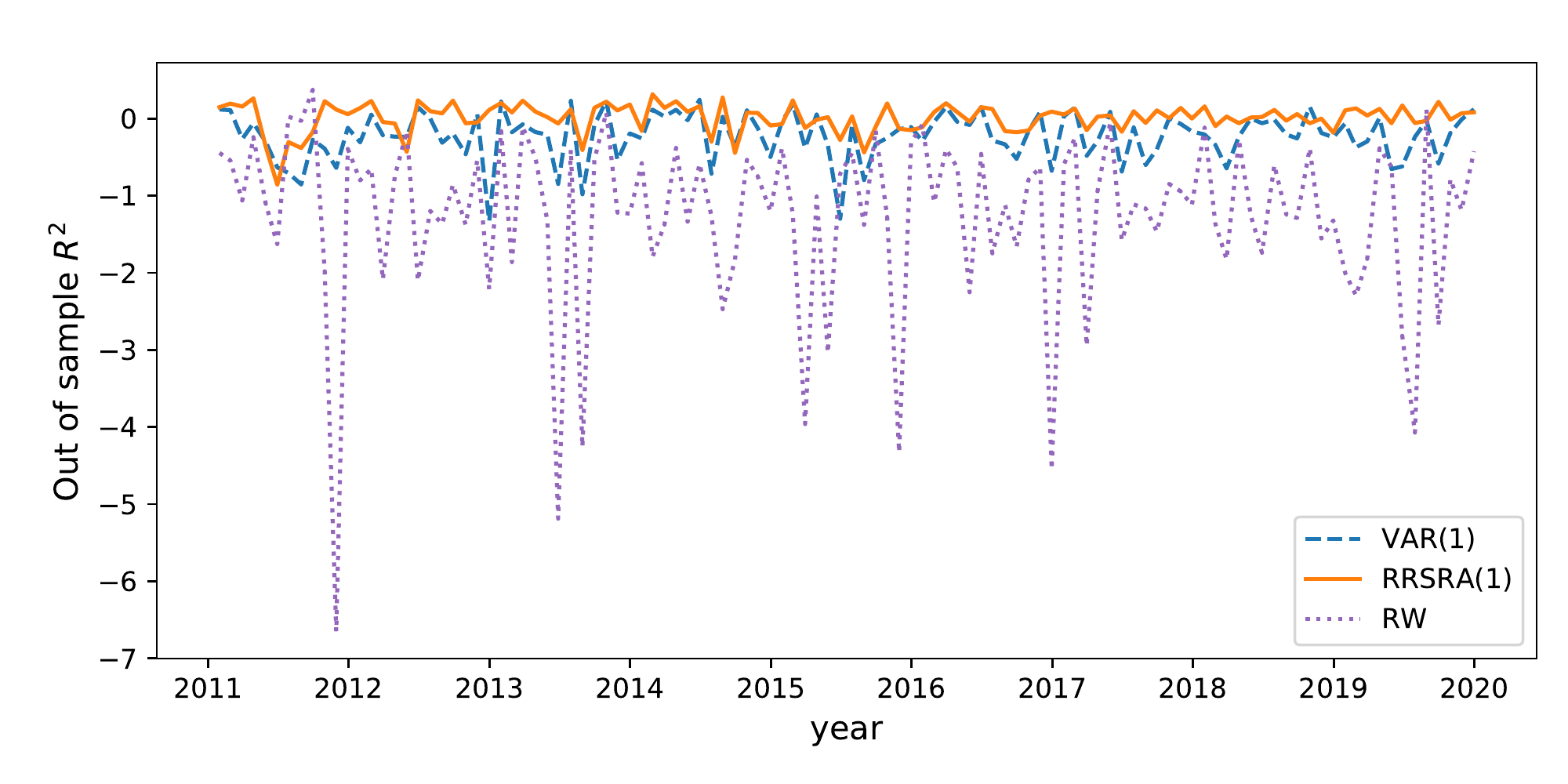}
			\caption{Out-of-sample $R^2$ for the proposed model with $d=1$, the VAR(1) model and the random walk model. For each time point in the test period, we fit each model using the data prior to that time point, make predictions for the returns, and calculate the out-of-sample $R^2$ of the predictions.}
			\label{fig:OOSR2}
		\end{figure}
		
		\subsection{Prediction Performance of IRRA}
		
		In this section, we evaluate the predictive performance of IRRA models described in Section~\ref{sec:IRRA}. The procedure of estimating the factor model \eqref{factor_structure} and obtaining $\bz_t$ are exactly the same as those in Section~\ref{subsec:pre_RRSRA}. With the estimated $\wh\bz_t$, we fit the data via \eqref{eq:IRRR} using both iterative method and ADMM method, and evaluate its  predictive performance. We expect that the results of IRRA are close to those of RRSRA, because the tuning parameter $\lambda_{\bPhi}$ selected by a grid search is relatively large in both models. When the tuning parameter $\lambda_{\bPhi}$ in both methods tends to infinity, the  estimated coefficients obtained by the two algorithms tend to be the same.
		
		We first examine the estimated coefficient matrices. Figure~\ref{fig:rank_rank} shows the rank of $\wh\bA$ and $\wh\bPhi$ estimated by Algorithm~\ref{algo:IRRA} in the case of $d=1$ at each time point. We find that the rank of $\wh\bA$ is reduced to $1$ or $2$ over the entire time horizon, which is the same as that of the RRSRA, implying  that the efficient cointegration rank is low in this particular application. In addition, the rank of $\wh\bPhi$ is also $1$ or $2$ over time. 
		
		\begin{figure}[htbp]
			\centering
			\includegraphics[width=0.9\textwidth,height=2.5in]{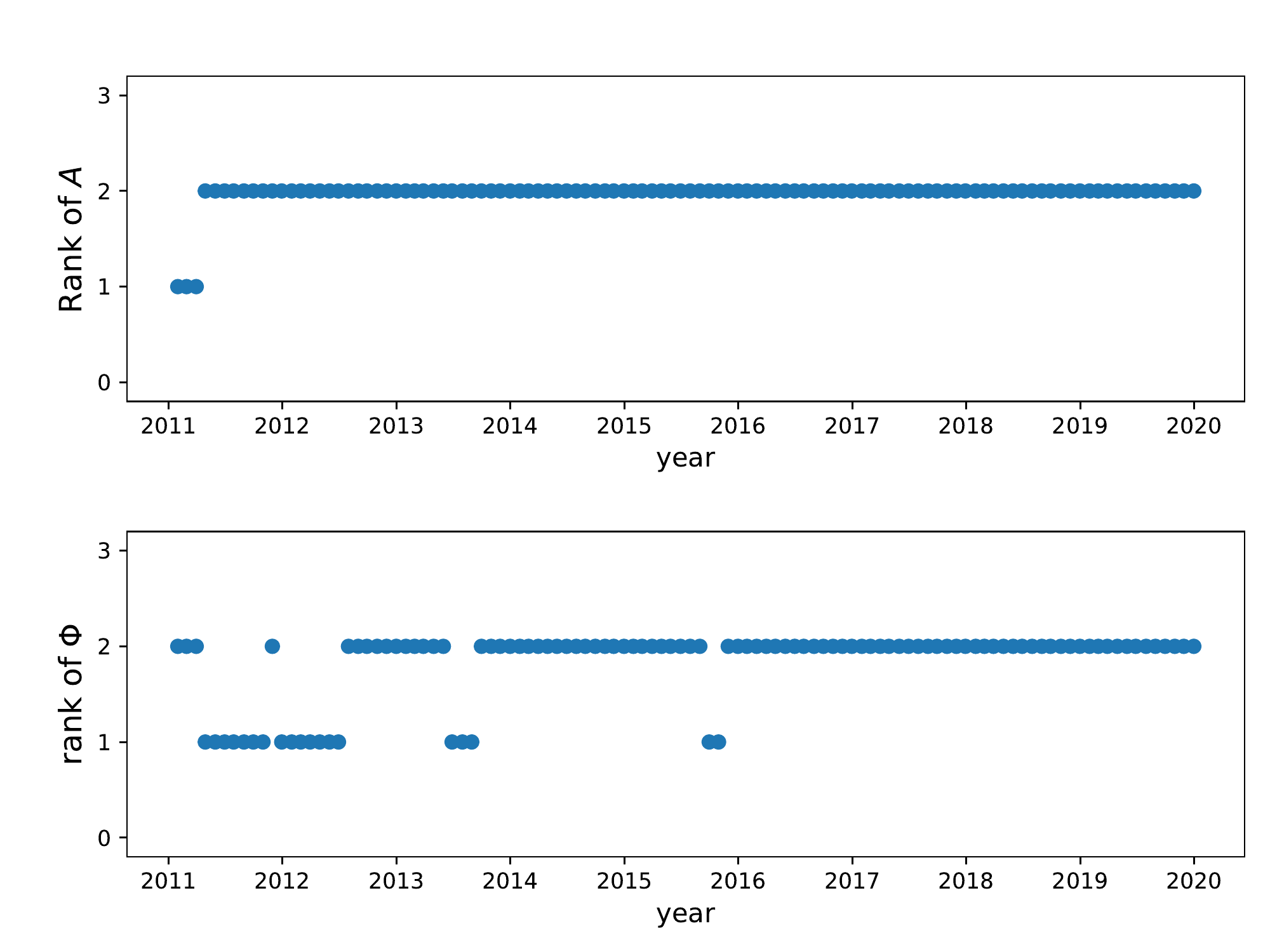}
			\caption{The estimated rank of $\wh\bA$ and the rank of $\wh\bPhi$. The two coefficient matrices are estimated by Algorithm~\ref{algo:IRRA} with $d=1$ at each time point using all the data prior to it.}
			\label{fig:rank_rank}
		\end{figure}
		
		Panel B of Table~\ref{table:empirical_performance} shows  the predictive performance of \eqref{eq:IRRR}, where the coefficients are estimated by both Algorithm~\ref{algo:IRRA} and the ADMM method with $d=1,2,3$, respectively. All six results are positive and  close to each other. They are also close to, but a little worse than, those of RRSRA. One possible reason is that both IRRA and RRSRA are constrained regressions but the latter one produces sparse solutions and reduces the model complexity more substantially compared to the low-rank structures. Similarly to the conclusion 
		of the RRSRA procedure, the $R^2_{\text{OOS}}$ results of \eqref{eq:IRRR} also outperform those of the benchmarks.
		
		Finally, we see that the performance of IRRA is close to that of RRSRA not only with respect to the out-of-sample $R^2$, but also with respect to the in-sample $R^2$; see  Tables~\ref{table:empirical_performance} and~\ref{table:empi_IS_R2}. Panel B of Table~\ref{table:empi_IS_R2} shows the in-sample $R^2$ results for IRRA, with all six estimation settings. Once again, we find that the results are close to those in Panel A, but IRRA fits the data slightly better, which may be a consequence of higher degrees of freedom in IRRA.

		\section{Concluding Remarks}
		\label{sec:conclusion}
		Finding proper cointegration relationships is an important topic in Econometrics and Statistics, yet the interpretation of 
		cointegrating structures might become  
		complicated if the dimension of the system  under study 
		is high. This paper introduced the concept of {\it effective cointegration rank} and considered a new method to identify the important cointegration relationships among a high-dimensional $I(1)$ series from a predictive perspective. In a nutshell, the effective cointegration rank is the number of cointegrating relationships that can produce useful  predictors in a given forecasting application.  The proposed method consists of a two-step estimation procedure, where we first use  the Principal Component Analysis to estimate the common stochastic trends of the $I(1)$ series and to identify all possible cointegrating vectors. We then employ all stationary series obtained via the cointegrating vectors and some lagged values of dependent variables to form predictors of the second-step estimation. A reduced-rank regression technique is applied to the co-integrated predictors and the dimension of relevant cointegrating vectors is defined as the effective cointegration rank. We also applied the LASSO penalty or reduced rank constraints to the coefficients of the lagged variables in the second step, and an iterative procedure is proposed to estimate the unknown coefficients. 
		
		Our proposed method has a wide range of applications in many scientific areas, including Economics, Finance, and Environmental studies, because it is common in these 
		areas to use nonstationary variables or factors to predict stationary series in empirical applications. We applied 
		the proposed method to the problem of predicting cross-sectional stock returns, and  illustrated 
		clearly the predictive advantages of the proposed 
		procedure over 
		some commonly used benchmarks available in the literature. 

		\printbibliography
		
		\newpage

		\begin{appendices}
			\newrefsection

\setcounter{lemma}{0}
\renewcommand{\thelemma}{A\arabic{lemma}}

In this supplement, we provide proofs 
of all theorems stated in Section~\ref{sec:theory} 
of the main article. We use $c$ or $C$ to denote a generic positive constant, 
its value may change for different places.

\section{Proof of Theorems}

\subsection{Proofs of Theorem~\ref{loading_estimator} and \ref{thm:r_consistency}}

To begin, we first introduce a useful lemma, which is commonly seen in matrix perturbation theory. See 
\textcite{golub2013matrix} (Theorem 8.1.10),  \textcite{johnstone2009consistency}, and \textcite{lam2011estimation}, among others.

\begin{lemma}
	\label{lemma:orthogonal_basis}
	Suppose $\bA$ and $\bA+\bE$ are $n\times n$ symmetric matrices, and $\bQ=[\bQ_1\  \bQ_2],$ 
	with $\bQ_1\in\mathbb{R}^{n\times r}$ and $\bQ_2\in\mathbb{R}^{n\times (n-r)}$, is an $n\times n$ orthogonal matrix such that $\range(\bQ_1)$ is an invariant subspace for $\bA$. Partition the matrices $\bQ'\bA\bQ$ and $\bQ'\bE\bQ$  as follows:
\[	
	\bQ' \bA\bQ = \begin{bmatrix}
	\bQ_1' \bA \bQ_1 \ &\mat{0}\\
	\mat{0} \ &\bQ_2' \bA \bQ_2
	\end{bmatrix} \quad \text{and} \quad \bQ'\bE\bQ = \begin{bmatrix}
	\bQ_1' \bE \bQ_1 \ & \bQ_1' \bE \bQ_2\\
	\bQ_2' \bE \bQ_1 \ &\bQ_2' \bE \bQ_2
	\end{bmatrix}.
\]

	If $\operatorname{sep}(\bQ_1' \bA \bQ_1,\bQ_2' \bA \bQ_2) = \min_{\mu\in\lambda(\bQ_1' \bA \bQ_1), \nu\in\lambda(\bQ_2' \bA \bQ_2)} \vert \mu - \nu \vert > 0$,  where $\lambda(\mat{M})$ denotes the set of eigenvalues of matrix $\mat{M}$, and 
	$$
	\FNorm{\bE} \leq \dfrac{1}{5}\operatorname{sep}(\bQ_1' \bA \bQ_1,\bQ_2' \bA \bQ_2),
	$$
	then there exists a matrix $\bP\in \mathbb{R}^{(n-r)\times r}$ with
	$$
	\FNorm{\bP} \leq \dfrac{4}{\operatorname{sep}(\bQ_1' \bA \bQ_1,\bQ_2' \bA \bQ_2)} \FNorm{\bQ_1' \bE \bQ_2}
	$$
	such that the columns of $\wh{\bQ}_1=(\bQ_1+\bQ_2 \bP)(\mat{I}_r+\bP' \bP)^{-1/2}$ define an orthonormal basis for a subspace that is invariant for $\bA+\bE$.
	
\end{lemma}

\begin{proof}[Proof of Theorem \ref{loading_estimator}]
	From \eqref{factor_structure}, we have the identity
	\[
	\wh{\mat{\Sigma}}_{\vect{x}} = \bB\wh{\mat{\Sigma}}_{\vect{f}}\bB'
	 +\bB\wh{\mat{\Sigma}}_{\vect{f}\bm{\varepsilon}}+\wh{\mat{\Sigma}}_{\bm{\varepsilon} \vect{f}}\bB' + \wh{\mat{\Sigma}}_{\bm{\varepsilon}},
	\]
	where $\wh{\mat{\Sigma}}_{\vect{x}} =T^{-1} \sum_{t=1}^{T} \vect{x}_t \vect{x}_t'$, $\wh{\mat{\Sigma}}_{\vect{f}} =T^{-1}  \sum_{t=1}^{T} \vect{f}_t \vect{f}_t'$, $\wh{\mat{\Sigma}}_{\vect{f}\bm{\varepsilon}} =T^{-1}  \sum_{t=1}^{T} \vect{f}_t \bm{\varepsilon}_t'$, $\wh{\mat{\Sigma}}_{\bm{\varepsilon} \vect{f}} =T^{-1}  \sum_{t=1}^{T}\bm{\varepsilon}_t \vect{f}_t'$, $\wh{\mat{\Sigma}}_{\bm{\varepsilon}} =T^{-1}  \sum_{t=1}^{T} \bm{\varepsilon}_t \bm{\varepsilon}_t'$. The sample means $\bar{\vect{x}}$, $\bar{\vect{f}}$, and $\bar{\bm{\varepsilon}}$ are set to ${\bf 0}$, because we assume the data $\bX_t$ are 
	properly centered in advance.
	
	By Assumption \ref{assumption:dependence}, we have
	\begin{align*}
	\norm{\wh{\mat{\Sigma}}_{\vect{x}} - \bB\wh{\mat{\Sigma}}_{\vect{f}}\bB'}_2 \leq& 2\norm{\bB}_2\norm{\wh{\mat{\Sigma}}_{\vect{f}\bm{\varepsilon}}}_2 +\norm{\wh{\mat{\Sigma}}_{\bm{\varepsilon}}-{\mat{\Sigma}}_{\bm{\varepsilon}}}_2+\norm{{\mat{\Sigma}}_{\bm{\varepsilon}}}_2\\
	=& O_p(N)+O_p(NT^{-1/2})+O_p(1)\\
	=& O_p(N),
	\end{align*}
	
	Because $[\bB\ \bB_c]$ is an orthogonal matrix, we have
	\begin{gather*}
	\begin{bmatrix}
	\bB'\\
	\bB_c'
	\end{bmatrix} \left(\bB \wh{\mat{\Sigma}}_{\vect{f}}\bB'\right) \begin{bmatrix}
	\bB \  \bB_c
	\end{bmatrix} = \begin{bmatrix}
	\wh{\mat{\Sigma}}_{\vect{f}} & \mat{0}\\
	\mat{0} & \mat{0}
	\end{bmatrix},\\
	\begin{bmatrix}
	\bB'\\
	\bB_c'
	\end{bmatrix} \left(\wh{\mat{\Sigma}}_{\vect{x}} - \bB \wh{\mat{\Sigma}}_{\vect{f}}\bB'\right) \begin{bmatrix}
	\bB \  \bB_c
	\end{bmatrix} = \begin{bmatrix}
	\bB'\wh{\mat{\Sigma}}_{\vect{x}}\bB-\wh{\mat{\Sigma}}_{\vect{f}} & \bB' \wh{\mat{\Sigma}}_{\vect{x}}\bB_c\\
	\bB_c'\wh{\mat{\Sigma}}_{\vect{x}}\bB & \bB_c'\wh{\mat{\Sigma}}_{\vect{x}}\bB_c
	\end{bmatrix}.
	\end{gather*}
	
 
 By Theorem 1 in \textcite{pena2006nonstationary} and Assumptions \ref{assumption:convergence} and \ref{assumption:dependence}, we can show that 
  $\operatorname{sep}(\wh{\mat{\Sigma}}_{\vect{f}})=\lambda_r(\wh{\mat{\Sigma}}_{\vect{f}})>CNT$ almost surely for some constant $C>0$. By Lemma \ref{lemma:orthogonal_basis}, there is an $(N-r)\times r$ matrix $\bP$ with an upper bounded $\norm{\bP}_2$, such that $\wh{\bB}=(\bB+\bB_c \bP)(\mat{I}_r+\bP'\bP)^{-1/2}$ defines an orthonormal basis for a subspace that is invariant for $\wh{\mat{\Sigma}}_{\vect{x}}$. Therefore,
	\begin{align*}
	\norm{\wh{\bB} - \bB}_2 \leq& \norm{\bB\left(\mat{I}_r - (\mat{I}_r+\bP'\bP)^{-1/2}\right)}_2 + \norm{\bB_c \bP (\mat{I}_r+\bP'\bP)^{-1/2}}_2\\
	\leq& 2\norm{\bP}_2\\
	\leq& \dfrac{8}{\lambda_r(\wh{\mat{\Sigma}}_{\vect{f}})} \norm{\bB' (\wh{\mat{\Sigma}}_{\vect{x}}-\bB\wh{\mat{\Sigma}}_{\vect{f}}\bB')}_2\\
	=& O_p(T^{-1}).
	\end{align*}

A similar result also holds for $\norm{\wh{\bB}_c - \bB_c}_2$ as we can 
	exchange the position of $\bB$ and $\bB_c$ in the orthogonal matrix $[\bB,\bB_c]$, and apply the same argument as above again. 
	
		Consequently, we have
 \begin{align*}
     \norm{\bB\bB'-\wh{\bB} \wh{\bB}'}_2 =& \norm{(\bB - \wh{\bB}) \bB' + \wh{\bB} (\bB - \wh{\bB})'}_2\\
        \leq& \norm{\bB-\wh{\bB}}_2 (\norm{\bB}_2 + \norm{\wh{\bB}}_2)\\
        =& O_p(T^{-1}).
 \end{align*}
 
 {Furthermore, from a least-squares perspective in  \eqref{solve_factor}, the factors are estimated as $\wh{\bff}_t=\wh{\bB}'\bx_t=\wh{\bB}'(\bB\bff_t+\bve)$. Then, 
 \begin{align*}
 \norm{\wh{\bB} \wh{\bff}_t - \bB \bff_t}_2 =& \norm{\wh{\bB}\wh{\bB}'\bB \bff_t + \wh{\bB}\wh{\bB}'\bm{\varepsilon}_t - \bB \bff_t}_2\\
 \leq& \norm{\wh{\bB}\wh{\bB}'(\bB-\wh{\bB}) \bff_t}_2 + \norm{(\wh{\bB}-\bB)\bff_t}_2 + \norm{\wh{\bB}\wh{\bB}' \bm{\varepsilon}_t}_2\\
 \leq& 2 \norm{(\wh{\bB}-\bB)\bff_t}_2 + \norm{\wh{\bB}' \bm{\varepsilon}_t}_2\\
 =& O_p(\sqrt{N/T}) + O_p(1),
 \end{align*}
 where the last line follows from the fact that $\norm{\bff_t}_2=\norm{\sum_{s=1}^{t}\bu_t}_2=O_p(\sqrt{NT})$ and $\wh{\bB}'\bm{\varepsilon}$ is an $r$-dimensional random vector with finite variance. Therefore, 
 \[
 N^{-1/2} \norm{\wh{\bB} \wh{\bff}_t - \bB \bff_t}_2 = O_p(N^{-1/2}+T^{-1/2}).
 \]
 }
 This completes the proof.
\end{proof}


\begin{proof}[\bf Proof of Theorem \ref{thm:r_consistency}] For any column vector $\wh\bb$ in $\wh\bB_c$, denote its corresponding true one by $\bb$. Then,  
\[\wh{\bb}'\bX_t=\wh\bb'\bB\bff_t+\wh\bb'\bve_t=(\wh\bb-\bb)'\bB\bff_t+(\wh\bb-\bb)'\bve_t+\bb'\bve_t.\]
By a similar argument as the proof of theorem 4 in \textcite{gao2021modeling_IJF}, the autocorrelations of $\wh\bb\bX_t$ will only depend on those of the third terms if the magnitudes of the first two terms are asymptotically negligible. Therefore, we only need to show 
\[
   \max_{1\leq t\leq T} |(\wh\bb-\bb)'\bB\bff_t|=o_P(1),
\]
    as the second term is obviously dominated by the third one according to Theorem 1. Note that $\bff_t$ has an additional strength of order $\sqrt{N}$, and Model (\ref{f_process}) implies that
    \[\bff_t=\sum_{i=1}^t\bu_i,\]
    is a partial sum of weakly dependent variables.
    Under Assumptions \ref{assumption:mixing}--\ref{assumption:dependence}, we see that the conditions for Theorem 1 of \textcite{merlevede2011bernstein} hold. By aforementioned Theorem 1 
    or the proof of Theorem 2 in \textcite{gao2021modeling_IJF}, we can show that
    \begin{align*}
        \max_{1\leq t\leq T} |(\wh\bb-\bb)'\bB\bff_t| \leq& \|\wh\bb-\bb\|_2\max_{1\leq t\leq T}\|\bB\bff_t\|_2\\
        \leq& CT^{-1}\sqrt{N} T^{1/2}\log(T)\\
        \leq& CN^{1/2}\log(T)T^{-1/2}.
    \end{align*}
    Therefore, it suffices to require $N^{1/2}\log(T)T^{-1/2}=o(1)$. This completes the proof. 
\end{proof}
    
\section{Proof of Theorem \ref{RRR_LASSO}}
\label{ssec:prf_RRR_LASSO}

Recall the SVD of any $m\times n $ matrix $\bTheta$ in \eqref{SVD}, the subspaces in \eqref{subspaces} and the decomposition of any $m\times n$ matrix $\bM$ in \eqref{decom12}. When $\bM=\bTheta$, we obviously have $\norm{\bM}_* = \norm{\bM_1}_* + \norm{\bM_2}_*$. In addition, for the decomposition of a general matrix $\bM$ that may be different from $\bTheta$, we introduce the following lemma.

\begin{lemma}
	\label{lemma:decomposition_of_nuclear_norm}
	Given the SVD of $\mat{\Theta}$ in \eqref{SVD}, for any matrix $\mat{M}\in\mathbb{R}^{m\times n}$, the decomposition
 \[
\norm{\bTheta_1 + \bM_2}_* = \norm{\bTheta_1}_* + \norm{\bM_2}_*
 \]
	holds.
\end{lemma}

\begin{proof}
	Given the SVD of $\mat{\Theta}$ in \eqref{SVD}, we have
	\begin{align*}
	\norm{\bTheta_1 + \bM_2}_* =& \lrnorm{\begin{bmatrix}
		\mat{D}_k & \mat{0}\\
		\mat{0} & \mat{U}_{k,c}' \mat{M} \mat{V}_{k,c}
		\end{bmatrix}}_*\\
	=& \lrnorm{\begin{bmatrix}
		\mat{D}_k
		\end{bmatrix}}_* + \lrnorm{\begin{bmatrix}
		 \mat{U}_{k,c}' \mat{M} \mat{V}_{k,c}
		\end{bmatrix}}_*\\
	=& \norm{\bTheta_1}_* + \norm{\bM_2}_*.\qedhere
	\end{align*}
\end{proof}



\begin{lemma}
	\label{lemma:sample_RSC}
	Let the conditions of Theorem \ref{RRR_LASSO} hold. As $N,p,T\to\infty$, 
	\[
	\dfrac{1}{2T}\FNorm{\bDelta_{\bA}\wh{\mat{Z}}+\bDelta_{\mat{\Phi}} \bP}^2 \geq  C_1 \kappa_1 \FNorm{\bDelta}^2-C_2\tau_T\Psi^2(\bDelta),
	\]
	with probability tending to $1$, where $\Psi$ is defined in \eqref{psi:n} and $\bDelta=[\bDelta_{\bA},\bDelta_{\bPhi}]$.
\end{lemma}

\begin{proof}
	Note that
	\begin{align*}
	\FNorm{\bDelta_{\bA}{\mat{Z}}+\bDelta_{\mat{\Phi}} \bP}^2
	=& \FNorm{\bDelta_{\bA}\wh\bZ+\bDelta_{\mat{\Phi}} \bP+\bDelta_{\bA}(\bZ-\wh\bZ) }^2\\
	\leq& 2\left( \FNorm{\bDelta_{\bA}\wh\bZ+\bDelta_{\mat{\Phi}} \bP}^2 +\FNorm{\bDelta_{\bA}(\mat{Z}-\wh{\mat{Z}})}^2 \right),
	\end{align*}
	we have
	\begin{equation}\label{delta:az}
	    \frac{1}{2T}	\FNorm{\bDelta_{\bA}\wh{\mat{Z}}+\bDelta_{\mat{\Phi}} \bP}^2 \geq \dfrac{1}{4T}\FNorm{\bDelta_{\bA}\mat{Z}+\bDelta_{\mat{\Phi}}\bP}^2 -\frac{1}{2T} \FNorm{\bDelta_{\bA}(\mat{Z}-\wh{\mat{Z}})}^2.
	\end{equation}
We consider the last term on the right-hand side of the above inequality,
	\begin{align*}
	\frac{1}{2T}\FNorm{\bDelta_{\bA}(\mat{Z}-\wh{\mat{Z}})}^2	\leq \frac{1}{2T} \FNorm{\bDelta_{\bA}}^2 \norm{\mat{Z}-\wh{\mat{Z}}}_2^2
	=\frac{1}{2T}\FNorm{\bDelta_{\bA}}^2 \norm{(\bB_c - \wh{\bB}_c)' \mat{X}}_2^2,
	\end{align*}
		where we used the inequality $\FNorm{\bM\bN}\leq \norm{\bM}_2 \FNorm{\bN}$.
	Notice that
	\[
	\FNorm{\mat{X}}^2 = \sum_{t=1}^T \tr(\vect{x}_t \vect{x}_t') = \sum_{t=1}^T (\vect{f}_t' \vect{f}_t + 2\vect{f}_t' \bB' \bm{\varepsilon}_t + \bm{\varepsilon}_t' \bm{\varepsilon}_t).
	\]
	
    By Assumptions \ref{assumption:sub_exponential} and \ref{assumption:dependence}, we have $\sum_{t=1}^T (\vect{f}_t' \vect{f}_t) = O_p(NT^2)$, $\sum_{t=1}^T (2\vect{f}_t' \bB' \bm{\varepsilon}_t) = O_p(NT)$, and $\sum_{t=1}^T (\bm{\varepsilon}_t' \bm{\varepsilon}_t)=O_p(N\sqrt{T})$. Then, by the results in Theorem \ref{loading_estimator},
\[
\frac{1}{2T} \FNorm{\bDelta_{\bA}(\mat{Z}-\wh{\mat{Z}})}^2	\leq\FNorm{\bDelta_{\bA}}^2\|\bB_c-\wh\bB_c\|_2^2\frac{1}{2T}\FNorm{\bX}^2 = \FNorm{\bDelta_{\bA}}^2 O_p\left(\dfrac{N}{T}\right) = o_p(\FNorm{\bDelta_{\bA}}^2),
\]
	under the assumption that $N/T\rightarrow 0$. It follows from (\ref{delta:az}) and the above rate that
	\[  \frac{1}{2T}	\FNorm{\bDelta_{\bA}\wh{\mat{Z}}+\bDelta_{\mat{\Phi}} \bP}^2 \geq \dfrac{1}{4T}\FNorm{\bDelta_{\bA}\mat{Z}+\bDelta_{\mat{\Phi}}\bP}^2 -o_p(\FNorm{\bDelta_{\bA}}^2).\]
	By the RSC condition specified in Theorem \ref{RRR_LASSO}, we have
	\begin{align*}
	    \frac{1}{2T}	\FNorm{\bDelta_{\bA}\wh{\mat{Z}}+\bDelta_{\mat{\Phi}} \bP}^2 &\geq \dfrac{1}{4T}\FNorm{\bDelta_{\bA}\mat{Z}+\bDelta_{\mat{\Phi}}\bP}^2 -o_p(\FNorm{\bDelta_{\bA}}^2)\\
	    &\geq C_1 \kappa_1 \FNorm{\bDelta}^2-C_2\tau_T\Psi^2(\bDelta)-o_p(\FNorm{\bDelta_{\bA}}^2)\\
	    &\geq C_1 \kappa_1 \FNorm{\bDelta}^2-C_2\tau_T\Psi^2(\bDelta),
	\end{align*}
where we assume $\kappa_1 >0$ in the last step.
This completes the proof.
\end{proof}


\begin{proof}[Proof of Theorem \ref{RRR_LASSO}]
Let 
\[
\mathit{Loss}_1(\bA,\mat{\Phi}) = \dfrac{1}{2T}\FNorm{\bY-\bA \wh\bZ - \bPhi\bP }^2+\lambda_{\bA}\norm{\bA}_*+\lambda_{\bPhi}\norm{\vectorize(\bPhi)}_1
\]
be the loss function defined in (\ref{solve_A_Phi}). Because the solutions $\wh\bA$ and $\wh\bPhi$ are obtained by solving the minimization problem 
\[
\wh\bA, \wh\bPhi = \arg\min_{\bA,\bPhi} \ \mathit{Loss}_1(\bA,\mat{\Phi}),
\]
implying that
\[
    \mathit{Loss}_1(\wh{\bA},\wh{\bPhi}) \leq \mathit{Loss}_1(\bA,\mat{\Phi}),
\]
where $\bA$ and $\bPhi$ are the corresponding true ones. Recall that $\bDelta_{\bA} = \wh{\bA}-\bA$ and $\bDelta_{\mat{\Phi}} = \wh{\mat{\Phi}}-\mat{\Phi}$, it follows from the above inequality that
\begin{equation}
\label{eq:basic}
\begin{split}
\dfrac{1}{2T} \FNorm{\bDelta_{\bA}\wh{\bZ}+\bDelta_{\bPhi} \bP + \bA (\wh\bZ-\bZ)}^2 \leq& \dfrac{1}{T}\inprod{\bE, \bDelta_{\bA} \wh{\bZ}+\bDelta_{\bPhi} \bP} + \dfrac{1}{2T}\FNorm{\bA(\wh\bZ-\bZ)}^2\\
&+ \lambda_{\bA}(\norm{\bA}_* - \norm{\bA+\bDelta_{\bA}}_*) \\
&+ \lambda_{\bPhi}(\norm{\vectorize(\bPhi)}_1 - \norm{\vectorize(\bPhi + \bDelta_{\bPhi})}_1)\\
=& \dfrac{1}{T}\inprod{\bE, \bDelta_{\bA} \bZ+\bDelta_{\bPhi} \bP} + \dfrac{1}{T} \inprod{\bE, \bDelta_{\bA} (\wh\bZ - \bZ)}\\
& +\dfrac{1}{2T}\FNorm{\bA (\wh\bZ - \bZ)}^2 + \lambda_{\bA}(\norm{\bA}_* - \norm{\bA+\bDelta_{\bA}}_*)\\
& + \lambda_{\bPhi}(\norm{\vectorize(\bPhi)}_1 - \norm{\vectorize(\bPhi + \bDelta_{\bPhi})}_1).
\end{split}
\end{equation}
Notice that the second term of the right-hand side of \eqref{eq:basic} satisfies
\begin{align*}
    \dfrac{1}{T} \inprod{\bE, \bDelta_{\bA} (\wh\bZ - \bZ)} \leq& \dfrac{1}{T} \norm{\bDelta_{\bA}}_* \norm{(\wh\bZ-\bZ) \bE'}_2\\
    \leq& \dfrac{1}{T} \norm{\bDelta_{\bA}}_* \norm{(\wh{\bB}_c-\bB_c)'\bB}_2 \Bnorm{\sum_{t=1}^T \bff_t \be_t'}_2\\
    &+ \dfrac{1}{T} \norm{\bDelta_{\bA}}_* \norm{\wh{\bB}_c-\bB_c}_2 \Bnorm{\sum_{t=1}^T \bve_t \be_t'}_2\\
    =& O_p\left(\sqrt{\dfrac{p}{T^3}} + \sqrt{\dfrac{pN}{T^3}}\right) \norm{\bDelta_{\bA}}_*\\
    =& o_p(1) \norm{\bDelta_{\bA}}_*,
\end{align*}
and under the assumption of $\norm{\bA}_2=O_p(1)$, the third term in the right-hand side of \eqref{eq:basic} satisfies
\begin{equation}
\label{eq:A(Zhat-Z)}
    \dfrac{1}{2T}\FNorm{\bA (\wh\bZ - \bZ)}^2 \leq \dfrac{1}{2T} \norm{\bA}_2^2 \norm{\wh{\bB}_c - \bB_c}_2^2 \FNorm{\bX}^2 = O_p\left(\dfrac{N}{T}\right) \norm{\bA}_2^2 = o_p(1).
\end{equation}
Therefore, {with probability tending to $1$}, \eqref{eq:basic} implies that
\begin{equation}
\label{eq:Holder}
\begin{split}
&\dfrac{1}{2T} \FNorm{\bDelta_{\bA}\wh{\bZ}+\bDelta_{\bPhi} \bP + \bA (\wh\bZ-\bZ)}^2\\
\leq& \dfrac{1}{T} \norm{\bDelta_{\bA}}_* \norm{\bE \bZ'}_2 + \dfrac{1}{T} \norm{\vectorize(\bDelta_{\bPhi})}_1 \norm{\vectorize\left(\bE \bP' \right) }_\infty + o_p(1) \norm{\bDelta_{\bA}}_* + o_p(1)\\
&+ \lambda_{\bA}(\norm{\bA}_* - \norm{\bA+\bDelta_{\bA}}_*) + \lambda_{\bPhi}(\norm{\vectorize(\bPhi)}_1 - \norm{\vectorize(\bPhi) + \bDelta_{\bPhi}}_1) \\
\leq& \dfrac{1}{2} \lambda_{\bA} \left(\norm{\bDelta_{\bA}}_* + 2\norm{\bA}_* - 2\norm{\bDelta_{\bA} +\bA}_*\right) \\
&+\dfrac{1}{2}\lambda_{\bPhi}\left(\norm{\vectorize(\bDelta_{\bPhi})}_1 + 2\norm{\vectorize(\bPhi)}_1 -2\norm{\vectorize(\bDelta_{\bPhi}+\bPhi)}_1\right),
\end{split}
\end{equation}
 where we use the condition that $\lambda_{\bA} \geq \dfrac{3}{T} \norm{\bE\mat{Z}'}_2$ and $\lambda_{\bPhi} \geq \dfrac{2}{T} \norm{\vectorize\left(\bE  \bP' \right) }_\infty$ in the second inequality.

Let $\bDelta_{\bA,2} = \Pi_{\mathcal{S}_\bA^\perp (r_{\bA})}(\bDelta_{\bA})$ be the projection of $\bDelta_{\bA}$ onto $\mathcal{S}_{\bA}^\perp (r_{\bA})$, where $r_{\bA}=\rank(\bA)$. Then we have the decomposition $\bDelta_{\bA} = \bDelta_{\bA,1}+\bDelta_{\bA,2}$, as discussed at the beginning of the Section \ref{ssec:prf_RRR_LASSO}. Similarly, we can decompose $\bA$ as $\bA = \bA_1+\bA_2$, where $\bA_2= \Pi_{\mathcal{S}_\bA^\perp (r_{\bA})}(\bA) = \mat{0}$. Then,
\begin{equation}
\label{eq:triangular}
\begin{split}
&\norm{\bDelta_{\bA}}_* + 2\norm{\bA}_* - 2\norm{\bDelta_{\bA} +\bA}_*\\
=& \norm{\bDelta_{\bA,1}+\bDelta_{\bA,2}}_* + 2\norm{\bA_1}_* - 2\norm{\bDelta_{\bA,1}+\bDelta_{\bA,2} +\bA_1}_*\\
\leq& \norm{\bDelta_{\bA,1}}_*+\norm{\bDelta_{\bA,2}}_* + 2\norm{\bA_1}_* - 2\norm{\bDelta_{\bA,2} +\bA_1}_* + 2\norm{\bDelta_{\bA,1}}_*\\
=& 3 \norm{\bDelta_{\bA,1}}_* - \norm{\bDelta_{\bA,2}}_*,
\end{split}
\end{equation}
where the last line comes from Lemma \ref{lemma:decomposition_of_nuclear_norm}.

Let $\mathcal{S}_{\mat{\Phi}}$ be the support set of $\mat{\Phi}$, that is, the set of the indexes of the nonzero elements in $\mat{\Phi}$, and $s_{\mat{\Phi}} = \operatorname{card}(\mathcal{S}_{\mat{\Phi}})$. For $\bDelta_{\mat{\Phi}}$, use the decomposition $\bDelta_{\mat{\Phi}} = \bDelta_{\mat{\Phi},1} + \bDelta_{\mat{\Phi},2}$, where the entries of $\bDelta_{\mat{\Phi},1}$ can only be non-zero in the positions in $\mathcal{S}_{\mat{\Phi}}$, and the entries of $\bDelta_{\mat{\Phi},2}$ can only be non-zero in the complement set of  $\mathcal{S}_{\mat{\Phi}}$. Similarly, we can decompose $\mat{\Phi}$ as $\mat{\Phi} = \mat{\Phi}_1 + \mat{\Phi}_2$, and it is not hard to see that $\mat{\Phi}_2 = {\bf 0}$.

By a similar argument as \eqref{eq:triangular}, we have
\begin{equation} \label{eq:LASSO_triangular}
\begin{split}
&\norm{\vectorize(\bDelta_{\mat{\Phi}})}_1 + 2\norm{\vectorize(\mat{\Phi})}_1 - 2\norm{\vectorize(\bDelta_{\mat{\Phi}} + \mat{\Phi})}_1\\
=& \norm{\vectorize(\bDelta_{\bPhi,1}) + \vectorize(\bDelta_{\bPhi,2})}_1 + 2\norm{\vectorize(\bPhi_1)}_1 - 2\norm{\vectorize(\bDelta_{\bPhi,1} + \bDelta_{\bPhi,2} + \bPhi_1)}_1\\
\leq& \norm{\vectorize(\bDelta_{\bPhi,1})}_1+\norm{\vectorize(\bDelta_{\bPhi,2})}_1 + 2\norm{\vectorize(\bPhi_1)}_1 - 2\norm{\vectorize(\bDelta_{\bPhi,2} + \bPhi_1)}_1\\
&+ 2\norm{\vectorize(\bDelta_{\bPhi,1})}_1\\
=& 3 \norm{\vectorize(\bDelta_{\bPhi,1})}_1 - \norm{\vectorize(\bDelta_{\bPhi,2})}_1.
\end{split}
\end{equation}
By \eqref{eq:triangular} and \eqref{eq:LASSO_triangular}, the right-hand side of \eqref{eq:Holder} can be upper bounded by
\begin{equation}\label{up:b}
    \dfrac{1}{2}\lambda_{\bA} \left( 3 \norm{\bDelta_{\bA,1}}_* - \norm{\bDelta_{\bA,2}}_* \right) + \dfrac{1}{2}\lambda_{\mat{\Phi}} \left( 
3 \norm{\vectorize(\bDelta_{\bPhi,1})}_1 - \norm{\vectorize(\bDelta_{\bPhi,2})}_1 \right),
\end{equation}
which also implies that
\[
\dfrac{1}{2}\lambda_{\bA} \left( 3 \norm{\bDelta_{\bA,1}}_* - \norm{\bDelta_{\bA,2}}_* \right) + \dfrac{1}{2}\lambda_{\mat{\Phi}} \left( 
3 \norm{\vectorize(\bDelta_{\bPhi,1})}_1 - \norm{\vectorize(\bDelta_{\bPhi,2})}_1 \right) \geq 0.
\]

Now we turn to the left-hand side of \eqref{eq:basic}. Notice that 
\begin{align*}
    \dfrac{1}{2T}\FNorm{\bDelta_{\bA} \wh\bZ + \bDelta_{\bPhi} \bP + \bA (\wh\bZ-\bZ)}^2 \geq \dfrac{1}{4T} \FNorm{\bDelta_{\bA} \wh\bZ +\bDelta_{\bPhi} \bP}^2 - \dfrac{1}{2T}\FNorm{\bA (\wh{\bZ}-\bZ)}^2,
\end{align*}
it follows from \ref{eq:A(Zhat-Z)} and Lemma~\ref{lemma:sample_RSC} that 
\begin{equation}
\label{RSC1}
\begin{split}
    \dfrac{1}{2T}\FNorm{\bDelta_{\bA} \wh\bZ + \bDelta_{\bPhi} \bP + \bA (\wh\bZ-\bZ)}^2 \geq& C_1 \kappa_1 (\FNorm{\bDelta_{\bA}}^2 + \FNorm{\bDelta_{\bPhi}}^2)\\
 &- C_2 (\lambda_{\bA} \norm{\bDelta_{\bA}}_* + \lambda_{\bPhi} \norm{\vectorize(\bDelta_{\bPhi})}_1)^2,
\end{split}
\end{equation}
where the last term on the right-hand side in the parentheses satisfies
\begin{align*}
    \lambda_{\bA} \norm{\bDelta_{\bA}}_{*} + \lambda_{\bPhi} \norm{\vectorize(\bDelta_{\bPhi})}_1 \leq& \lambda_{\bA}(\norm{\bDelta_{\bA,1}}_* + \norm{\bDelta_{\bA,2}}_*)\\
    &+ \lambda_{\bPhi} (\norm{\vectorize(\bDelta_{\bPhi,1})}_1 + \norm{\vectorize(\bDelta_{\bPhi,2})}_1)\\
    \leq& 4(\lambda_{\bA} \norm{\bDelta_{\bA,1}}_{*} + \lambda_{\bPhi} \norm{\vectorize(\bDelta_{\bPhi,1})}_1)\\
    \leq& 4(\lambda_{\bA} \sqrt{2r_{\bA}} \FNorm{\bDelta_{\bA,1}} + \lambda_{\bPhi} \sqrt{s_{\bPhi}} \FNorm{\bDelta_{\bPhi,1}})\\
   \leq& 4(\lambda_{\bA} \sqrt{2r_{\bA}} \FNorm{\bDelta_{\bA}} + \lambda_{\bPhi}\sqrt{s_{\bPhi}}\FNorm{\bDelta_{\bPhi}}),
\end{align*}
where we use the inequality rank$(\bDelta_{\bA,1})\leq 2r_{\bA}$.
Under the assumptions that $\kappa_1 \geq C_0 \lambda_{\bA}^2 r_{\bA} \tau_T$ and $\kappa_1 \geq C_0 \lambda_{\bPhi}^2 s_{\bPhi} \tau_T$ for some $C_0>0$, it follows that
\[
\tau_T \Psi^2(\bDelta) \leq \dfrac{1}{4}C_0 \tau_T (r_{\bA} \lambda_{\bA}^2 \FNorm{\bDelta_{\bA}}^2 + s_{\bPhi} \lambda_{\bPhi}^2 \FNorm{\bDelta_{\bPhi}}^2)\leq \dfrac{1}{4}\kappa_1 (\FNorm{\bDelta_{\bA}}^2 + \FNorm{\bDelta_{\bPhi}}^2).
\]
Therefore, (\ref{RSC1}) implies that
\[
\dfrac{1}{2T} \FNorm{\bDelta_{\bA} \bZ + \bDelta_{\bPhi}\bP}^2 \geq \dfrac{1}{4} \kappa_1 (\FNorm{\bDelta_{\bA}}^2 + \FNorm{\bDelta_{\bPhi}}^2).
\]
By (\ref{up:b}) and the above inequality,
\begin{align*}
    \dfrac{1}{4} \kappa_1 (\FNorm{\bDelta_{\bA}}^2 + \FNorm{\bDelta_{\bPhi}}^2)\leq& \dfrac{1}{2} \lambda_{\bA} \left( 3 \norm{\bDelta_{\bA,1}}_* - \norm{\bDelta_{\bA,2}}_* \right)\\
    &+ \dfrac{1}{2}\lambda_{\mat{\Phi}} \left( 
3 \norm{\vectorize(\bDelta_{\bPhi,1})}_1 - \norm{\vectorize(\bDelta_{\bPhi,2})}_1 \right)\\
\leq& \dfrac{3}{2} (\lambda_{\bA} \sqrt{2r_{\bA}} \FNorm{\bDelta_{\bA}} + \lambda_{\bPhi} \sqrt{s_{\bPhi}} \FNorm{\bDelta_{\bPhi}})\\
\leq& \dfrac{3}{2} \sqrt{ 2r_{\bA} \lambda_{\bA}^2  + \lambda_{\bPhi}^2 s_{\bPhi} } \sqrt{\FNorm{\bDelta_{\bA}}^2+\FNorm{\bDelta_{\bPhi}}^2}.
\end{align*}
Dividing both sides by $\sqrt{\FNorm{\bDelta_{\bA}}^2+\FNorm{\bDelta_{\bPhi}}^2}$, we  obtain
\begin{align*}
    \FNorm{\bDelta_{\bA}}^2 + \FNorm{\bDelta_{\bPhi}}^2\leq C(r_{\bA} \lambda_{\bA}^2  + s_{\bPhi} \lambda_{\bPhi}^2)/\kappa_1 ^2. 
\end{align*}
This completes the proof.
\end{proof}

\section{Proof of Theorem~\ref{tm4}}

We now provide a proof of Theorem~\ref{tm4}, which is similar to the proof of Theorem~\ref{RRR_LASSO} but with adjustments for different regularizations and different conditions. 

\begin{proof}[Proof of Theorem~\ref{tm4}]
    Let
    \[
\mathit{Loss}_2(\bA,\bPhi) = \dfrac{1}{2T} \FNorm{\bY - \bA\wh\bZ - \bPhi \bP}^2 + \lambda_{\bA}\norm{\bA}_* + \sum_{i=1}^d \lambda_{i} \norm{\bPhi_i}_*.
    \]
Then following $\mathit{Loss}_2(\wh\bA,\wh\bPhi) \leq \mathit{Loss}_2(\bA,\bPhi)$, we can obtain a similar result as \eqref{eq:basic} and \eqref{eq:Holder}, that is, {with probability tending to $1$},
\begin{equation}
\label{eq:IRRA_basic}
\begin{split}
\dfrac{1}{2T} \FNorm{\bDelta_{\bA}\wh{\bZ}+\bDelta_{\bPhi} \bP+ \bA (\wh{\bZ}-\bZ)}^2
\leq& \dfrac{1}{T} \inprod{\bE, \bDelta_{\bA} \bZ+\bDelta_{\bPhi} \bP} + \dfrac{1}{T}\inprod{\bE,\bDelta_{\bA} (\wh\bZ-\bZ)}\\
& + \dfrac{1}{2T}\FNorm{\bA(\wh\bZ-\bZ)}^2 +\lambda_{\bA}(\norm{\bA}_* - \norm{\bA+\bDelta_{\bA}}_*)\\
&+ \sum_{i=1}^{d} \lambda_{i}(\norm{\bDelta_{\bPhi_i}}_* - \norm{\bPhi_i + \bDelta_{\bPhi_i}}_*)\\
\leq& \dfrac{1}{T} \norm{\bDelta_{\bA}}_* \norm{\bE\bZ'}_2 + \dfrac{1}{T} \sum_{i=1}^{d} \norm{\bDelta_{\bPhi_i}}_* \norm{\bE L^i(\bY)'}_2\\
&+ o_p(1)\norm{\bDelta_{\bA}}_* + o_p(1) + \lambda_{\bA}(\norm{\bA}_* - \norm{\bA+\bDelta_{\bA}}_*)\\
&+ \sum_{i=1}^{d} \lambda_{i}(\norm{\bDelta_{\bPhi_i}}_* - \norm{\bPhi_i + \bDelta_{\bPhi_i}}_*)\\
\leq& \dfrac{1}{2} \lambda_{\bA} (\norm{\bDelta_{\bA}}_* + 2\norm{\bA}_* - 2\norm{\bDelta_{\bA} +\bA}_*)\\
&+ \dfrac{1}{2} \sum_{i=1}^{d} \lambda_{i}(\norm{\bDelta_{\bPhi_i}}_* + 2\norm{\bPhi_i}_* - 2\norm{\bDelta_{\bPhi_i} + \bPhi_i}_*),
\end{split}
\end{equation}
 where we use the condition $\lambda_{\bA}\geq \dfrac{3}{T}\norm{\bE\bZ'}_2$ and $\lambda_i \geq \dfrac{2}{T} \norm{\bE L^i(\bY)'}_2$ in the last inequality.

Again, we decompose $\bDelta_{\bA}$ as that in the proof of Theorem~\ref{RRR_LASSO}, and decompose $\bDelta_{\bPhi_i}$ as $\bDelta_{\bPhi_i} = \bDelta_{\bPhi_i,1} + \bDelta_{\bPhi_i,2}$, $i=1,\ldots,d$, where $\bDelta_{\bPhi_i,2} = \Pi_{\mathcal{S}_{\bPhi_i}^\perp(r_{\bPhi_i})}(\bDelta_{\bPhi_i})$. Then, by a similar argument as \eqref{eq:triangular}, the right-hand side of the \eqref{eq:IRRA_basic} has an upper bound
\begin{equation}
\label{eq:RHS}
\dfrac{1}{2} \lambda_{\bA} (3\norm{\bDelta_{\bA,1}}_* - \norm{\bDelta_{\bA,2}}_*) + \dfrac{1}{2} \sum_{i=1}^{d} \lambda_{i}(3\norm{\bDelta_{\bPhi_i,1}}_*- \norm{\bDelta_{\bPhi_i,2}}_*),
\end{equation}
which also implies that $\bDelta$ is in the restricted set $\mathcal{C}$ defined in \eqref{eq:set}. Furthermore, by the RE condition defined in Assumption~\ref{asm:RE}, the left-hand side of \eqref{eq:IRRA_basic} has a lower bound, that is,
\begin{align*}
\dfrac{1}{2T}\FNorm{\bDelta_{\bA} \wh\bZ + \bDelta_{\bPhi} \bP + \bA (\wh\bZ-\bZ)}^2 \geq& \dfrac{1}{4T} \FNorm{\bDelta_{\bA} \wh\bZ +\bDelta_{\bPhi} \bP}^2 - \dfrac{1}{2T}\FNorm{\bA (\wh{\bZ}-\bZ)}^2\\
\geq& \dfrac{1}{8T} \FNorm{\bDelta_{\bA} \bZ +\bDelta_{\bPhi} \bP}^2 -\dfrac{1}{4T} \FNorm{\bDelta_{\bA}(\wh\bZ-\bZ)}^2\\
&- \dfrac{1}{2T} \FNorm{\bA (\wh{\bZ}-\bZ)}^2\\
\geq& \dfrac{1}{8T} \FNorm{\bDelta_{\bA} \bZ +\bDelta_{\bPhi} \bP}^2 - o_p(1)\\
\geq& C \kappa_2 (\FNorm{\bDelta_{\bA}}^2 + \sum_{i=1}^{d} \FNorm{\bDelta_{\bPhi_i}}^2).
\end{align*}
Then by \eqref{eq:RHS} and the above inequality,
\begin{align*}
C \kappa_2 \left(\FNorm{\bDelta_{\bA}}^2 + \sum_{i=1}^{d} \FNorm{\bDelta_{\bPhi_i}}^2 \right) \leq& \dfrac{1}{2} \lambda_{\bA} (3\norm{\bDelta_{\bA,1}}_* - \norm{\bDelta_{\bA,2}}_*)\\
&+ \dfrac{1}{2} \sum_{i=1}^{d} \lambda_{i}(3\norm{\bDelta_{\bPhi_i,1}}_*- \norm{\bDelta_{\bPhi_i,2}}_*)\\
\leq& \dfrac{3}{2} \left(\lambda_{\bA}\sqrt{2r_{\bA}}\FNorm{\bDelta_{\bA}} + \sum_{i=1}^{d} \lambda_{i} \sqrt{2 r_{\bPhi,i}}\FNorm{\bDelta_{\bPhi_i}}\right)\\
\leq& \dfrac{3}{2} \sqrt{2 r_{\bA}\lambda_{\bA}^2 + 2 \sum_{i=1}^{d} r_{\bPhi_i} \lambda_{i}^2} \sqrt{\FNorm{\bDelta_{\bA}}^2 + \sum_{i=1}^{d}\FNorm{\bDelta_{\bPhi_i}}^2}
.
\end{align*}
Therefore, 
\[
\FNorm{\bDelta_{\bA}}^2 + \sum_{i=1}^{d}\FNorm{\bDelta_{\bPhi_i}}^2 \leq C\left(r_{\bA}\lambda_{\bA}^2 + \sum_{i=1}^{d} r_{\bPhi_i} \lambda_{i}^2\right) / \kappa_2^2.
\]
This completes the proof.
\end{proof}

			\printbibliography[heading=bibliography]

@article{agarwal2012noisy,
	title={Noisy matrix decomposition via convex relaxation: Optimal rates in high dimensions},
	author={Agarwal, Alekh and Negahban, Sahand and Wainwright, Martin J},
	journal={The Annals of Statistics},
	volume={40},
	number={2},
	pages={1171--1197},
	year={2012},
	publisher={Institute of Mathematical Statistics}
}

@article{aznar2002selecting,
  title={Selecting the rank of the cointegration space and the form of the intercept using an information criterion},
  author={Aznar, Antonio and Salvador, Manuel},
  journal={Econometric Theory},
  volume={18},
  number={4},
  pages={926--947},
  year={2002},
  publisher={Cambridge University Press}
}

@article{bai2002determining,
	title={Determining the number of factors in approximate factor models},
	author={Bai, Jushan and Ng, Serena},
	journal={Econometrica},
	volume={70},
	number={1},
	pages={191--221},
	year={2002},
	publisher={Wiley Online Library}
}

@article{bai2004estimating,
  title={Estimating cross-section common stochastic trends in nonstationary panel data},
  author={Bai, Jushan},
  journal={Journal of Econometrics},
  volume={122},
  number={1},
  pages={137--183},
  year={2004},
  publisher={Elsevier}
}

@article{banerjee2014forecasting,
  title={Forecasting with factor-augmented error correction models},
  author={Banerjee, Anindya and Marcellino, Massimiliano and Masten, Igor},
  journal={International Journal of Forecasting},
  volume={30},
  number={3},
  pages={589--612},
  year={2014},
  publisher={Elsevier}
}

@book{billingsley1999convergence,
  title={Convergence of probability measures},
  author={Billingsley, Patrick},
  year={1999},
  publisher={John Wiley \& Sons}
}

@book{boyd2004convex,
  title={Convex optimization},
  author={Boyd, Stephen and Vandenberghe, Lieven},
  year={2004},
  publisher={Cambridge university press}
}

@article{boyd2011distributed,
  title={Distributed optimization and statistical learning via the alternating direction method of multipliers},
  author={Boyd, Stephen and Parikh, Neal and Chu, Eric and Peleato, Borja and Eckstein, Jonathan and others},
  journal={Foundations and Trends in Machine learning},
  volume={3},
  number={1},
  pages={1--122},
  year={2011},
  publisher={Now Publishers, Inc.}
}

@article{burai2013necessary,
  title={Necessary and sufficient condition on global optimality without convexity and second order differentiability},
  author={Burai, P{\'a}l},
  journal={Optimization Letters},
  volume={7},
  number={5},
  pages={903--911},
  year={2013},
  publisher={Springer}
}

@article{chen2013reduced,
  title={Reduced rank regression via adaptive nuclear norm penalization},
  author={Chen, Kun and Dong, Hongbo and Chan, Kung-Sik},
  journal={Biometrika},
  volume={100},
  number={4},
  pages={901--920},
  year={2013},
  publisher={Oxford University Press}
}

@article{engle1987co,
  title={Co-integration and error correction: representation, estimation, and testing},
  author={Engle, Robert F and Granger, Clive WJ},
  journal={Econometrica},
  pages={251--276},
  year={1987},
  publisher={JSTOR}
}

@article{fan2013large,
  title={Large covariance estimation by thresholding principal orthogonal complements},
  author={Fan, Jianqing and Liao, Yuan and Mincheva, Martina},
  journal={Journal of the Royal Statistical Society: Series B (Statistical Methodology)},
  volume={75},
  number={4},
  pages={603--680},
  year={2013},
  publisher={Wiley Online Library}
}

@article{forni2005generalized,
  title={The generalized dynamic factor model: one-sided estimation and forecasting},
  author={Forni, Mario and Hallin, Marc and Lippi, Marco and Reichlin, Lucrezia},
  journal={Journal of the American Statistical Association},
  volume={100},
  number={471},
  pages={830--840},
  year={2005},
  publisher={Taylor \& Francis}
}

@article{gao2019banded,
  title={Banded spatio-temporal autoregressions},
  author={Gao, Zhaoxing and Ma, Yingying and Wang, Hansheng and Yao, Qiwei},
  journal={Journal of Econometrics},
  volume={208},
  number={1},
  pages={211--230},
  year={2019},
  publisher={Elsevier}
}

@article{gao2019structural,
  title={A structural-factor approach to modeling high-dimensional time series and space-time data},
  author={Gao, Zhaoxing and Tsay, Ruey S},
  journal={Journal of Time Series Analysis},
  volume={40},
  number={3},
  pages={343--362},
  year={2019},
  publisher={Wiley Online Library}
}

@article{gao2021modeling_JASA,
  title={Modeling high-dimensional time series: A factor model with dynamically dependent factors and diverging eigenvalues},
  author={Gao, Zhaoxing and Tsay, Ruey S},
  journal={Journal of the American Statistical Association},
  pages={1--17},
  year={2021},
  publisher={Taylor \& Francis}
}

@article{gao2021modeling_IJF,
  title={Modeling high-dimensional unit-root time series},
  author={Gao, Zhaoxing and Tsay, Ruey S},
  journal={International Journal of Forecasting},
  volume={37},
  number={4},
  pages={1535--1555},
  year={2021},
  publisher={Elsevier}
}

@article{gao2021two,
  title={A two-way transformed factor model for matrix-variate time series},
  author={Gao, Zhaoxing and Tsay, Ruey S},
  journal={Econometrics and Statistics},
  year={2021},
  publisher={Elsevier}
}

@article{gao2022divide,
  title={Divide-and-conquer: a distributed hierarchical factor approach to modeling large-scale time series data},
  author={Gao, Zhaoxing and Tsay, Ruey S},
  journal={Journal of the American Statistical Association},
  %number={forthcoming},
  pages={forthcoming},
  year={2022},
  publisher={Taylor \& Francis}
}

@book{golub2013matrix,
	title={Matrix computations},
	author={Golub, Gene H and Van Loan, Charles F},
	year={2013},
	publisher={JHU press}
}

@article{gu2020empirical,
  title={Empirical asset pricing via machine learning},
  author={Gu, Shihao and Kelly, Bryan and Xiu, Dacheng},
  journal={The Review of Financial Studies},
  volume={33},
  number={5},
  pages={2223--2273},
  year={2020},
  publisher={Oxford University Press}
}

@book{hastie2015statistical,
  title={Statistical Learning with Sparsity: The Lasso and Generalizations},
  author={Hastie, Trevor and Tibshirani, Robert and Wainwright, Martin J},
  year={2015},
  publisher={CRC Press}
}

@inproceedings{ji2009accelerated,
  title={An accelerated gradient method for trace norm minimization},
  author={Ji, Shuiwang and Ye, Jieping},
  booktitle={Proceedings of the 26th Annual International Conference on Machine Learning},
  pages={457--464},
  year={2009}
}

@article{johansen1988statistical,
  title={Statistical analysis of cointegration vectors},
  author={Johansen, S{\o}ren},
  journal={Journal of Economic Dynamics and Control},
  volume={12},
  number={2-3},
  pages={231--254},
  year={1988},
  publisher={Elsevier}
}

@article{johansen1991estimation,
  title={Estimation and hypothesis testing of cointegration vectors in Gaussian vector autoregressive models},
  author={Johansen, S{\o}ren},
  journal={Econometrica},
  pages={1551--1580},
  year={1991},
  publisher={JSTOR}
}

@article{johansen2002small,
  title={A small sample correction for the test of cointegrating rank in the vector autoregressive model},
  author={Johansen, S{\o}ren},
  journal={Econometrica},
  volume={70},
  number={5},
  pages={1929--1961},
  year={2002},
  publisher={Wiley Online Library}
}

@article{johnstone2009consistency,
  title={On consistency and sparsity for principal components analysis in high dimensions},
  author={Johnstone, Iain M and Lu, Arthur Yu},
  journal={Journal of the American Statistical Association},
  volume={104},
  number={486},
  pages={682--693},
  year={2009},
  publisher={Taylor \& Francis}
}

@article{koo2020high,
  title={High-dimensional predictive regression in the presence of cointegration},
  author={Koo, Bonsoo and Anderson, Heather M and Seo, Myung Hwan and Yao, Wenying},
  journal={Journal of Econometrics},
  volume={219},
  number={2},
  pages={456--477},
  year={2020},
  publisher={Elsevier}
}

@article{lam2011estimation,
  title={Estimation of latent factors for high-dimensional time series},
  author={Lam, Clifford and Yao, Qiwei and Bathia, Neil},
  journal={Biometrika},
  volume={98},
  number={4},
  pages={901--918},
  year={2011},
  publisher={Oxford University Press}
}

@article{lam2012factor,
  title={Factor modeling for high-dimensional time series: inference for the number of factors},
  author={Lam, Clifford and Yao, Qiwei},
  journal={The Annals of Statistics},
  pages={694--726},
  year={2012},
  publisher={JSTOR}
}

@article{li2019integrative,
  title={Integrative multi-view regression: Bridging group-sparse and low-rank models},
  author={Li, Gen and Liu, Xiaokang and Chen, Kun},
  journal={Biometrics},
  volume={75},
  number={2},
  pages={593--602},
  year={2019},
  publisher={Wiley Online Library}
}

@article{lin2017regularized,
  title={Regularized estimation and testing for high-dimensional multi-block vector-autoregressive models},
  author={Lin, Jiahe and Michailidis, George},
  journal={Journal of Machine Learning Research},
  volume={18},
  year={2017}
}

@book{lutkepohl2006new,
  title={New introduction to multiple time series analysis},
  author={L{\"u}tkepohl, Helmut},
  year={2006},
  publisher={Springer Science \& Business Media}
}

@article{merlevede2011bernstein,
  title={A Bernstein type inequality and moderate deviations for weakly dependent sequences},
  author={Merlev{\`e}de, Florence and Peligrad, Magda and Rio, Emmanuel},
  journal={Probability Theory and Related Fields},
  volume={151},
  number={3},
  pages={435--474},
  year={2011},
  publisher={Springer}
}

@article{negahban2011estimation,
  title={Estimation of (near) low-rank matrices with noise and high-dimensional scaling},
  author={Negahban, Sahand and Wainwright, Martin J},
  journal={The Annals of Statistics},
  volume={39},
  number={2},
  pages={1069--1097},
  year={2011},
  publisher={Institute of Mathematical Statistics}
}

@article{pan2008modelling,
  title={Modelling multiple time series via common factors},
  author={Pan, Jiazhu and Yao, Qiwei},
  journal={Biometrika},
  volume={95},
  number={2},
  pages={365--379},
  year={2008},
  publisher={Oxford University Press}
}

@article{pena2006nonstationary,
	title={Nonstationary dynamic factor analysis},
	author={Pe{\~n}a, Daniel and Poncela, Pilar},
	journal={Journal of Statistical Planning and Inference},
	volume={136},
	number={4},
	pages={1237--1257},
	year={2006},
	publisher={Elsevier}
}

@book{reinsel2022multivariate,
  title={Multivariate reduced-rank regression: theory and applications (2nd ed.).},
  author={Reinsel, Gregory C and Velu, Raja P and Chen, Kun},
  year={2022+},
  publisher={Springer}
}

@article{saikkonen2000testing,
  title={Testing for the cointegrating rank of a VAR process with structural shifts},
  author={Saikkonen, Pentti and L{\"u}tkepohl, Helmut},
  journal={Journal of Business \& Economic Statistics},
  volume={18},
  number={4},
  pages={451--464},
  year={2000},
  publisher={Taylor \& Francis}
}

@article{stock1987asymptotic,
  title={Asymptotic properties of least squares estimators of cointegrating vectors},
  author={Stock, James H},
  journal={Econometrica},
  pages={1035--1056},
  year={1987},
  publisher={JSTOR}
}

@misc{stock2005implications,
  title={Implications of dynamic factor models for VAR analysis},
  author={Stock, James H and Watson, Mark W},
  year={2005},
  publisher={National Bureau of Economic Research Cambridge, Mass., USA}
}

@article{tiao1989model,
  title={Model specification in multivariate time series (with discussion)},
  author={Tiao, George C and Tsay, Ruey S},
  journal={Journal of the Royal Statistical Society: Series B (Methodological)},
  volume={51},
  number={2},
  pages={157--195},
  year={1989},
  publisher={Wiley Online Library}
}

@book{tsay2014multivariate,
  title={Multivariate time series analysis: with {R} and financial applications},
  author={Tsay, Ruey S},
  year={2014},
  publisher={John Wiley \& Sons}
}

@article{tseng2001convergence,
  title={Convergence of a block coordinate descent method for nondifferentiable minimization},
  author={Tseng, Paul.},
  journal={Journal of Optimization Theory and Applications},
  volume={109},
  number={3},
  pages={475--494},
  year={2001},
  publisher={Springer}
}

@article{vandenberghe1996semidefinite,
  title={Semidefinite programming},
  author={Vandenberghe, Lieven and Boyd, Stephen},
  journal={SIAM Review},
  volume={38},
  number={1},
  pages={49--95},
  year={1996},
  publisher={SIAM}
}

@book{wainwright2019high,
  title={High-dimensional statistics: A non-asymptotic viewpoint},
  author={Wainwright, Martin J},
  year={2019},
  publisher={Cambridge University Press}
}

@article{welch2008comprehensive,
  title={A comprehensive look at the empirical performance of equity premium prediction},
  author={Welch, Ivo and Goyal, Amit},
  journal={The Review of Financial Studies},
  volume={21},
  number={4},
  pages={1455--1508},
  year={2008},
  publisher={Society for Financial Studies}
}

@article{zhang2019identifying,
  title={Identifying cointegration by eigenanalysis},
  author={Zhang, Rongmao and Robinson, Peter and Yao, Qiwei},
  journal={Journal of the American Statistical Association},
  volume={114},
  number={526},
  pages={916--927},
  year={2019},
  publisher={Taylor \& Francis}
}
		\end{appendices}

	\end{spacing}
	
	\newpage


\end{document}